\documentclass[review]{elsarticle}

\usepackage{lineno,hyperref}
\usepackage{amssymb}
\usepackage{color}
\usepackage{amsmath}
\usepackage{tikz}
\usepackage{amsfonts}
\usepackage{mathrsfs}
\usepackage{amsthm}
\usepackage{float}
\usepackage{bbding}
\usepackage{mathrsfs}
\usepackage{amsmath, amsfonts, amssymb, mathrsfs, txfonts}
\usepackage{graphicx, subfigure}
\usepackage{wasysym}
\usepackage{color}
\usepackage{bm}
\usepackage{enumerate}

\usepackage{pifont}
\usepackage{txfonts}
\usepackage{amsmath}
\usepackage{graphicx}
\usepackage{amsfonts}
\usepackage{amssymb}
\usepackage{mathrsfs,psfrag,eepic,epsfig}
\usepackage{epstopdf}

\biboptions{numbers,sort&compress}
\newtheorem{thm}{Theorem}[section]

\theoremstyle{definition}

\modulolinenumbers[5]

\journal{Journal of \LaTeX\ Templates}

\makeatletter \@addtoreset{equation}{section}

\begin{document}

\begin{frontmatter}
\title{A new matrix
modified Korteweg-de Vries equation: Riemann-Hilbert approach and exact solutions}
\tnotetext[mytitlenote]{Project supported by the Fundamental Research Fund for the Central Universities under the grant No. 2019ZDPY07.\\
\hspace*{3ex}$^{*}$Corresponding author.\\
\hspace*{3ex}\emph{E-mail addresses}:  sftian@cumt.edu.cn and
shoufu2006@126.com (S. F. Tian)}

\author{Wei-Kang Xun, Shou-Fu Tian$^{*}$ and  Jin-Jie Yang}
\address{
School of Mathematics and Institute of Mathematical Physics, China University of Mining and Technology, Xuzhou 221116, People's Republic of China
 }

\begin{abstract}
      A new matrix modified  Korteweg-de Vries (mmKdV)   equation with a $p\times q$ complex-valued  potential matrix function  is first studied via Riemann-Hilbert approach, which can be reduced to the well-known coupled modified Korteweg-de Vries   equations by selecting  special potential  matrix. Starting from the special analysis for the Lax pair of this equation,  we successfully establish a Riemann-Hilbert problem of the equation.  By introducing the special conditions  of irregularity and reflectionless case,  some interesting exact solutions, including the $N$-soliton solution formula, of the mmKdV   equation are derived  through  solving the corresponding Riemann-Hilbert problem. Moreover, due to the special symmetry of special potential matrices and the $N$-soliton solution formula, we  make further efforts to classify the original exact solutions to obtain some other interesting solutions which are all displayed graphically.
      It is interesting that  the local structures and dynamic behaviors of soliton solutions, breather-type   solutions and bell-type soliton solutions are all analyzed  via taking different types of  potential matrices.
\end{abstract}

\begin{keyword}
Matrix modified Korteweg-de  Vries  equation \sep   Riemann-Hilbert  approach \sep Exact solutions \sep Multi-soliton solutions \sep   Soliton classification.
\end{keyword}

\end{frontmatter}


\newpage
\tableofcontents
\newpage
\section{Introduction}
     Soliton theory is one of the important research directions in the field of nonlinear science in the world today.  Solitons of nonlinear differential equations play an important role in revealing some important physical laws such as fluid mechanics, plasma physics, nonlinear optics, classical field theory, and quantum theory. Based on this fact, many domestic and foreign scholars have devoted themselves to the study of solitons. In recent decades, a number of  effective methods  have been produced  such as Hirota bilinear method \cite{hirota1980direct},  Darboux and B\"{a}cklund transformation  \cite{matveev1979darboux},  inverse scattering transformation \cite{ablowitz1981solitons,beals1984scattering,beals1988direct}.  Apart from them, the Riemann-Hilbert(RH) approach is also  a very effective method which  can not only  slove the soliton solutions of a series of nonlinear evolution equations \cite{fokas2012unified,lenells2012initial,guo2012riemann,de2013riemann,geng2016riemann,yan2017initial,
     ma2018riemann,ma2019inverse,wang2010integrable,zhang2017riemann,tian2017initial,tian2018initial,tian2016mixed,
     xia2018initial,peng2019riemann,yang2019n}, but also study  integrable systems with non-zero boundaries \cite{ablowitz2018inverse,vekslerchik1992discrete,biondini2014inverse,prinari2015inverse,yang2019riemann}, the asymptoticity of integrable system solutions \cite{deift1992steepest,xu2015long,tian2018long,wang2019long,liu2019long},   certain improtant properties of orthogonal polynomials\cite{deift1999orthogonal}, etc.
      Because of the superiority of RH approach, many scholars have done a lot of work about the exact solutions of single equation and
partial coupled equations with it. However, there are very few work about the solutions of matrix-type equation via the RH approach. Therefore, the main purpose of our work is to  study the RH problem and exact solutions with their  properties of a new matrix modified Korteweg-de Vries(mmKdV) equation in this work.

     In this work, we focus on  Riemann-Hilbert problem and exact solutions with their propagation  behaviors for the mmKdV  equation \cite{tsuchida1998coupled}
     \begin{equation} \label{mmkdv}
         Q_t + Q_{xxx} - 3 \epsilon (Q_x Q^{\dagger} Q + Q Q^{\dagger} Q_x ) = 0,  \qquad  \epsilon = \pm 1,
     \end{equation}
      where $Q$ is a $p\times q$ complex-valued matrix function of variation $x$ and  $t$. If $Q$ is restricted to be a real matrix, Eq.  \eqref{mmkdv} is the same as the equation presented in \cite{athorne1987generalised}. When $Q$ is taken as some special forms, Eq. \eqref{mmkdv} can be reduced to the coupled modified  Korteweg-de Vries(cmKdV) equations
      \begin{equation} \label{cmkdv}
           \frac{\partial{v_i}}{\partial t}-6 \left( \sum_{j=0}^{M-1} \epsilon_{j}v_j^2 \right) \frac{\partial{v_i}}{\partial x} +\frac{\partial^3 v_i}{\partial x^3} = 0, \quad \epsilon_j=\pm 1, \quad i=0,1,\dots,M-1.
      \end{equation}
      Moreover, as far as known, there are already two different methods to complete the reduction to cmKdV in \cite{tsuchida1998coupled,zhang2008lax}. Here is a brief introduction to the reduced method in  \cite{tsuchida1998coupled}.
      We define
      \begin{equation}
            Q^{(1)}=\mu_0 v_0 +i v_1, \quad R^{(1)}=\epsilon_1 ( \mu_0 v_0 - i v_1),
      \end{equation}
      \begin{equation}
       Q^{(m+1)}= \begin{bmatrix}  Q^{(m)} &  -\epsilon_{2m+1}(\mu_{2m} v_{2m}+i v_{2m})I_{2^{m-1}} \\  -(\mu_{2m} v_{2m} - i v_{2m})I_{2^{m-1}}  & -R^{(m)}  \end{bmatrix},
      \end{equation}
      \begin{equation}
       R^{(m+1)}= \begin{bmatrix}  R^{(m)} &  -\epsilon_{2m+1}(\mu_{2m} v_{2m}+i v_{2m})I_{2^{m-1}} \\  -(\mu_{2m} v_{2m} - i v_{2m})I_{2^{m-1}}  & -Q^{(m)}  \end{bmatrix},
      \end{equation}
      where $I_{2^{m-1}}$ is the $2^{m-1}\times 2^{m-1}$ identity matrix, $\epsilon=\pm 1$, and $\mu_{2m}$ satisfies
      \begin{equation}
            \mu_{2m}^2 = \frac{\epsilon_{2m}}{\epsilon_{2m+1}}=\epsilon_{2m}\epsilon_{2m+1}.
      \end{equation}
      Substituting $Q^{(m)}$  and $R^{(m)}$ for $Q$ and $R$ into Eq. \eqref{mmkdv}, we can obtain Eq. \eqref{cmkdv} $(M=2m)$.

      In addition, the conservation laws  and Hamilton structure of the mmKdV equation  have been studied carefully  by  Tsuchida  and   Wadati in  \cite{tsuchida1998coupled}. Starting from a special class $p=q=n$, the Lax representation for Eq. \eqref{cmkdv}  can be written as
      \begin{equation}  \label{laxs}
         \begin{aligned}
            \begin{bmatrix} \Phi_1 \\  \Phi_2  \end{bmatrix}_x & = \begin{bmatrix} -i \zeta I & Q \\ R & i \zeta I \end{bmatrix}   \begin{bmatrix} \Phi_1 \\  \Phi_2  \end{bmatrix},  \\
                        \begin{bmatrix} \Phi_1 \\  \Phi_2  \end{bmatrix}_t & = \begin{bmatrix} V_{11} &    V_{12} \\ V_{21} &    V_{22}  \end{bmatrix}   \begin{bmatrix} \Phi_1 \\  \Phi_2  \end{bmatrix},
         \end{aligned}
      \end{equation}
      where every element is a $n\times n$ matrix. We define $\Gamma=\Phi_2 \Phi_1^{-1}$, with the aid of  Lax pair and  the compatibility condition of it,  some special relationship can be derived by
      \begin{equation}
        \begin{aligned}
         \left\{ tr(Q \Gamma) \right\}_t &= \left\{ \mbox{some   function    of  } \Gamma, \mbox{Q, R   and } \zeta  \right\}_x,    \\
          2i Q \Gamma  & = - QR + Q(Q^{-1} Q \Gamma)_x + (Q \Gamma)^2,
        \end{aligned}
      \end{equation}
      where $tr(Q \Gamma)$ is the generating form of conserved densities. By expanding $Q \Gamma$  with respect to the spectral parameter $\zeta$ as   $Q \Gamma = \sum_{l=1}^{\infty}\frac{1}{(2i \zeta)^l}F_l$.  A recursion formula can be obtained as follows
      \begin{equation}
         F_{l+1}= -\delta_{l,0} Q R +Q (Q^{-1}F_l)_x + \sum_{k=1}^{l=1} F_k F_{l-k}, \quad l=0,1,\dots,
      \end{equation}
      where   $tr(F_l)$ is a conserved  density for any positive integer $l$. Moreover, on the basis of conserved density, the Hamilton structure and the Possion bracket  of  mmKdV equation can be otained  as follows
      \begin{equation}
                   H = tr \int \left\{ i F_4 \right\} dx = tr \int \left\{-i Q R_{xxx} + i \frac{3}{2} QR (Q R_x- Q_x R) \right\} dx,
      \end{equation}
      \begin{equation}
        \left\{ Q(x)  \stackrel{\otimes}{,}  Q(y)   \right\}  = \left\{ R(x)    \stackrel{\otimes}{,}   Q(y)   \right\} = 0,
      \end{equation}
      \begin{equation}
        \left\{ Q(x)    \stackrel{\otimes}{,}   R(y)   \right\}  =  i \delta (x-y) \Pi,
      \end{equation}
       where $\left\{  X  \stackrel{\otimes}{,}  Y \right\}_{kl}^{ij}=\left\{ X_{ij} , Y_{kl}  \right\}$, and $\Pi$ denotes the $n^2 \times n^2$ permutation matrix. The complete integrability can be proved by a classical $r-$matrix  \cite{tsuchida1998coupled}.  The exact solutions of mmKdV equation as $p=q=n$ and  $\epsilon=-1$ have been studied by using the classical  inverse scattering method \cite{tsuchida1998coupled}. Different from previous work about mmKdV equation, our work is to perfect the exact solutions of mmKdV equation for any positive integers $p,q$ and $\epsilon = \pm 1$ via RH approach, and obtain some other interesting and meaningful phenomenon  by analyzing the special properties of potential matrix which is important to understand  mmKdV equation more thoroughly.

       The outline of this work is as follows. In Section 2, we perform the spectral analysis of Lax pair, and analyze the  symmetry and analyticity of scattering matrix. In Section 3, on the basis of the results of  the last section,  the RH problem is formulated. In Section 4, we complete the spatiotemporal evolution of the scattered data. In Section 5, we obtain general form of solutions  by solving the RH problem. Finally,  some specific forms of potential matrix are considered. According to the special properties unique to the particular forms, we can further classify soliton solutions based on the original soliton solutions. Then, we obtain some other interesting solutions such as $N$-soliton solutions, breather-type soliton solutions and  bell-type soliton solutions.  The conclusions are discussed in the last section.

\section{Spectral   Analysis}
  \subsection{Lax pair and eigenfuction}
   The equivalent form of Lax pair of Eq. (\ref{mmkdv}) can be written as

\begin{equation}  \label{Lax}
\left\{
 \begin{aligned}
         \Phi_x & =M\Phi, & M =& -i \zeta \sigma +U,  \\
         \Phi_t  & =N\Phi, & N=& -4 i \zeta^3 \sigma + 4 \zeta^2 U - 2 i \zeta (   U^2 +   U_{x} )\sigma + 2 U^3- U_{xx} +  U_{x}U - UU_x,
    \end{aligned}
\right.
\end{equation}
with
\begin{equation}
  \sigma = \begin{bmatrix} I_1 & 0 \\ 0 & -I_2 \end{bmatrix}, \quad  U = \begin{bmatrix} 0 & Q \\ R & 0 \end{bmatrix},
\end{equation}
where   $ \Phi = \Phi (x, t; \zeta)$ is  a  $(p+q)$-component vector, $Q$ is a  $p\times q$ matrix, $R$ is a $q\times p$ matrix,    and $I_1$  and $ I_2$ are $p \times p$ and $q \times q $ identity matrix, respectively.  The potential matrices satisfy  $R= \epsilon Q^{\dagger}$ and $\epsilon =\pm 1$.  In addition,  Eq. \eqref{mmkdv} can be derived via the compatibility  condition of  Eq. \eqref{Lax}.

For the convenience of discussion, we can rewritte  Eq. \eqref{Lax} as
\begin{equation} \label{lax2}
\left\{
 \begin{aligned}
   & \Phi_x + i \zeta \sigma  \Phi = U \Phi, \\
    & \Phi_t +4 i \zeta^3 \sigma \Phi = V \Phi,
 \end{aligned}
\right.
\end{equation}
where
\begin{equation}
          V =   4   \zeta^2  U + i \zeta ( 2 U_{x} U + 2 U_{x} )\sigma +  2 U^3- U_{xx} +  U_{x}U - UU_x.
\end{equation}
According to Eq. \eqref{lax2}, when  $|x| \to \infty $,
\begin{equation}
\Phi \propto \exp(-i \zeta \sigma x-4 i \zeta^3 \sigma t).
\end{equation}
Let  $\Psi=\Phi \exp(i \zeta \sigma x+4 i \zeta^3 \sigma t)$, then  $\Psi$ satisfy :
\begin{equation} \label{lax12}
         \Psi_x +i \zeta [\sigma, \Psi] =U \Psi,
\end{equation}
\begin{equation}\label{lax13}
       \Psi_ t+4 i \zeta^3 [\sigma, \Psi] =V \Psi,
\end{equation}
where $[\sigma, \Psi]=\sigma \Psi- \Psi \sigma $ is the commutator. Based on  Eqs. \eqref{lax12} and  \eqref{lax13}, we  can  get  the formula
\begin{equation}
d(e^{i(\zeta  x+4\zeta^3 t)\widehat{\sigma}}\Psi)=e^{i(\zeta  x+4\zeta^3 t)\widehat{\sigma}}(Udx+Vdt) \Psi.
\end{equation}

Now we begin to consider the spectral analysis of   Lax pair   \eqref{lax12} and \eqref{lax13},  for which we merely focus on the spectral problem  \eqref{lax12}, because the analysis will take place at a fixed time, and the $t$-dependence will be suppressed. As for \eqref{lax12},  we can write its two matrix Jost solutions as a  collection of columns, that is
\begin{equation}
\begin{aligned}
          \Psi_{1}& =([\Psi_{1}]_{1},[\Psi_{1}]_{2},\dots,[\Psi_{1}]_{p},[\Psi_{2}]_{p+1},\dots,[\Psi_{2}]_{p+q}),     \\
\Psi_{2}& =([\Psi_{2}]_{1},[\Psi_{2}]_{2},\dots,[\Psi_{2}]_{p},[\Psi_{1}]_{p+1},\dots,[\Psi_{1}]_{p+q}),     \\
\end{aligned}
\end{equation}
obeying the asymptotic conditions
\begin{equation}
\begin{aligned}
          \Psi_{1}\to \mathbb{I} ,&  \qquad   x \to - \infty,\\
 \Psi_{2}\to \mathbb{I} ,&  \qquad   x \to   +\infty.
\end{aligned}
\end{equation}
Here  $\mathbb{I}$  is a  $(p+q) \times (p+q)$ identity matrix.  $\Psi_{1}$ and $\Psi_{2}$  are uniquely determined by the integral equations of Volterra type
\begin{equation}
   \begin{aligned}
         \Psi_{1}&= \mathbb{I} + \int_{-\infty}^{x} e^{-i \zeta \sigma (x-y)} U(y) \Psi_{1}(y,\zeta)  e^{i \zeta \sigma (x-y)} dy,  \\
        \Psi_{2}&=  \mathbb{I} - \int_{x}^{+ \infty} e^{-i \zeta \sigma (x-y)} U(y) \Psi_{2}(y,\zeta)  e^{i \zeta \sigma (x-y)} dy.
   \end{aligned}
\end{equation}
By direct computation, we can get
\begin{equation}
        e^{-i \zeta \sigma (x-y)} U(y) e^{i \zeta \sigma (x-y)}=
\begin{pmatrix} 0 & e^{-i \zeta (x-y) I_1 }Qe^{-i \zeta  (x-y) I_2}    \\ e^{i \zeta (x-y) I_2 }R e^{i \zeta   (x-y) I_1}& 0 \end{pmatrix}.
\end{equation}
By  direct analysis, we can see that
\begin{equation}
[\Psi_{1}]_1,[\Psi_{1}]_2,\dots,[\Psi_{1}]_p,[\Psi_{2}]_{p+1},[\Psi_{2}]_{p+2}, \dots, [\Psi_{2}]_{p+q}
\end{equation}
 are analytic for  $\zeta \in \mathbb{C}^{+}$ and continuous for $\zeta \in \mathbb{C}^{+} \cup \mathbb{R}$,
  and
\begin{equation}
[\Psi_{2}]_1,[\Psi_{2}]_2, \dots,  [\Psi_{2}]_p,[\Psi_{1}]_{p+1},[\Psi_{1}]_{p+2},\dots,[\Psi_{1}]_{p+q}
\end{equation}
 are analytic for  $\zeta \in \mathbb{C}^{-}$ and   continuous for $\zeta \in \mathbb{C}^{-}\cup \mathbb{R} $,   where  $\mathbb{C}^{+}$ and  $\mathbb{C}^{-}$   respectively the upper and lower half  $\zeta$-plane.
It is indicated owing to the Abel¡¯s identity and $tr(U) = 0$ that the determinants of $\Psi_{1}$ and $\Psi_{2}$ are  independent of the variable $x$. Through evaluating $\det(\Psi_{1})$  at  $x=-\infty$  and $\det(\Psi_{2})$  at $x=+ \infty$ , we
know
\begin{equation}
     \det (\Psi_{1})=\det (\Psi_{2})=1  ,\quad   \zeta \in \mathbb{R}.
\end{equation}

\subsection{Symmetry and analyticity of scattering matrix}
Since  $\Psi_{1}E$ and $\Psi_{2}E$ are  matrix solutions of the spectral problem(12), where $E=e^{-i (\zeta  x+4\zeta^3 t)\sigma}$, therefore, $\Psi_{1}E$ and $\Psi_{2}E$ are linear  dependent, namely,
\begin{equation}  \label{scatter}
    \Psi_{1}E=\Psi_{2}E S(\zeta)   , \qquad  \zeta \in  \mathbb{R},
\end{equation}
where  $S(\zeta)$ is a $(p+q)\times (p+q)$  matrix, i.e., $S(\zeta)=(S_{i,j})_{(p+q)\times(p+q)}$.
By taking  determinant at both ends of  Eq. \eqref{scatter},  it is  obvious that
\begin{equation}
\det(S)=1.
\end{equation}

    In order  to construct the  RH  problem,  we   need to consider the inverse of  $\Psi_{1}$  and $\Psi_{2}$,  and we partition the inverse matrices of $\Psi_{1}$  and $\Psi_{2}$ into rows, that is
\begin{equation}
      \Psi_{1}^{-1}=\begin{pmatrix} [\Psi_{1}^{-1}]^{1}&\\ [\Psi_{1}^{-1}]^{2}&\\  \vdots \\  [\Psi_{1}^{-1}]^{p+q}  &   \end{pmatrix},   \qquad  \Psi_{2}^{-1}=\begin{pmatrix} [\Psi_{2}^{-1}]^{1}&\\ [\Psi_{2}^{-1}]^{2}&\\  \vdots \\  [\Psi_{2}^{-1}]^{p+q}  &   \end{pmatrix},
\end{equation}
in which each  $[\Psi_{1}^{-1}]^{l} $ and   $ [\Psi_{2}^{-1}]^{l}$ denote the  $l-th$ row of the matrices $ \Psi_{1}^{-1}$  and  $\Psi_{2}^{-1}$.  At   the same time,  $ \Psi_{1}^{-1}$  and  $\Psi_{2}^{-1}$  meet the equation
\begin{equation} \label{lax21}
    \Upsilon_{x}+ i\zeta  [\sigma ,\Upsilon]=- \Upsilon U.
\end{equation}
With the aid  of Eq. \eqref{lax21}, we can  see that
 \begin{equation}
 [\Psi_{1}^{-1}]_{1}, [\Psi_{1}^{-1}]_{2},\dots, [\Psi_{1}^{-1}]_{p}, [\Psi_{2}^{-1}]_{p+1}, \dots,  [\Psi_{2}^{-1}]_{p+q}
  \end{equation}
 are  analytic in $\zeta \in \mathbb{C}^{-}$, while
 \begin{equation}
 [\Psi_{2}^{-1}]_{1},  [\Psi_{2}^{-1}]_{2}, \dots, [\Psi_{2}^{-1}]_{p},  [\Psi_{1}^{-1}]_{p+1},  \dots,  [\Psi_{1}^{-1}]_{p+q}
 \end{equation}
 are  analytic in $\zeta \in \mathbb{C}^{+}$.  Moreover, from Eq. \eqref{scatter},  we  can get that
\begin{equation}
   E^{-1} \Psi_{1}^{-1}=R(\zeta)E^{-1} \Psi_{2}^{-1},
\end{equation}
where $R(\zeta)=(R_{i,j})_{(p+q) \times (p+q)}=S^{-1}(\zeta)$, we can call it as inverse scattering matrix. \\
After  finding the scattering matrix $S(\zeta)$ and the inverse scattering matrix $R(\zeta)$, we give the following theorem to describe the  analyticity  of the elements for the two matrices.

\begin{thm}
       Assume that  the scattering matrix $S(\zeta)$ and the inverse scattering matrix $R(\zeta)$ are divided into the following forms
        \begin{equation} \label{3}
                \begin{pmatrix} S_{1}&S_{2}  \\   S_{3} &  S_{4}
                 \end{pmatrix},
              \qquad
                              \begin{pmatrix} R_{1}&R_{2}  \\  R_{3} &  R_{4}
                 \end{pmatrix},
        \end{equation}
       where   $S_{1}$ and $R_{1}$ are p$\times$p matrices, $ S_{4}$ and $R_{4}$  are  q$\times$q  matrices,   $ S_{2}$ and $R_{2}$ are p$\times$q matrices, and $S_{3}$ and $R_{3}$ are q$\times$p  matrices. Then the  elements of $S_{1}$ and $R_{4}$  can be analytic extension to $\zeta \in \mathbb{C}^{+}$; the  elements of $S_{4}$ and $R_{1}$  can be   extended analytically to $\zeta \in \mathbb{C}^{-}$;  the  elements of $S_{2},S_{3},R_{2}$ and $R_{3}$  are not analytic in $\zeta \in \mathbb{C}^{+}$ and $\zeta \in \mathbb{C}^{-}$;  the  elements of $S_{2}$ and $S_{3}$  are  continuous in $\zeta \in \mathbb{R}$;  the  elements of $R_{2}$ and $R_{3}$  can not  be   extended analytically to $\zeta \in \mathbb{R}$.
\end{thm}
\begin{proof}
          From Eq. \eqref{scatter},  we can get
        \begin{equation}
              E^{-1} \Psi_{2}^{-1} \Psi_{1} E =  S(\zeta).
        \end{equation}
More  explicity, we have

 \begin{equation}
 \begin{aligned}
 \Psi_{2}^{-1} \Psi_{1} = &  \\
           & \begin{bmatrix}
   [\Psi_{2}^{-1} ]_{1}[\Psi_{1}]_{1}& \dots & [\Psi_{2}^{-1} ]_{1}[\Psi_{1}]_{p} &[\Psi_{2}^{-1} ]_{1}[\Psi_{1}]_{p+1}& \dots &  [\Psi_{2}^{-1} ]_{1}[\Psi_{1}]_{p+q} \\
                                            \vdots &      \ddots                  & \vdots &   \vdots   &   \ddots& \vdots        \\
                                      [\Psi_{2}^{-1} ]_{p}[\Psi_{1}]_{1}    &   \dots &[\Psi_{2}^{-1} ]_{p}[\Psi_{1}]_{p}                                                         &[\Psi_{2}^{-1} ]_{p}[\Psi_{1}]_{p+1}&   \dots  & [\Psi_{2}^{-1} ]_{p}[\Psi_{1}]_{p+q} \\
  [\Psi_{2}^{-1} ]_{p+1}[\Psi_{1}]_{1}& \dots & [\Psi_{2}^{-1} ]_{p+1}[\Psi_{1}]_{p} &[\Psi_{2}^{-1} ]_{p+1}[\Psi_{1}]_{p+1}& \dots &  [\Psi_{2}^{-1} ]_{p+1}[\Psi_{1}]_{p+q} \\
                                            \vdots &      \ddots                  & \vdots &   \vdots   &   \ddots& \vdots        \\
                                      [\Psi_{2}^{-1} ]_{p+q}[\Psi_{1}]_{1}    &   \dots &[\Psi_{2}^{-1} ]_{p+q}[\Psi_{1}]_{p}                                                         &[\Psi_{2}^{-1} ]_{p+q}[\Psi_{1}]_{p+1}&   \dots  & [\Psi_{2}^{-1} ]_{p+q}[\Psi_{1}]_{p+q}
          \end{bmatrix}.
 \end{aligned}
 \end{equation}
With the aid of  the analyticity of the  columns of $\Psi_{1}$,  the rows of $\Psi_{2}$ and the  matrice E, we can proof the analyticity of the elements of $S(\zeta)$. Similarly, the analyticity of the elements of $R(\zeta)$ can be obtained. Thus, we complete the proof.
\end{proof}

In addition,  the potential  matrix $U$  has  the symmerty  as  follows:
\begin{equation}
       U=\begin{bmatrix}  0&  Q \\  R & 0  \end{bmatrix},\qquad  R= \epsilon Q^{\dagger}, \quad  \epsilon = \pm 1.
\end{equation}
\textbf{Case 1}:       $\epsilon = -1$.

From $R= - Q^{\dagger}$, we  can derive that  $U^{\dagger}=-U$, and then   obtain
\begin{equation}  \label{e1}
           \Psi_{i}^{\dagger}(\zeta^{*})=\Psi_{i}^{-1}(\zeta), \qquad   i=1,2.
\end{equation}
According  to  Eq. \eqref{scatter}, we  have
\begin{equation}
    S^{\dagger}(\zeta^{*})=R(\zeta),
\end{equation}
which  gives the  following relationships
\begin{equation}
    \left\{
          \begin{array}{l l l}
                            s_{i,j}^{*}(\zeta^{*})=r_{j,i}(\zeta), & \qquad  1\leq i \leq p, \quad  1\leq j \leq p,  &\qquad  \zeta \in \mathbb{C}^{+},  \\
                             s_{i,j}^{*}(\zeta)=r_{j,i}(\zeta), & \qquad  1\leq i \leq p,  \quad  p\leq j \leq p+q,   &\qquad  \zeta \in \mathbb{R},  \\
                              s_{i,j}^{*}(\zeta)=r_{j,i}(\zeta), & \qquad  p\leq i \leq p+q,  \quad  1\leq j \leq p,   &\qquad  \zeta \in \mathbb{R},  \\
                             s_{i,j}^{*}(\zeta^{*})=r_{j,i}(\zeta), & \qquad  p+1\leq i \leq p+q, \quad  p+1\leq j \leq p+q,  &\qquad  \zeta \in \mathbb{C}^{-}.  \\
          \end{array}
          \right.
\end{equation}
\textbf{Case 2}:       $\epsilon = 1$.

From $R=  Q^{\dagger}$,  we  can see that  $U^{\dagger}=-\sigma U \sigma $,  and then   obtain
\begin{equation} \label{e2}
           \Psi_{i}^{\dagger}(\zeta^{*})=\sigma \Psi_{i}^{-1}(\zeta) \sigma, \qquad   i=1,2.
\end{equation}
According  to  Eq. \eqref{scatter}, we  have
\begin{equation}
    S^{\dagger}(\zeta^{*})= \sigma R(\zeta) \sigma.
\end{equation}
Thus, we  successfully obtain  the relationship between the  scattering matrix and the inverse scattering matrix.

\section{Riemann-Hilbert Problem}
 A RH problem desired for the mmKdV  equation  involves two matrix functions: one is analytic in  $\mathbb{C}^{+}$,  and the other is analytic in  $\mathbb{C}^{-}$. Define the first matrix  function \begin{equation}
   \Gamma_{1}(x,\zeta)=([\Psi_{1}]_1,[\Psi_{1}]_2,\dots,[\Psi_{1}]_p,[\Psi_{2}]_{p+1},[\Psi_{2}]_{p+2},\dots, [\Psi_{2}]_{p+q}),
\end{equation}
which is an  analytic  function  of $\zeta$ in $\mathbb{C}^{+}$.  Next, we study the very large $\zeta$ asymptotic  behavior of $\Gamma_{1}(x,\zeta)$, which has the asymptotic expansion
\begin{equation}  \label{expansion}
   \Gamma_{1}=\Gamma_{1}^{(0)}+\frac{\Gamma_{1}^{(1)}}{\zeta}+\frac{\Gamma_{1}^{(2)}}{\zeta^2}+O(\frac{1}{\zeta^3}), \quad \zeta \to \infty.
\end{equation}
Inserting Eq. \eqref{expansion} into Eq. \eqref{lax12},  and comparing the coefficients
of $\zeta$ directly bring about
\begin{equation}
  \begin{array}{lll}
      O(1) &: & \Gamma_{1,x}^{(0)}+ i [\sigma,\Gamma_{1}^{(1)}]=U \Gamma_{1}^{(0)}, \\
  O(\zeta)& : &i [\sigma,\Gamma_{1}^{(0)}]=0, \\
  O(\zeta^{3}) &: &  i [\sigma,\Gamma_{1}^{(1)}]= U \Gamma_{1}^{(0)},  \\
  \end{array}
\end{equation}
from which we have $\Gamma_{1}^{(0)}=\mathbb{I} $,   namely,
\begin{equation}
      \Gamma_{1} \to \mathbb{I},\qquad  \zeta \in \mathbb{C}^{+}\to \infty.
\end{equation}
We  can  introduce the  matrix function    $\Gamma_{2}$,   which is  analytic in $\zeta \in \mathbb{C}^{-}$ in terms  of
\begin{equation}
      \Gamma_{2}(x,\zeta)= \begin{pmatrix}             [\Psi_{1}^{-1}]_{1}  \\   \vdots     \\   [\Psi_{1}^{-1}]_{p}    \\    [\Psi_{2}^{-1}]_{p+1}  \\   \vdots         \\  [\Psi_{2}^{-1}]_{p+q} \end{pmatrix} (x, \zeta).
\end{equation}
Similarly,  $\Gamma_{2}$ satisfies  the following condition:
\begin{equation}
      \Gamma_{2} \to \mathbb{I},\qquad  \zeta \in \mathbb{C}^{-}\to \infty.
\end{equation}
In  addition,  we  define some auxiliary  matrices
\begin{equation}
 H_{i}=diag(0,\dots,0,\mathop{\underline{1}}\limits_{i},0,\dots,0),
\end{equation}
then
\begin{equation}
   \begin{aligned}
             \Gamma_{1} &= \Psi_{1}H_{1}+\dots+ \Psi_{1}H_{p}+\Psi_{2}H_{p+1}+\dots+ \Psi_{2}H_{p+q},  \\
             \Gamma_{2} &= H_{1}\Psi_{1}^{-1}+\dots+ H_{p}\Psi_{1}^{-1}+H_{p+1}\Psi_{2}^{-1}+\dots+ H_{p+q}\Psi_{2}^{-1}.
    \end{aligned}
\end{equation}
Before  we give the normal  RH  Problem,  we need to do  some calculations:
\begin{equation}
   \begin{aligned}
                             \Gamma_{2}(\zeta)\Gamma_{1}(\zeta)       = &   ( H_{1}\Psi_{1}^{-1}+\dots+ H_{p}\Psi_{1}^{-1}+H_{p+1}\Psi_{2}^{-1}+\dots+ H_{p+q}\Psi_{2}^{-1}  )                    \\
                                   &    \times ( \Psi_{1}H_{1}+\dots+ \Psi_{1}H_{p}+\Psi_{2}H_{p+1}+\dots+ \Psi_{2}H_{p+q}  ).                  \\
    \end{aligned}
\end{equation}
Note that
\begin{equation}
      H_{i} H_{j} = \left\{
       \begin{array}{lll}
          H_{i},  & i=j,  &  1 \leq i,j \leq p+q,  \\
            0,    & i\neq j, &  1 \leq i,j \leq p+q,
      \end{array}
    \right.
\end{equation}

\begin{equation}
\begin{aligned}
 \Psi_{1}^{-1} \Psi_{2}&=E R(\zeta) E^{-1} = &    \\
        & \begin{bmatrix}
             r_{1,1} &  \dots &  r_{1,p}&r_{1,p+1}e^{-2i(\zeta x+4\zeta^3 t)}&  \dots &  r_{1,p+q}e^{-2i(\zeta x+4\zeta^3 t)}\\
              \vdots &  \ddots &   \vdots  &   \vdots  &  \ddots &    \vdots   \\
             r_{p,1} &  \dots &  r_{p,p}&r_{p,p+1}e^{-2i(\zeta x+4\zeta^3 t)}&  \dots &  r_{p,p+q}e^{-2i(\zeta x+4\zeta^3 t)}\\
             r_{p+1,1}e^{2i(\zeta x+4\zeta^3 t)} &  \dots &  r_{p+1,p}e^{2i(\zeta x+4\zeta^3 t)}&   r_{p+1,p+1}&  \dots &  r_{p+1,p+q}\\
              \vdots &  \ddots &   \vdots  &   \vdots  &  \ddots &    \vdots   \\
              r_{p+q,1}e^{2i(\zeta x+4\zeta^3 t)} &  \dots &  r_{p+q,p}e^{2i(\zeta x+4\zeta^3 t)}&   r_{p+q,p+1}&  \dots &  r_{p+q,p+q}
 \end{bmatrix},&   \\
 \end{aligned}
\end{equation}
and
\begin{equation}
\begin{aligned}
     \Psi_{2}^{-1} \Psi_{1}&=E S(\zeta) E^{-1}= &\\
      &   \begin{bmatrix}
             s_{1,1} &  \dots &  s_{1,p}&  s_{1,p+1}e^{-2i(\zeta x+4\zeta^3 t)}  &  \dots &  s_{1,p+q}e^{-2i(\zeta x+4\zeta^3 t)}\\
              \vdots &  \ddots &   \vdots  &   \vdots  &  \ddots &    \vdots   \\
             s_{p,1} &  \dots &  s_{p,p}&   s_{p,p+1}e^{-2i(\zeta x+4\zeta^3 t)}&  \dots &  s_{p,p+q}e^{-2i(\zeta x+4\zeta^3 t)}\\
             s_{p+1,1}e^{2i(\zeta x+4\zeta^3 t)} &  \dots &  s_{p+1,p}e^{2i(\zeta x+4\zeta^3 t)}&   s_{p+1,p+1}&  \dots &  s_{p+1,p+q}\\
              \vdots &  \ddots &   \vdots  &   \vdots  &  \ddots &    \vdots   \\
              s_{p+q,1}e^{2i(\zeta x+4\zeta^3 t)} &  \dots &  s_{p+q,p}e^{2i(\zeta x+4\zeta^3 t)}&   s_{p+q,p+1}&  \dots &  s_{p+q,p+q}
 \end{bmatrix},  &  \\
 \end{aligned}
\end{equation}
we  can  obtain
\begin{equation} \label{RH2}
  \begin{aligned}
    \Gamma_{2} \Gamma_{1} =  & J(x,\zeta) \\
     = & \begin{bmatrix}
               1 &  \dots &  0      &  r_{1,p+1}e^{-2i(\zeta x+4\zeta^3 t)}  &  \dots & r_{1,p+q}e^{-2i(\zeta x+4\zeta^3 t)}\\
              \vdots &  \ddots &   \vdots  &   \vdots  &  \ddots &    \vdots   \\
             0   &  \dots &  1&   r_{p,p+1}e^{-2i(\zeta x+4\zeta^3 t)}&  \dots &  r_{p,p+q}e^{-2i(\zeta x+4\zeta^3 t)}\\
             s_{p+1,1}e^{2i(\zeta x+4\zeta^3 t)} &  \dots &  s_{p+1,p}e^{2i(\zeta x+4\zeta^3 t)}&   1  &  \dots &  0  \\
              \vdots &  \ddots &   \vdots  &   \vdots  &  \ddots &    \vdots   \\
              s_{p+q,1}e^{2i(\zeta x+4\zeta^3 t)} &  \dots &  s_{p+q,p}e^{2i(\zeta x+4\zeta^3 t)}&  0  &  \dots &1
 \end{bmatrix}.
 \end{aligned}
\end{equation}

We  denote   $\Gamma_{1}$  and $\Gamma_{2}$ as   $\Gamma^{+}$  and $\Gamma^{-}$,   based  on  which  a RH  problem  can be set  up  as follows:
\begin{itemize}
  \item    $\Gamma^{\pm}$ are analytic in $\mathbb{C}^{+}=  \overline{\mathbb{C}^{+}}\backslash  \Sigma $   and   $\mathbb{C}^{-}=  \overline{\mathbb{C}^{-}}\backslash  \Sigma $, respectively,  where  $\Sigma$  represent the directed path  along the positive direction of   $\mathrm{Im} \zeta =0$.
  \item    jump  condition
                \begin{equation}
                   \Gamma^{-}\Gamma^{+}=J(x,\zeta),  \qquad  \zeta \in \mathbb{R}
                 \end{equation}
          where $J(x,\zeta)$ is from  Eq. \eqref{RH2}.
  \item  normalization  condition
           \begin{equation}
                  \begin{aligned}
                             \Gamma^{+}  \to \mathbb{I},    & \quad as  &  \zeta \to  \infty,  \\
                             \Gamma^{-}  \to \mathbb{I},    & \quad as  &  \zeta \to  \infty.
                  \end{aligned}
            \end{equation}
\end{itemize}

\section{Time  evolution  of  the scattering data}
In the previous section, we have obtained  the zeros  of $\det(\Gamma_{1})$ or $\det(\Gamma_{2})$. The zero and non-zero vectors $\vartheta_{j}$, $\widehat{\vartheta}_{j}$ form complete scattering data, which satisfy

 \begin{equation} \label{time}
       \begin{aligned}
              \Gamma_{1}(\zeta_{j}) \vartheta_{j} &=0,   \\
                \widehat{\vartheta}_{j} \Gamma_{2}(\zeta_{j}^{*}) & =0,
       \end{aligned}
\end{equation}
where $\vartheta_{j}$ and $\widehat{\vartheta}_{j}$  represent row vector and column vector, respectively.

Next, we begin to analyze the law of time and space evolution. First, we  derive the Eq. \eqref{time}  from $x$,
\begin{equation} \label{time2}
       \begin{aligned}
               \Gamma_{1,x}\vartheta_{j} + \Gamma_{1} \vartheta_{j,x} & =0,   \\
               \Gamma_{1,t}\vartheta_{j} + \Gamma_{1} \vartheta_{j,t} & =0.
       \end{aligned}
\end{equation}
Since
\begin{equation} \label{s1}
          \Gamma_{1,x}=\Psi_{1,x}H_{1}+\cdots +\Psi_{1,x}H_{p}+\Psi_{2,x}H_{p+1}+\cdots +\Psi_{2,x}H_{p+q},
\end{equation}
and
\begin{equation}  \label{s2}
         \Psi_{x}=- i \zeta [\sigma ,  \Psi] + U \Psi,
\end{equation}
we  can obtian
\begin{equation} \label{ss1}
          \Gamma_{1,x}=- i \zeta_{j} [\sigma , \Gamma_{1}] + U \Gamma_{1}.
\end{equation}
Following the  similar  process  as $\Gamma_{1,x}$, we get
\begin{equation}\label{ss2}
          \Gamma_{1,t}=- 4 i \zeta_{j}^{3} [\sigma , \Gamma_{1}] + U \Gamma_{1}.
\end{equation}
Substituting  Eqs. \eqref{ss1} and  \eqref{ss2}  into Eq. \eqref{time2}, and note that  $\Gamma_{1}\vartheta_{j}=0$, we  have
\begin{equation}
    \begin{aligned}
            \vartheta_{j,x} + i \zeta_{j} \sigma \vartheta_{j} & =0, \\
            \vartheta_{j,t} + 4i \zeta_{j}^{3} \sigma \vartheta_{j} & =0.
    \end{aligned}
\end{equation}
By sloving the above vector differential equations, we have
\begin{equation}\label{tt}
    \vartheta_{j}=e^{-i(\zeta_{j} x + 4\zeta_{j}^{3}t) \sigma} \vartheta_{j,0},
\end{equation}
where  $\vartheta_{j,0}$ is a constant vector.

Moreover, for $\epsilon=-1$,  we  have  $\Psi_{i}^{\dagger}(\zeta^{*})=\Psi_{i}^{-1}(\zeta) (i=1,2)$,  then
\begin{equation}  \label{sy1}
  \Gamma_{1}^{\dagger}(\zeta^{*})=\Gamma_{2}(\zeta), \qquad  \zeta  \in \mathbb{C}^{-}.
\end{equation}
From Eq. \eqref{time} and Eq. \eqref{sy1},  we know
\begin{equation}
     \widehat{\vartheta_{j}}=\vartheta_{j}^{\dagger},  \qquad   1\leq j \leq N,
\end{equation}
where $N$ is the numeber of zeros of $\det(\Gamma_{1}(\zeta))$.  Therefore,
\begin{equation}
     \widehat{\vartheta}_{j} = \vartheta_{j,0}^{\dagger} e^{i(\zeta_{j}^{*}x +4 \zeta^{*3}_{j}t)\sigma}.
\end{equation}
For $\epsilon= 1$,  we  have  $\Psi_{i}^{\dagger}(\zeta^{*})=\sigma \Psi_{i}^{-1}(\zeta)\sigma,(i=1,2) $,  then
\begin{equation}   \label{sy2}
  \Gamma_{1}^{\dagger}(\zeta^{*})=\sigma \Gamma_{2}(\zeta) \sigma , \qquad  \zeta  \in \mathbb{C}^{-}.
\end{equation}
From Eq. \eqref{time} and Eq. \eqref{sy2},  we know
\begin{equation}
     \widehat{\vartheta_{j}}=\vartheta_{j}^{\dagger} \sigma,  \qquad   1\leq j \leq N,
\end{equation}
where  $N$ is the numeber of zeros of $\det(\Gamma_{1}(\zeta))$. Therefore,
\begin{equation}
     \widehat{\vartheta}_{j} = \vartheta_{j,0}^{\dagger} e^{i(\zeta_{j}^{*}x +4 \zeta^{*3}_{j}t)\sigma} \sigma, \qquad   1\leq j \leq N.
\end{equation}
Thus, we comlpete the analysis  of the law of time and space evolution.

\section{Exact solutions of mmKdV equation}
In fact,   $\Gamma_{1}$  has  the asymptotic expansion as follows
\begin{equation}  \label{expansion1}
     \Gamma_{1}(\zeta)= \mathbb{I} + \frac{P_{1}^{(1)}}{\zeta} +  \frac{P_{1}^{(2)}}{\zeta^2} + O(\frac{1}{\zeta^3}), \qquad  \zeta \to \infty.
\end{equation}
Substituting  the expression \eqref{expansion1} into Eq. \eqref{lax12}  yields
\begin{equation}\label{compare}
   i [\sigma, \Gamma_{1}^{(1)}] = U.
\end{equation}
Comparing with the elements of the matrices of Eq. \eqref{compare}, we can finally recover the  potential  function.

 We  define
\begin{equation}
  Q=
  \begin{bmatrix}
     Q_{1,1}   & Q_{1,2} & \dots & Q_{1,q}  \\
      \vdots      & \vdots    & \ddots   & \vdots \\
     Q_{p,1}    &Q_{p,2}      &  \dots  & Q_{p,q}
   \end{bmatrix},
\end{equation}
where
\begin{equation}
   Q_{i,j}= 2 i (\Gamma_{1}^{(1)})_{i,j+p},  \qquad   1\leq i,j \leq p.
\end{equation}
For the above RH problem, under the reflectionless case($S_3$=0, in Eq. \eqref{3}), the discrete spectrum corresponds to multiple soliton solutions. In addition, the solution $\Gamma_{1}$, $\Gamma_{2}$ of the RH problem can be given explicitly,
\begin{equation} \label{solu}
  \begin{aligned}
       \Gamma_{1}(\zeta) = \mathbb{I} -\sum_{k=1}^{N} \sum_{j=1}^{N} \frac{\vartheta_{k}\widehat{\vartheta}_{j}(M^{-1})_{k,j}}{\zeta-\zeta_{j}^{*}},   \\
       \Gamma_{2}(\zeta) =  \mathbb{I} + \sum_{k=1}^{N} \sum_{j=1}^{N} \frac{\vartheta_{k}\widehat{\vartheta}_{j}(M^{-1})_{k,j}}{\zeta-\zeta_{k}},
  \end{aligned}
\end{equation}
where $M$ is a $N\times N$ matrix, the  element of $M$ is $ M_{k,j}=\frac{\widehat{\vartheta}_{k}\vartheta_{j}}{\zeta_{j}-\zeta_{k}^{*}} $.
From the above equations, we  can get
\begin{equation}
  \Gamma_{1}^{(1)}=-\sum_{k=1}^{N} \sum_{j=1}^{N} \vartheta_{k}\widehat{\vartheta}_{j}(M^{-1})_{k,j}.
\end{equation}
Let
\begin{equation}
  \vartheta_{k,0}=(\alpha_{1,k},\alpha_{2,k},\dots,\alpha_{p+q,k})^{T}, \qquad \theta_{k}=-i(\zeta_{k} x+4\zeta_{k}^{3}t),
\end{equation}
then we obtain the  solutions of the mmKdV equation as follows:

\begin{equation} \label{solution}
  \begin{aligned}
           Q_{m,n}=2i \epsilon \sum_{k=1}^{N} \sum_{j=1}^{N} \alpha_{m,k} \alpha_{n+p,j}^{*}e^{\theta_{k}-\theta_{j}^{*}}(M^{-1})_{k,j}, \quad \epsilon = \pm 1.
   \end{aligned}
\end{equation}
If we define
\begin{equation}
   F_{i,j}=
       \begin{pmatrix}
                0 &  \alpha_{i,1}e^{\theta_{1}} & \alpha_{i,2}e^{\theta_{2}} & \dots & \alpha_{i,N}e^{\theta_{N}} \\
               -\epsilon \alpha_{j+p,1}^{*}e^{-\theta_{1}^{*}}  & M_{1,1}  & M_{1,2} &  \dots   &  M_{1,N}  \\
            - \epsilon    \alpha_{j+p,2}^{*}e^{-\theta_{2}^{*}}  & M_{2,1}  & M_{2,2} &  \dots   &  M_{2,N}  \\
                  \vdots &  \vdots  & \vdots  &     &  \vdots              \\
            - \epsilon     \alpha_{j+p,N}^{*}e^{-\theta_{N}^{*}}  & M_{N,1}  & M_{N,2} &  \dots   &  M_{N,N}  \\
       \end{pmatrix},  \quad \epsilon = \pm 1,
\end{equation}
  we obtain the explicit expression of   the exact solutions  of  mmKdV equation by
\begin{equation}
         Q_{i,j}=2i \frac{\det(F_{i,j})}{\det(M)}, \qquad  1 \leq i \leq p,   1 \leq j \leq q.
\end{equation}

\section{Special examples}

Next, let's consider the following special cases:

\subsection{ $Q$ is taken as $2\times 2$ form}
 At first, we consider the following form of potential matrix
\begin{equation} \label{Q1}
    Q= \begin{bmatrix}  u_1(x,t)  &  u_2(x,t)\\  -u_{2}(x,t)  & u_{1}(x,t)  \end{bmatrix},
\end{equation}
where  $R=Q^{\dagger}$ and $p=q=2$.
It is a direct calculation to verify that the two-component  mKdV equations of system \eqref{cmkdv} can be obtained by substituting Eq. \eqref{Q1} into Eq. \eqref{mmkdv}, which is also studied in  \cite{zhang2008lax}.

\subsubsection{ A variety of Rational Solutions and Physical Visions}
Taking  $N=1$ in Eq. \eqref{solution}, once-iterated solutions are as follows
\begin{equation} \label{solution 11}
\left\{
   \begin{aligned}
          u_{1}(x,t)=  &  2i \frac{\alpha_{1,1} \alpha_{3,1}^{*} (\zeta_{1}-\zeta_{1}^{*})e^{\theta_1-\theta_{1}^{*}}}{ (|\alpha_{1,1}|^2+|\alpha_{2,1}|^2)e^{\theta_1+\theta_{1}^{*} }- (|\alpha_{3,1}|^2+|\alpha_{4,1}|^2)e^{-\theta_1-\theta_{1}^{*} }},  \\
          u_{2}(x,t)=  &  2i \frac{\alpha_{1,1} \alpha_{4,1}^{*} (\zeta_{1}-\zeta_{1}^{*})e^{\theta_1-\theta_{1}^{*}}}{ (|\alpha_{1,1}|^2+|\alpha_{2,1}|^2)e^{\theta_1+\theta_{1}^{*} } - (|\alpha_{3,1}|^2+|\alpha_{4,1}|^2)e^{-\theta_1-\theta_{1}^{*} }},
   \end{aligned}
\right.
\end{equation}
where
\begin{equation}
\zeta_{1}=a_1+i b_1, \quad \theta_1= - i (\zeta_1 x +4 \zeta_1^3 t).
\end{equation}
If we set  $\alpha_{3,1}=\frac{\sqrt{3}}{2}$, $\alpha_{4,1}=\frac{1}{2}$  and $|\alpha_{1,1}|^2 = |\alpha_{2,1}|^2 =\frac{1}{2} e^{2 \xi}$,  then Eq. \eqref{solution 11} can be simplified as follows
\begin{equation} \label{sf1.1}
\left\{
       \begin{aligned}
             u_1(x,t) &= -\sqrt{3} \alpha_{1,1} b_1 e^{\theta_1-\theta_1^{*}-\xi} csch(\theta_{1}+\theta_{1}^{*}+\xi), \\
              u_2(x,t) &= -  \alpha_{1,1} b_1 e^{\theta_1-\theta_1^{*}-\xi} csch(\theta_{1}+\theta_{1}^{*}+\xi).
       \end{aligned}
\right.
\end{equation}

The localized structures and  dynamic behaviors  of once-iterated solutions are plotted in Fig. 1.
 It is noteworthy that although the waves corresponding to the soutions  propagate in a fixed direction  and show strong periodicity,   the values  are infinity in certain points,  i.e. the solutions in Eq. \eqref{sf1.1}  are  singular  which are different from usual soliton solutions.

\noindent
{\rotatebox{0}{\includegraphics[width=3.6cm,height=3.0cm,angle=0]{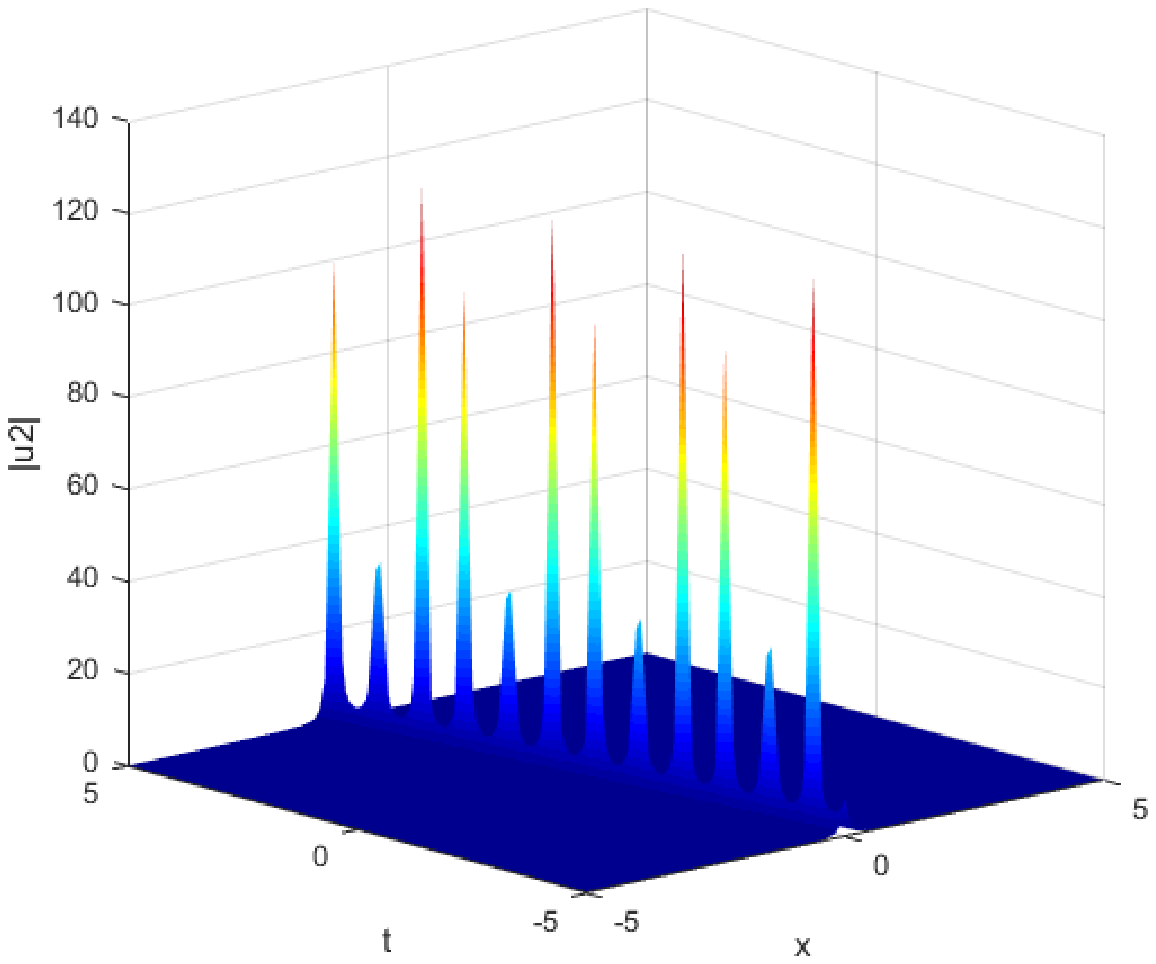}}}
~~~~
{\rotatebox{0}{\includegraphics[width=3.6cm,height=3.0cm,angle=0]{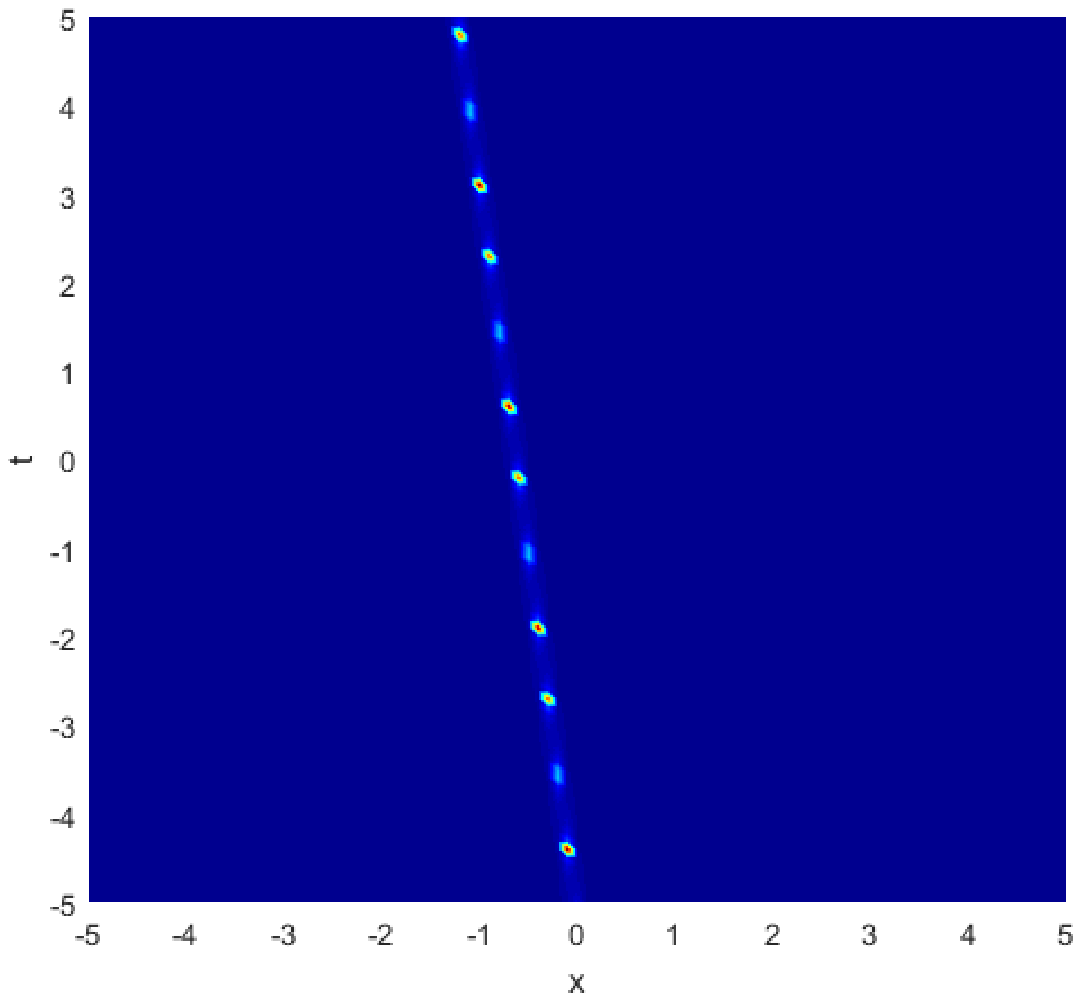}}}
~~~~
{\rotatebox{0}{\includegraphics[width=3.6cm,height=3.0cm,angle=0]{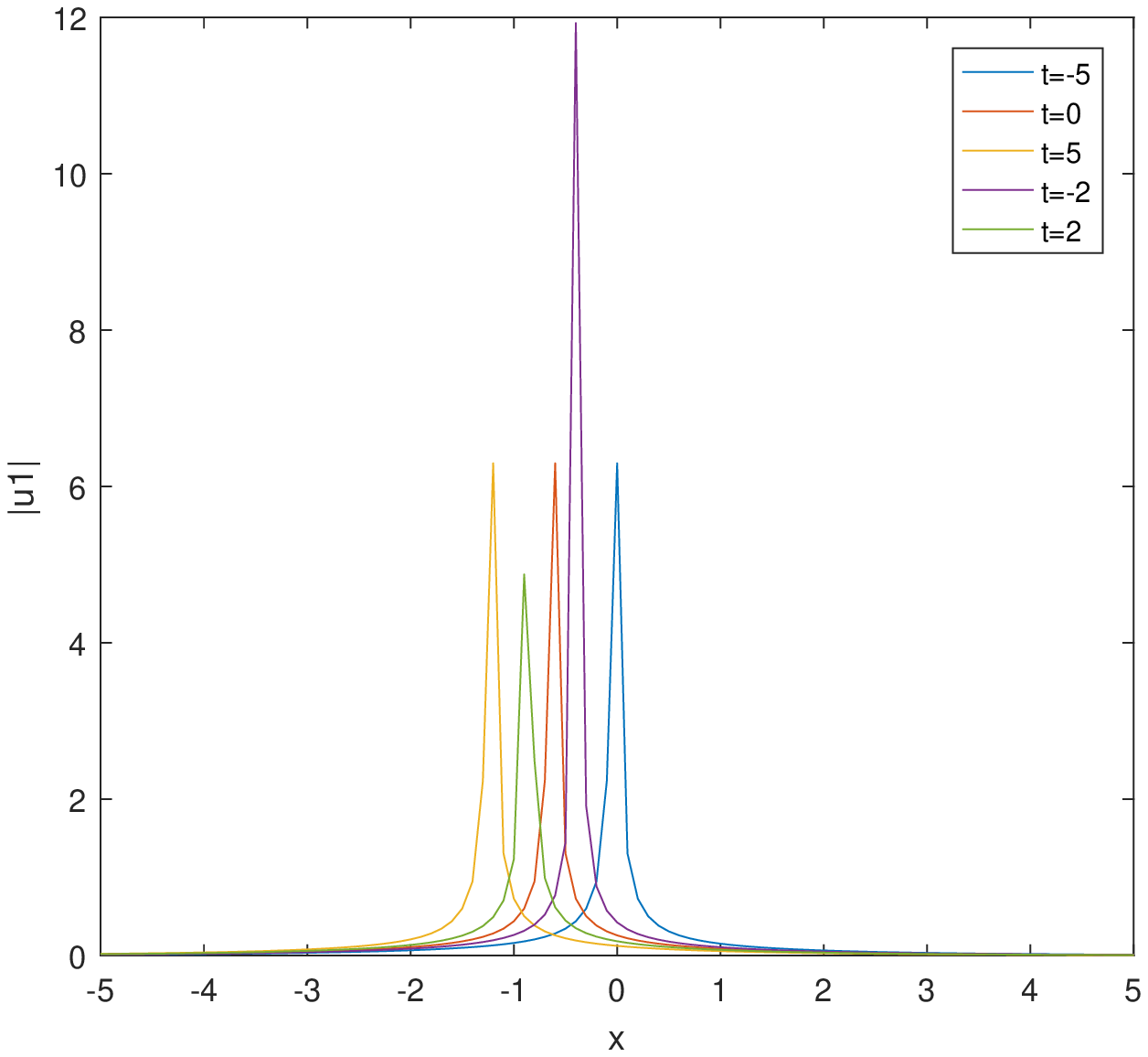}}}

$\ \qquad~~~~~~(\textbf{a})\qquad \ \qquad\qquad\qquad\qquad~(\textbf{b})
\ \qquad\qquad\qquad\qquad\qquad~(\textbf{c})$\\
\noindent
{\rotatebox{0}{\includegraphics[width=3.6cm,height=3.0cm,angle=0]{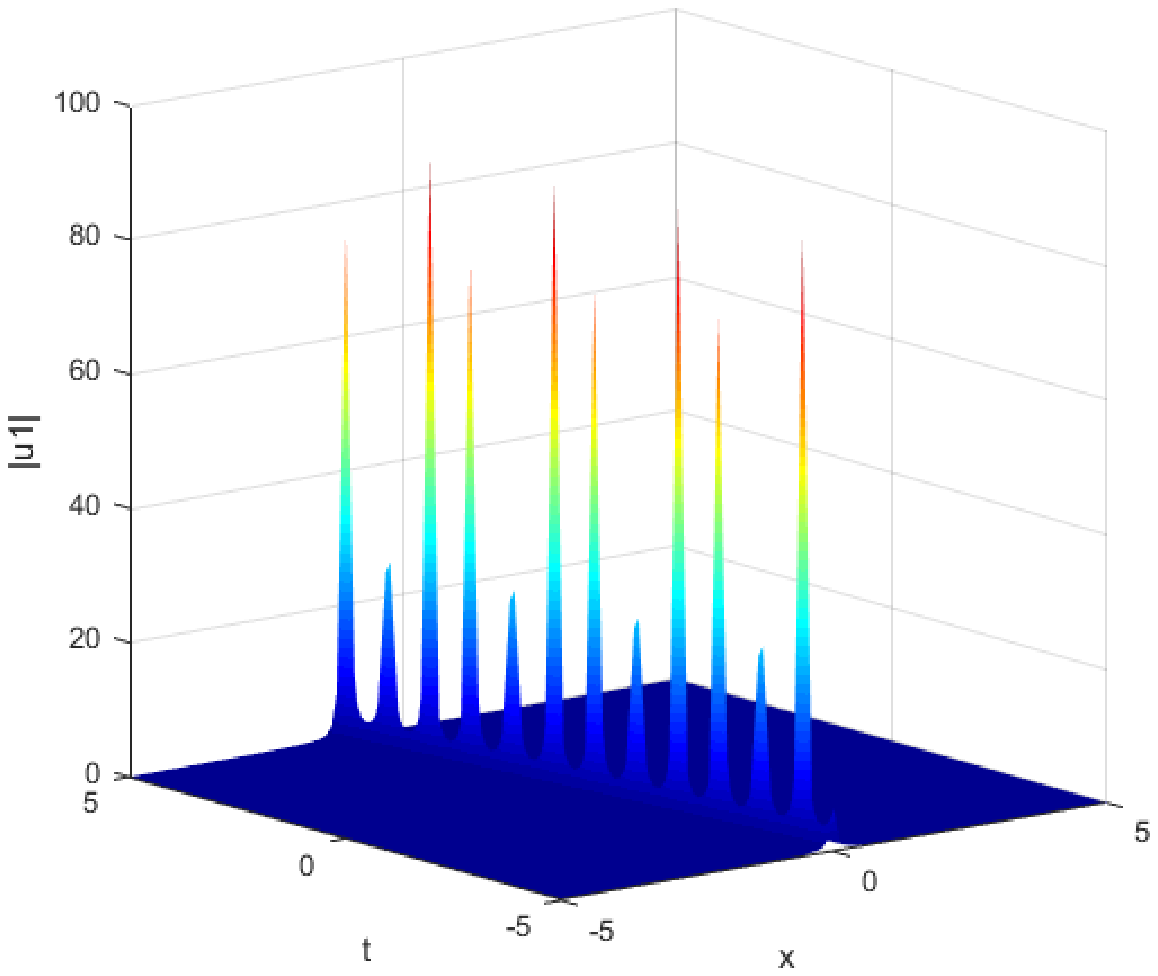}}}
~~~~
{\rotatebox{0}{\includegraphics[width=3.6cm,height=3.0cm,angle=0]{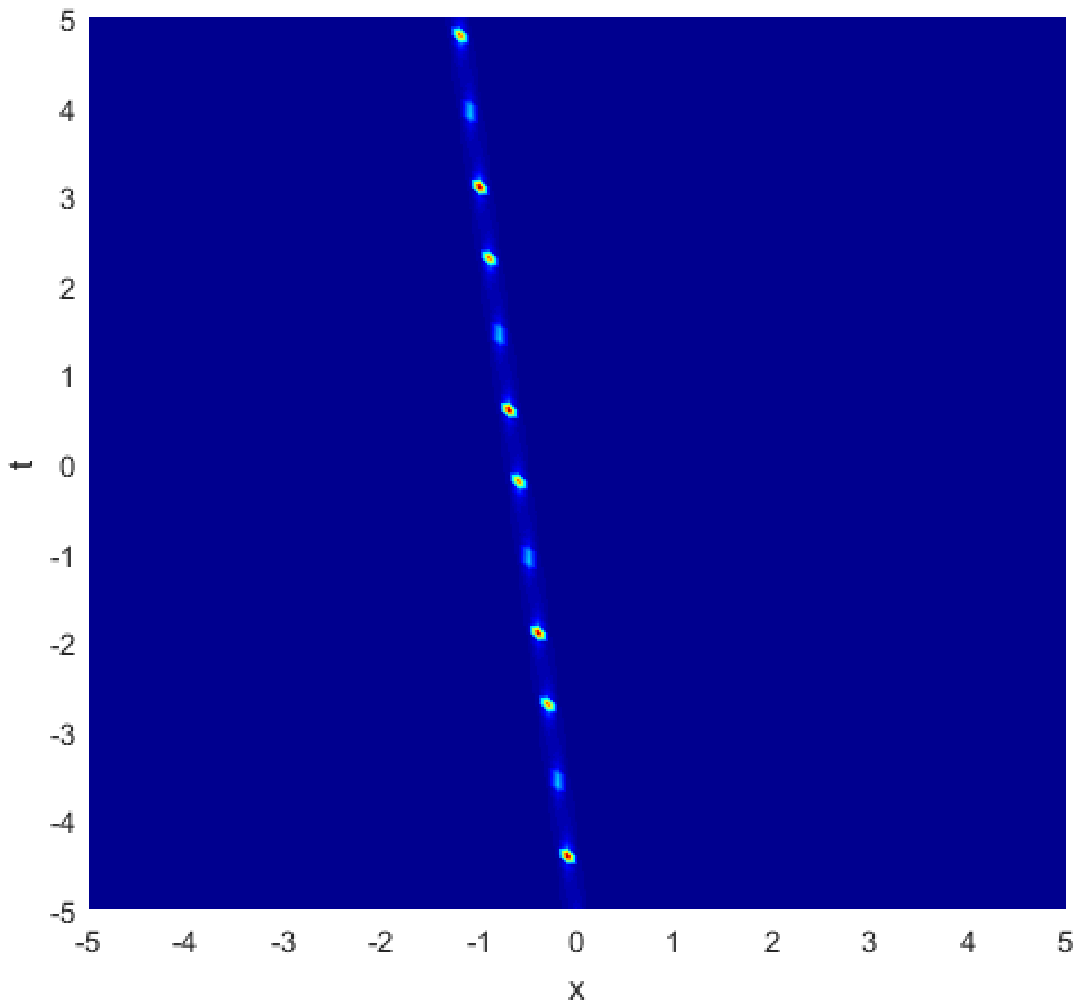}}}
~~~~
{\rotatebox{0}{\includegraphics[width=3.6cm,height=3.0cm,angle=0]{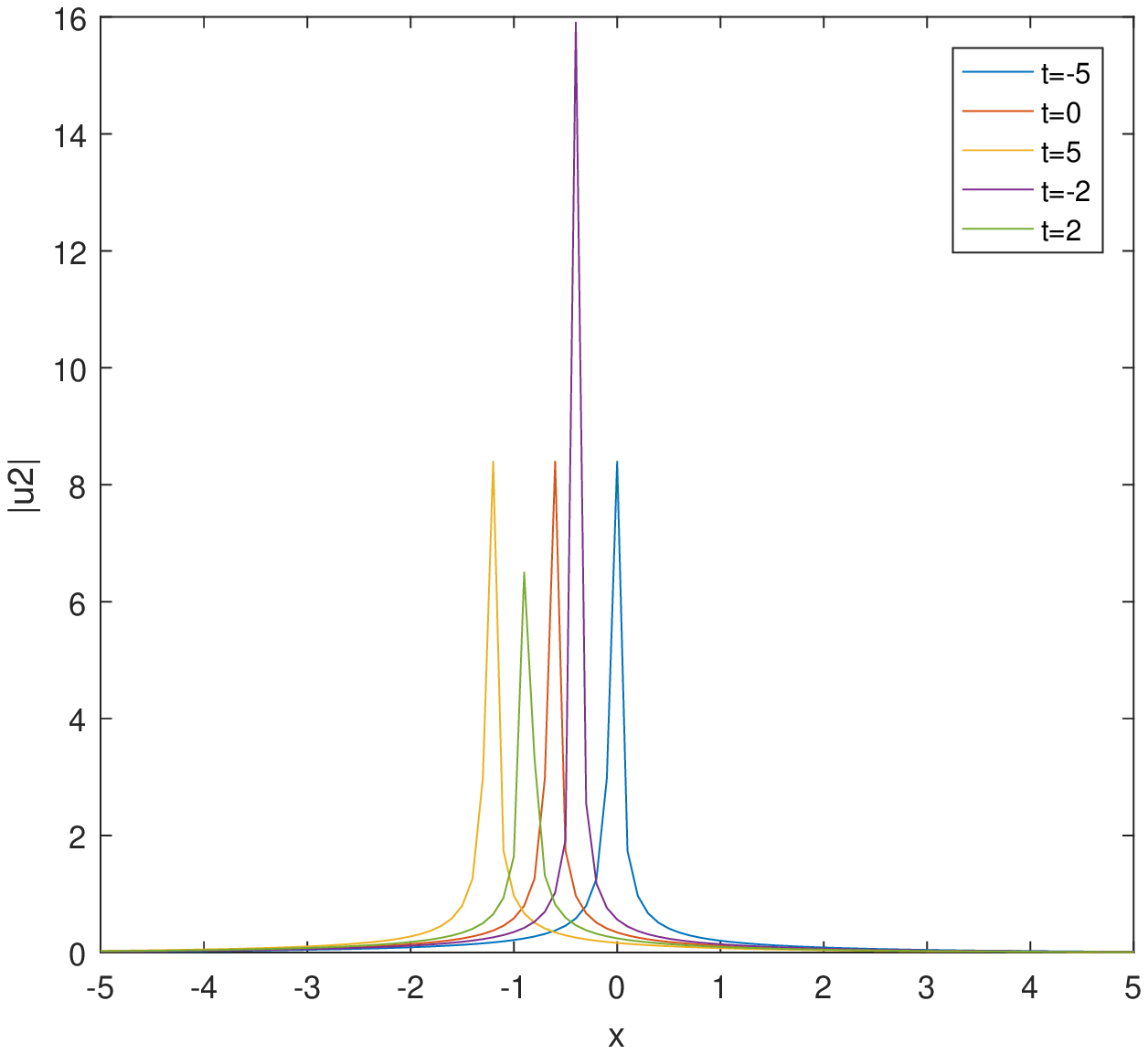}}}

$\ \qquad~~~~~~(\textbf{d})\qquad \ \qquad\qquad\qquad\qquad~(\textbf{e})
\ \qquad\qquad\qquad\qquad\qquad~(\textbf{f})$\\
\noindent { \small \textbf{Figure 1.} (Color online) Once-iterated solutions  to Eq. \eqref{sf1.1}  with the parameters  $a_1=0.2$, $b_1=0.3$,  $\xi=0.3$, $\alpha_{1,1}=\alpha_{1,2}=\frac{\sqrt{2}}{2} e^{\xi}$.
\label{fig1.1}
$\textbf{(a)(d)}$: the structures of the once-iterated solutions,
$\textbf{(b)(e)}$: the density plot,
$\textbf{(c)(f)}$: the wave propagation of the once-iterated solutions.} \\

In what follows, if we take $N=2$,  the twice-iterated solutions can be expressed by
\begin{equation}
\left\{
   \begin{aligned}
          u_{1}(x,t)=  &  \frac{2i}{M_{11}M_{22}-M_{12}M_{21}}   (\alpha_{1,1}\alpha_{3,1}^{*}e^{\theta_{1}-\theta_{1}^{*}}M_{22}-\alpha_{1,1}\alpha_{3,2}^{*}e^{\theta_{1}-\theta_{2}^{*}}M_{12}\\
    &-\alpha_{1,2}\alpha_{3,1}^{*}e^{\theta_{2}-\theta_{1}^{*}}M_{21}+\alpha_{1,2}\alpha_{3,2}^{*}e^{\theta_{2}-\theta_{2}^{*}}M_{11}),  \\
        u_{2}(x,t)=  &  \frac{2i}{M_{11}M_{22}-M_{12}M_{21}}  (\alpha_{1,1}\alpha_{4,1}^{*}e^{\theta_{1}-\theta_{1}^{*}}M_{22} - \alpha_{1,1}\alpha_{4,2}^{*}e^{\theta_{1}-\theta_{2}^{*}}M_{12} \\
  & -\alpha_{1,2}\alpha_{4,1}^{*}e^{\theta_{2}-\theta_{1}^{*}}M_{21}+\alpha_{1,2}\alpha_{4,2}^{*}e^{\theta_{2}-\theta_{2}^{*}}M_{11}),  \\
   \end{aligned}
\right.
\end{equation}
where
\begin{equation}
\left\{
   \begin{aligned}
         M_{11}= &  \frac{(|\alpha_{1,1}|^2+|\alpha_{2,1}|^2)e^{\theta_{1}^{*}+\theta_{1}}-(|\alpha_{3,1}|^2+|\alpha_{4,1}|^2)e^{-\theta_{1}^{*}-\theta_{1}}}{\zeta_{1}-\zeta_{1}^{*}}, \\
         M_{12}= &  \frac{(\alpha_{1,1}^{*}\alpha_{1,2}+\alpha_{2,1}^{*}\alpha_{2,2})e^{\theta_{1}^{*}+\theta_{2}}-(\alpha_{ 3,1}^{*}\alpha_{3,2}+\alpha_{4,1}^{*}\alpha_{4,2})e^{-\theta_{1}^{*}-\theta_{2}}}{\zeta_{2}-\zeta_{1}^{*}},  \\
          M_{21}= &  \frac{(\alpha_{1,2}^{*}\alpha_{1,1}+\alpha_{2,2}^{*}\alpha_{2,1})e^{\theta_{2}^{*}+\theta_{1}}-(\alpha_{ 3,2}^{*}\alpha_{3,1}+\alpha_{4,2}^{*}\alpha_{4,1})e^{-\theta_{2}^{*}-\theta_{1}}}{\zeta_{1}-\zeta_{2}^{*}},   \\
         M_{22}= &  \frac{(|\alpha_{1,2}|^2+|\alpha_{2,2}|^2)e^{\theta_{2}^{*}+\theta_{2}}-(|\alpha_{3,2}|^2+|\alpha_{4,2}|^2)e^{-\theta_{2}^{*}-\theta_{2}}}{\zeta_{2}-\zeta_{2}^{*}},\\
   \end{aligned}
\right.
\end{equation}
with  $\zeta_{1}=a_1+i b_1$, $\zeta_{2}=a_2+i b_2$, $\theta_1=-i(\zeta_1 x + 4 \zeta_1^3 t)$ and $\theta_2=-i(\zeta_2 x + 4 \zeta_2^3 t)$.

If we select the parameters as $\alpha_{3,1}=\alpha_{3,2}=\frac{\sqrt{2}}{2}$, $\alpha_{4,1}=\alpha_{4,2}=-\frac{\sqrt{2}}{2}$, $\alpha_{1,1}=\alpha_{1,2}$, $\alpha_{2,1}=\alpha_{2,2}$ and $|\alpha_{1,1}|^2 = |\alpha_{2,1}|^2 =\frac{1}{2} e^{2 \xi}$, then  the explicit solutions can be expressed as follow
\begin{equation} \label{sf1.2}
\left\{
   \begin{aligned}
         u_1(x,t) & = \sqrt{2} i \frac{(\alpha_{1,1}e^{\theta_{1}-\theta_{1}^{*}} M_{2,2}- \alpha_{1,1} e^{\theta_1-\theta_2^{*}}M_{1,2} -\alpha_{1,2} e^{\theta_2-\theta_1^{*}}M_{2,1} + \alpha_{1,2} e^{\theta_2- \theta_2^{*}}M_{1,1})}{M_{11}M_{22}-M_{12}M_{21}},   \\
                  u_2(x,t) & = - \sqrt{2}  i  \frac{(\alpha_{1,1}e^{\theta_{1}-\theta_{1}^{*}} M_{2,2}- \alpha_{1,1} e^{\theta_1-\theta_2^{*}}M_{1,2} -\alpha_{1,2} e^{\theta_2-\theta_1^{*}}M_{2,1} + \alpha_{1,2} e^{\theta_2- \theta_2^{*}}M_{1,1})}{M_{11}M_{22}-M_{12}M_{21} },
   \end{aligned}
\right.
\end{equation}
where
\begin{equation}
\left\{
   \begin{aligned}
         M_{11}  & = -\frac{i}{b_1}e^{\xi} \sinh (\theta_1 +\theta_1^{*}+ \xi), \\
         M_{12}  & = \frac{2 e^{\xi}}{a_2-a_1+i(b_1+b_2)} \sinh (\theta_1^{*} + \theta_2 + \xi), \\
         M_{21}  & = \frac{2 e^{\xi}}{a_1-a_2 +i(b_1+b_2)} \sinh (\theta_2^{*} + \theta_1 + \xi), \\
         M_{22}  & = -\frac{i}{b_2}e^{\xi} \sinh (\theta_2 +\theta_2^{*}+ \xi).
   \end{aligned}
\right.
\end{equation}

The localized structures and  dynamic behaviors  of twice-iterated solutions are  plotted  in  Fig. 2.
It can be seen from Fig. 2  that with the special parameters, the solution as a whole is wavy and shows a significant attenuation trend as time goes on.

\noindent
{\rotatebox{0}{\includegraphics[width=3.7cm,height=3.0cm,angle=0]{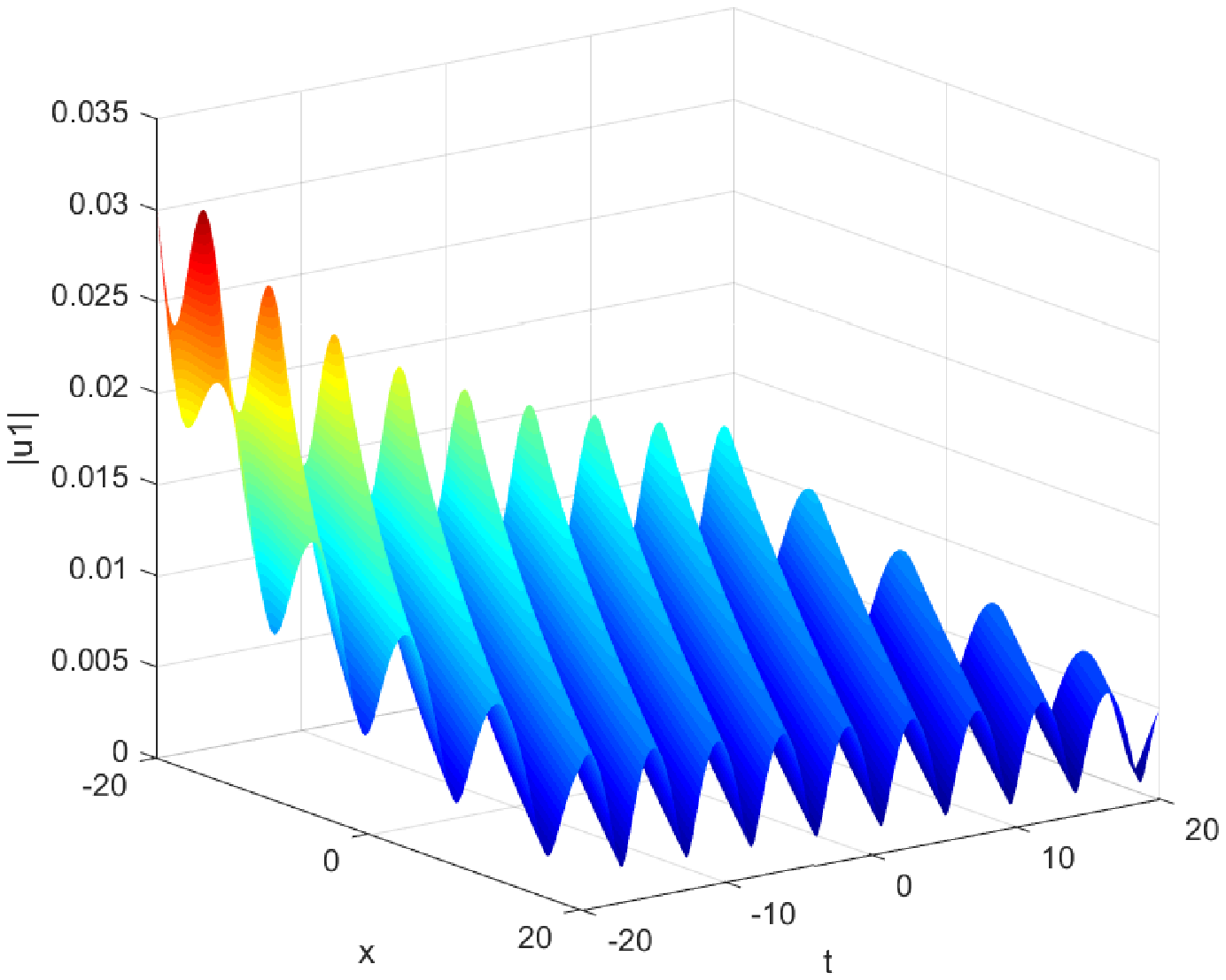}}}
~~~~
{\rotatebox{0}{\includegraphics[width=3.7cm,height=3.0cm,angle=0]{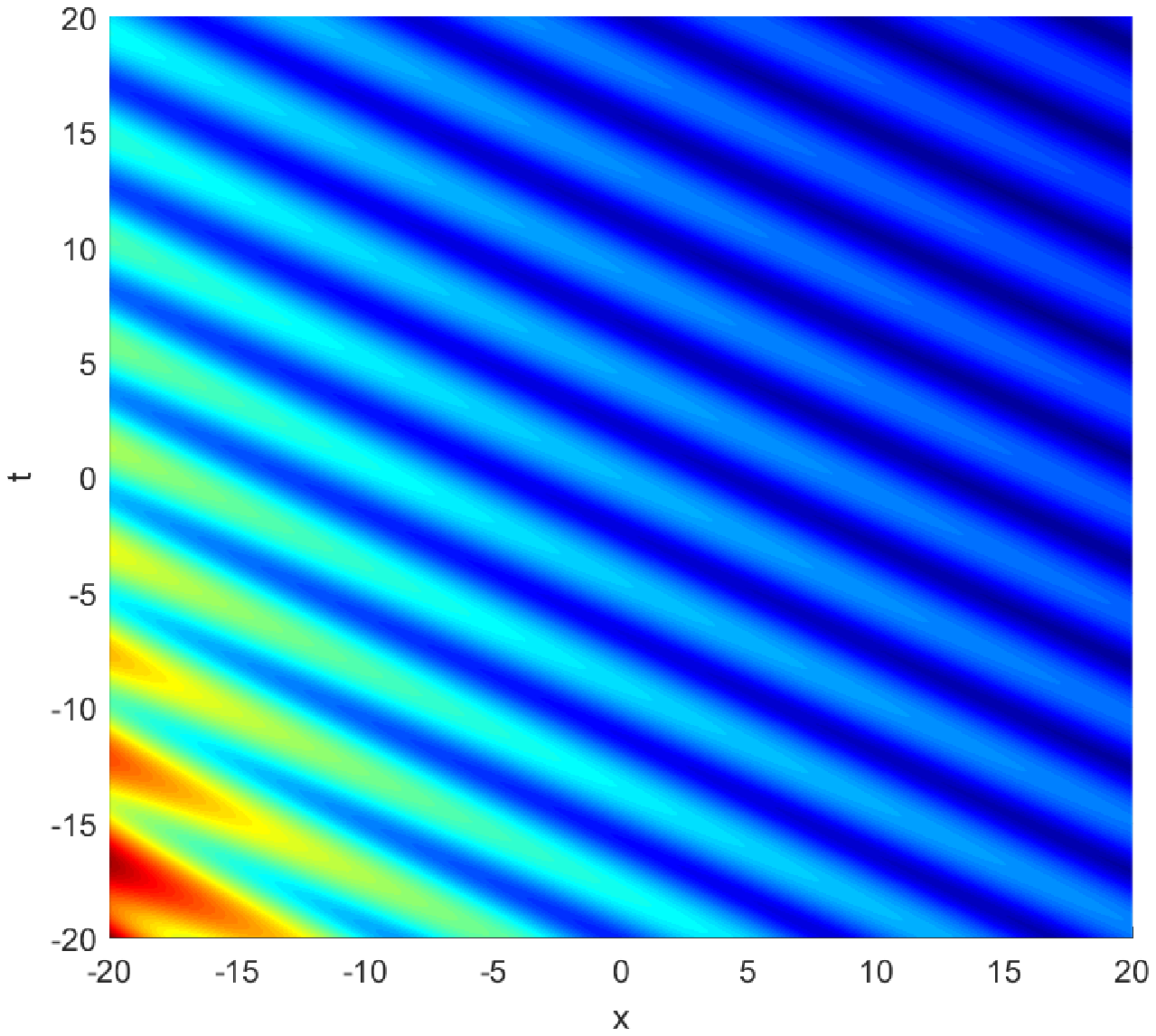}}}
~~~~
{\rotatebox{0}{\includegraphics[width=3.7cm,height=3.0cm,angle=0]{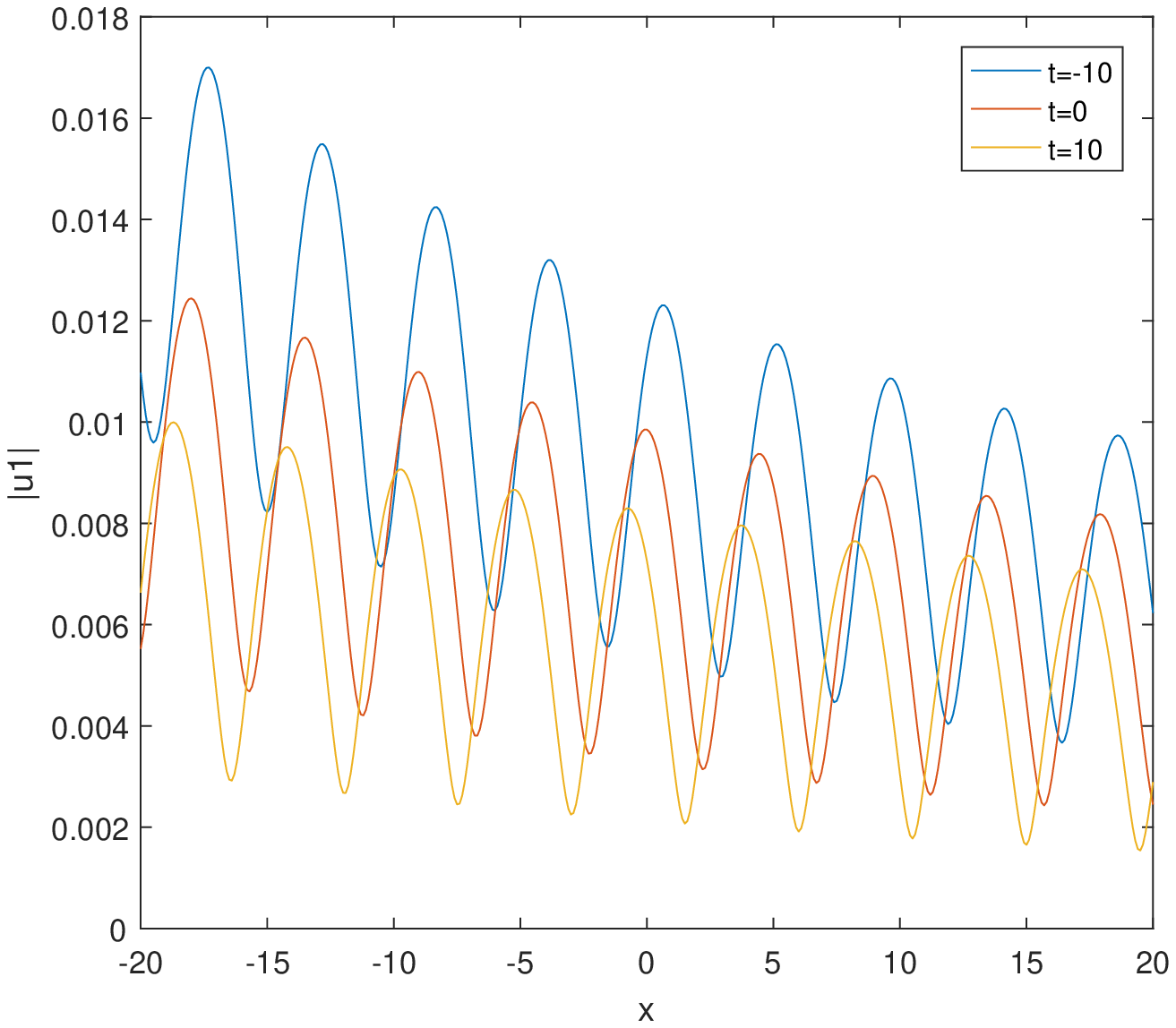}}}

$\ \qquad~~~~~~(\textbf{a})\qquad \ \qquad\qquad\qquad\qquad~(\textbf{b})
\ \qquad\qquad\qquad\qquad\qquad~(\textbf{c})$\\
\noindent
{\rotatebox{0}{\includegraphics[width=3.6cm,height=3.0cm,angle=0]{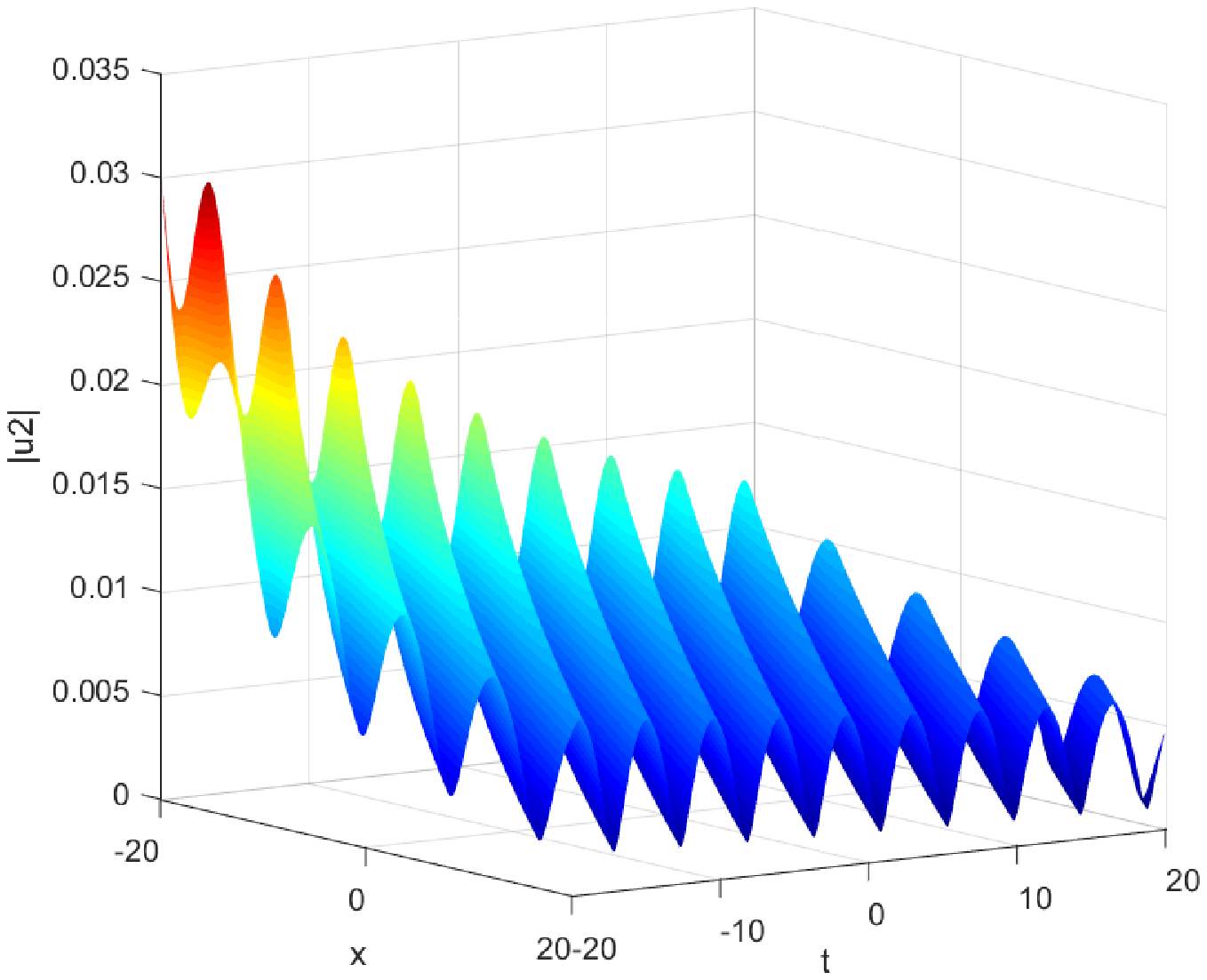}}}
~~~~
{\rotatebox{0}{\includegraphics[width=3.6cm,height=3.0cm,angle=0]{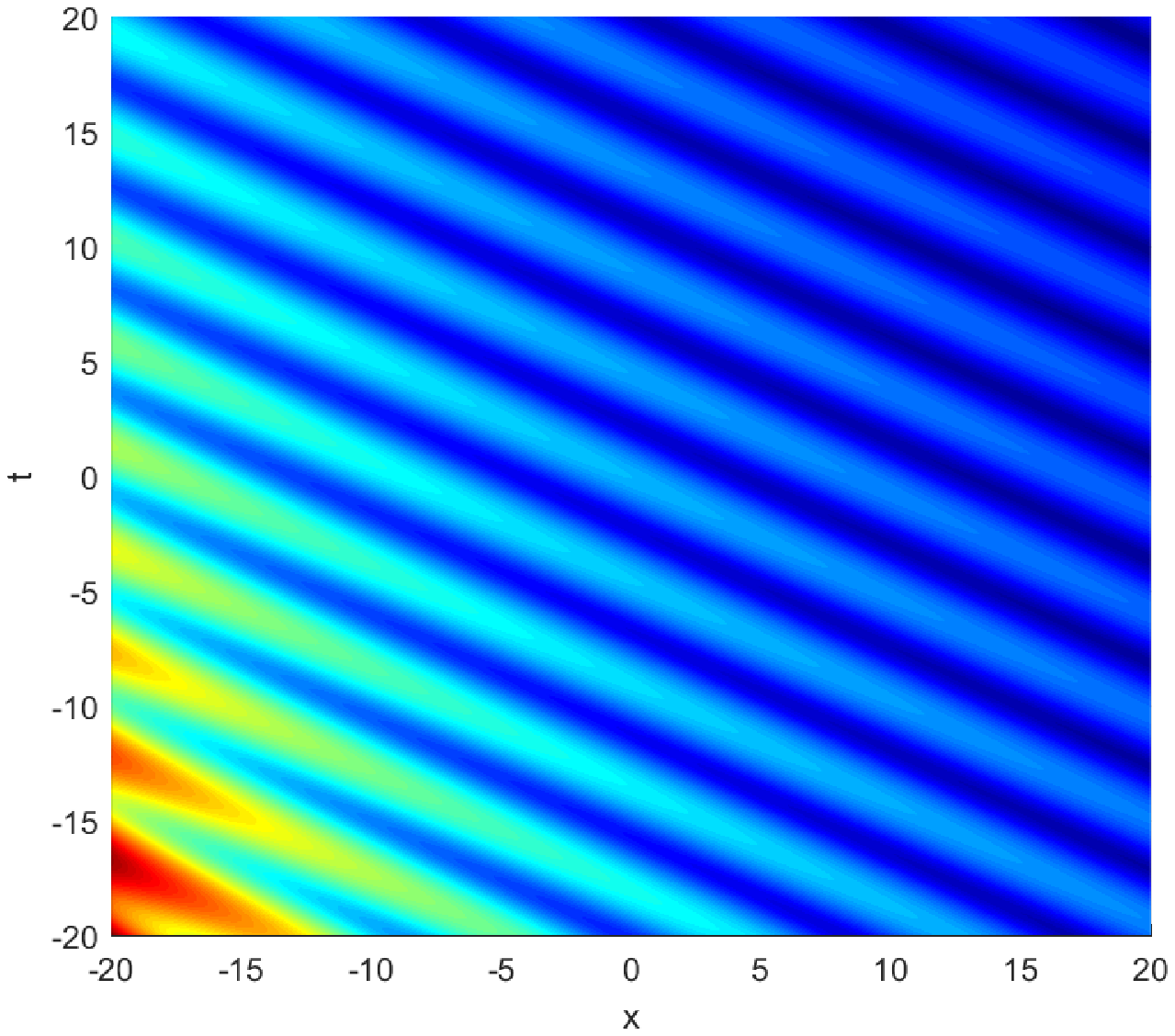}}}
~~~~
{\rotatebox{0}{\includegraphics[width=3.6cm,height=3.0cm,angle=0]{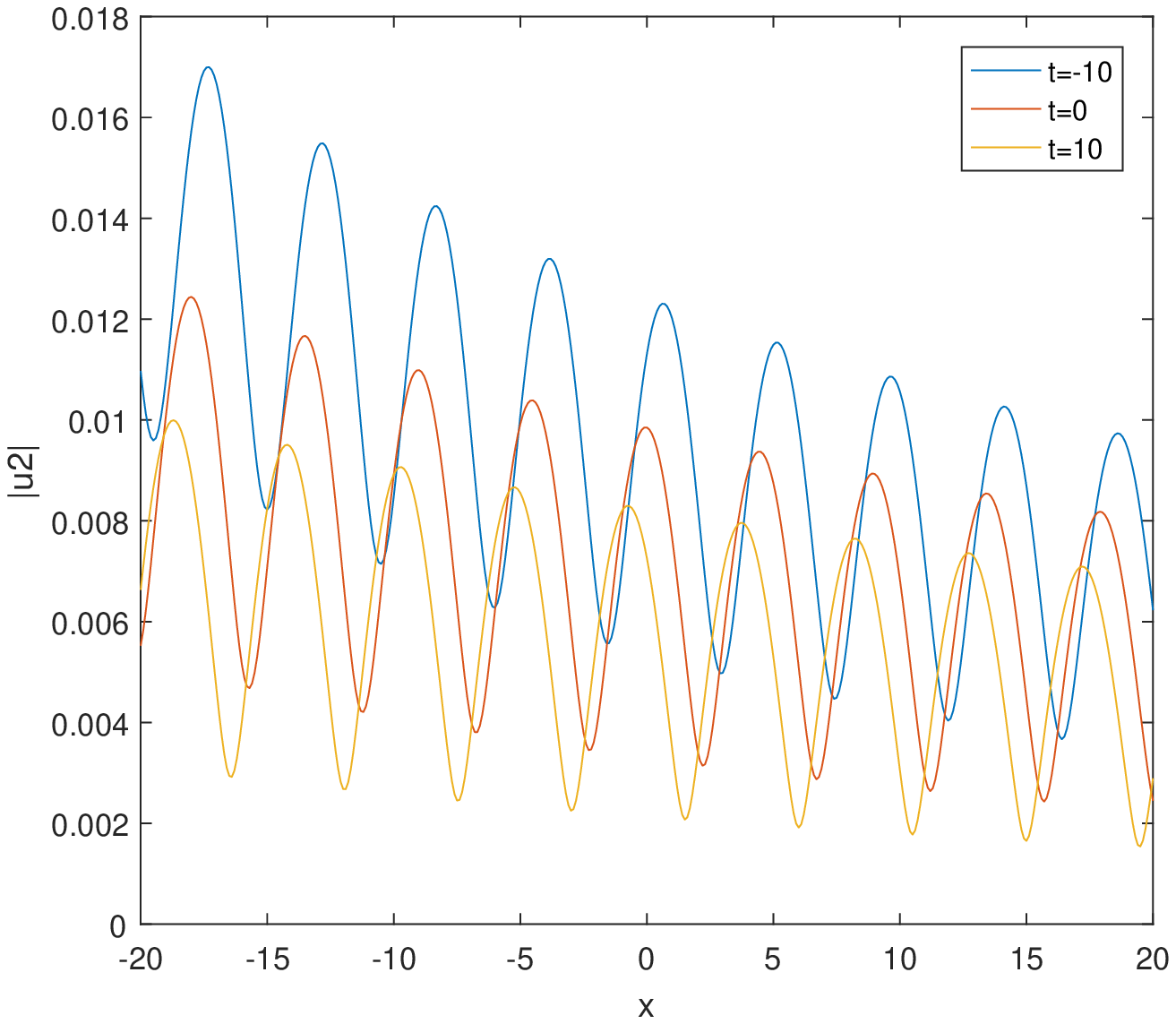}}}

$\ \qquad~~~~~~(\textbf{d})\qquad \ \qquad\qquad\qquad\qquad~(\textbf{e})
\ \qquad\qquad\qquad\qquad\qquad~(\textbf{f})$\\
\noindent { \small \textbf{Figure 2.} Twice-iterated  solutions to Eq. \eqref{sf1.2} with parameters  $\xi=0.3$, $a_1=-0.3$, $b_1=0.001$, $a_2=0.4$,  $b_2=0.002$, $\alpha_{1,1}=\alpha_{1,2}=\frac{\sqrt{2}}{2} e^{\xi}$.
$\textbf{(a)(d)}$: the structures of the twice-iterated solutions,
$\textbf{(b)(e)}$: the density plot,
$\textbf{(c)(f)}$: the wave propagation of the twice-iterated solutions.}  \\

\subsection{ $Q$ is taken as 4$\times$4 form}
In this subsection, we consider the special case for potential matrix as follows
\begin{equation} \label{Q2}
    Q= \begin{bmatrix}
        u_1(x,t)  &  u_2(x,t) &u_3(x,t) & 0\\
         -u_{2}(x,t)  & u_{1}(x,t) &0 & u_3(x,t) \\
         -u_{3}(x,t) & 0 & u_{1}(x,t) & -u_2(x,t) \\
          0  & -u_{3}(x,t) &u_2(x,t) & u_1(x,t)
        \end{bmatrix},
\end{equation}
here  $R=-Q^{\dagger}$, $p=q=4$ and  $Q$ is a real-valued matrix.
It is a direct calculation to verify that the three-component modified KdV equations of system \eqref{cmkdv} can be obtained by substituting Eq. \eqref{Q2} into Eq. \eqref{mmkdv}, which is also studied in \cite{zhang2008lax}.

Note that, besides $R=-Q^{\dagger}$, another symmetry $Q^{*}=Q$ should be considered. Based on above symmetry, we can easily get that
\begin{equation} \label{ssy1}
    \Psi_{i}^{*}(-\zeta^{*})= \Psi_{i}(\zeta), \qquad i=1,2.
\end{equation}

Moreover,
\begin{equation}
   S^{*}(-\zeta^{*})=S(\zeta),
\end{equation}
equivalently,
\begin{equation}
    \left\{
          \begin{array}{lll}
                            s_{i,j}(-\zeta^{*})=s_{i,j}(\zeta), & \qquad  1\leq i \leq p, \quad  1\leq j \leq p  &\qquad  \zeta \in \mathbb{C}^{+},  \\
                             s_{i,j}(-\zeta)=s_{i,j}(\zeta), & \qquad  1\leq i \leq p,  \quad  p\leq j \leq p+q,   &\qquad  \zeta \in \mathbb{R},  \\
                            s_{i,j}(-\zeta)=s_{i,j}(\zeta), & \qquad  p\leq i \leq p+q,  \quad  1\leq j \leq p,   &\qquad  \zeta \in \mathbb{R},  \\
                             s_{i,j}(-\zeta^{*})=s_{i,j}(\zeta), & \qquad  p+1\leq i \leq p+q, \quad  p+1\leq j \leq p+q,  &\qquad  \zeta \in \mathbb{C}^{-}.  \\
          \end{array}
    \right.
\end{equation}
Futher more, we obtain
\begin{equation}
     \det(\Gamma_{1}(-\zeta^{*}))^{*} =\det(\Gamma_{1}(\zeta)), \qquad \zeta \in \mathbb{C}^{+}.
\end{equation}

Since  $\det(\Gamma_{1}(-\zeta^{*}))^{*} =\det(\Gamma_{1}(\zeta)), \zeta \in \mathbb{C}^{+}$ and $ \det(\Gamma_{1}(\zeta^{*}))=(\det(\Gamma_{2}(\zeta)))^{*}$,   $\zeta \in \mathbb{C}^{-}$, we can obtain the relationship of zeros of $\Gamma_{1}(\zeta)$ and  $\Gamma_{2}(\zeta)$.    Based on these results, we consider the following three cases:  \\
\textbf{Case 1:  $\det(\Gamma_1(\zeta))$ has 2$\Delta_1$ simple zeros $\zeta_{j}$($1 \leq j \leq  2\Delta_1$) in $\mathbb{C}^{+}$, and $\zeta_{j}=-\zeta_{j-\Delta_1}^{*},(\Delta_1 +1 \leq j \leq  2\Delta_1)$.}

  Considering the Eq. \eqref{time} and Eq. \eqref{ssy1}, we have

 \begin{equation}\label{var1}
\left\{
\begin{array}{cc}
   \hat{\vartheta}_{j} = \vartheta_{j}^{\dagger},  & \quad    1 \leq j \leq  2 \Delta_1,  \\
   \vartheta_{j}=\vartheta_{j-\Delta_1}^{*},    & \quad  \Delta_1 + 1 \leq j \leq  2 \Delta_1.
\end{array} \right.
 \end{equation}
From  Eqs. \eqref{tt} and  \eqref{var1}, we can obtain

\begin{equation}
\vartheta_{j}=
\left\{
\begin{array}{cc}
     e^{\theta_{j}\sigma} \vartheta_{j,0}, & \quad    1 \leq j \leq   \Delta_1,  \\
    e^{\theta_{j-\Delta_1}^{*} \sigma} \vartheta_{j-\Delta_1,0}^{*},  & \quad  \Delta_1 + 1 \leq j \leq  2 \Delta_1,
\end{array} \right.
 \end{equation}

\begin{equation}
\hat{\vartheta}_{j}=
\left\{
\begin{array}{cc}
       \vartheta_{j,0}^{\dagger} e^{\theta_{j}^{*}\sigma},  & \quad    1 \leq j \leq   \Delta_1,  \\
          \vartheta_{j-\Delta_1,0}^{T}  e^{\theta_{j-\Delta_1} \sigma},  & \quad  \Delta_1 + 1 \leq j \leq  2 \Delta_1.
\end{array} \right.
 \end{equation}
Moreover, the solutions of RH problem in Eq. \eqref{solu} can be rewritten as follows:
\begin{equation}
\left\{
  \begin{aligned}
       \Gamma_{1}(\zeta) = \mathbb{I} -\sum_{k=1}^{2 \Delta_1} \sum_{j=1}^{2 \Delta_1} \frac{\vartheta_{k}\widehat{\vartheta}_{j}(M^{-1})_{k,j}}{\zeta-\zeta_{j}^{*}},   \\
       \Gamma_{2}(\zeta) = \mathbb{I} + \sum_{k=1}^{2 \Delta_1} \sum_{j=1}^{2 \Delta_1} \frac{\vartheta_{k}\widehat{\vartheta}_{j}(M^{-1})_{k,j}}{\zeta-\zeta_{k}}.
  \end{aligned}
\right.
\end{equation}
\textbf{Case 2: $\det(\Gamma_1(\zeta))$ has $\Delta_2$ simple zeros $\zeta_{j}$($1 \leq j \leq  \Delta_2$) in $\mathbb{C}^{+}$, and each $\zeta_{j}$ is pure imaginary.}

 Similarly, $\vartheta_{j}$ and $\hat{\vartheta_{j}}$ satisfy:
\begin{equation}
  \begin{aligned}
       \vartheta_{j} & = e^{\theta_{j}\sigma} \vartheta_{j,0},  \qquad \hat{\vartheta_{j}}=\vartheta_{j}^{\dagger}, \qquad 1 \leq j \leq  \Delta_2.
  \end{aligned}
\end{equation}
Futhermore,

\begin{equation}
\left\{
\begin{array}{ccc}
      \vartheta_{j} & = e^{\theta_{j}\sigma} \vartheta_{j,0}, & \qquad 1 \leq j \leq  \Delta_2,\\
       \hat{\vartheta}_{j} & = \vartheta_{j,0}^{\dagger} e^{\theta_{j}^{*}\sigma}, &  \qquad 1 \leq j \leq  \Delta_2.
\end{array} \right.
 \end{equation}
Moreover, the solutions of RH problem in Eq. \eqref{solu}  can be rewritten as follows
\begin{equation}
\left\{
  \begin{aligned}
       \Gamma_{1}(\zeta) = \mathbb{I} -\sum_{k=1}^{ \Delta_2} \sum_{j=1}^{\Delta_2} \frac{\vartheta_{k}\widehat{\vartheta}_{j}(M^{-1})_{k,j}}{\zeta-\zeta_{j}^{*}},   \\
       \Gamma_{2}(\zeta) = \mathbb{I} + \sum_{k=1}^{\Delta_2} \sum_{j=1}^{\Delta_2} \frac{\vartheta_{k}\widehat{\vartheta}_{j}(M^{-1})_{k,j}}{\zeta-\zeta_{k}}.
  \end{aligned}
\right.
\end{equation}
\textbf{Case 3: $\det(\Gamma_1(\zeta))$ has $2\Delta_1 + \Delta_2$ simple zeros in $\mathbb{C}^{+}$, where the first $2 \Delta_1$ zeros satisfy $\zeta_{\Delta_1+j}=-\zeta_{j}^{*}$, $1 \leq j \leq \Delta_1$,  and $\zeta_{j}, (2 \Delta_1 +1  \leq j \leq  2 \Delta_1 +\Delta_2)$ are pure imaginary.}

 Similarly,  $\vartheta_{j}$ and $\hat{\vartheta_{j}}$ satisfy
\begin{equation}
\left\{
  \begin{aligned}
     \vartheta_{j} =& e^{\theta_{j}} \vartheta_{j,0}, & \quad  1 \leq j \leq  2 \Delta_1 + \Delta_2, \\
      \vartheta_{j}=& \vartheta_{j-\Delta_1}^{*},    & \quad  \Delta_1 + 1 \leq j \leq  2 \Delta_1, \\
   \hat{\vartheta}_{j} =& \vartheta_{j}^{\dagger},  & \quad    1 \leq j \leq  2 \Delta_1 + \Delta_2.  \\
  \end{aligned}
\right.
\end{equation}
Futhermore,
\begin{equation}
\vartheta_{j}=
\left\{
\begin{array}{cc}
        e^{\theta_{j}\sigma} \vartheta_{j,0}, & \qquad 1 \leq j \leq  \Delta_1,\\
         e^{\theta_{j-\Delta_1}^{*}\sigma} \vartheta_{j-\Delta_1,0}^{*}, & \qquad  \Delta_1 + 1 \leq j \leq  2 \Delta_1, \\
       e^{\theta_{j}\sigma} \vartheta_{j,0}, & \qquad  2 \Delta_1 + 1 \leq j \leq  2 \Delta_1 + \Delta_2,
\end{array} \right.
 \end{equation}

\begin{equation}
\hat{\vartheta}_{j}=
\left\{
\begin{array}{cc}
      \vartheta_{j,0}^{\dagger}  e^{\theta_{j}^{*}\sigma},  & \qquad 1 \leq j \leq  \Delta_1,\\
       \vartheta_{j-\Delta_1,0}^{T}  e^{\theta_{j-\Delta_1} \sigma}, & \qquad  \Delta_1 + 1 \leq j \leq  2 \Delta_1, \\
       \vartheta_{j,0}^{\dagger}  e^{\theta_{j}^{*}\sigma}, & \qquad  2 \Delta_1 + 1 \leq j \leq  2 \Delta_1 + \Delta_2.
\end{array} \right.
 \end{equation}
Moreover, the solutions of RH problem in Eq. \eqref{solu}  can be rewritten as follows:
\begin{equation}
\left\{
  \begin{aligned}
       \Gamma_{1}(\zeta) = \mathbb{I} -\sum_{k=1}^{2 \Delta_1 + \Delta_2} \sum_{j=1}^{2 \Delta_1 + \Delta_2} \frac{\vartheta_{k}\widehat{\vartheta}_{j}(M^{-1})_{k,j}}{\zeta-\zeta_{j}^{*}},   \\
       \Gamma_{2}(\zeta) = \mathbb{I} + \sum_{k=1}^{2 \Delta_1 + \Delta_2} \sum_{j=1}^{2 \Delta_1 + \Delta_2} \frac{\vartheta_{k}\widehat{\vartheta}_{j}(M^{-1})_{k,j}}{\zeta-\zeta_{k}}.
  \end{aligned}
\right.
\end{equation}

\subsubsection{Multi-soliton solutions}
For case 1, we suppose that
\begin{equation}
  \vartheta_{j}=(\alpha_{1,j}e^{\theta_j},\dots,\alpha_{4,j}e^{\theta_j},\alpha_{5,j}e^{-\theta_j},\dots,\alpha_{8,j}e^{-\theta_j})^{T}, \qquad  1 \leq j \leq \Delta_1,
\end{equation}
naturally,
\begin{equation}
  \vartheta_{\Delta_1+j}=(\alpha_{1,j}^{*}e^{\theta_j^{*}},\dots,\alpha_{4,j}^{*}e^{\theta_j^{*}},\alpha_{5,j}^{*}e^{-\theta_j^{*}},\dots,\alpha_{8,j}^{*}e^{-\theta_j^{*}})^{T}, \qquad  1 \leq j \leq \Delta_1,
\end{equation}
then  $\Delta_1$-breather solution are as follows
\begin{equation}
    u_{k}=2 i
    \frac{\det(F_{k+4})}{\det(M)},
\end{equation}
where
\begin{equation}
F_{k+4}=
 \begin{bmatrix}
   0 &  \beta_{1}^{T}   \\
   \beta_{k+4} & M
 \end{bmatrix}, \qquad  1 \leq k \leq 3,
 \end{equation}

\begin{equation}
M   =
\begin{bmatrix}
M_{1,1} & \dots & M_{1,\Delta_1} & M_{1,\Delta_1 + 1} & \dots & M_{1,2 \Delta_1 }  \\
 \vdots   &    \ddots  &  \vdots  &  \vdots  & \ddots & \vdots  \\
 M_{\Delta_1,1}  &    \ddots  &  M_{\Delta_1,\Delta_1}  &   M_{\Delta_1,\Delta_1+1}  & \dots & M_{\Delta_1,2 \Delta_1}  \\
 M_{\Delta_1+1,1} & \dots & M_{\Delta_1+1,\Delta_1} & M_{\Delta_1+1,\Delta_1 + 1} & \dots & M_{\Delta_1+1,2 \Delta_1 }  \\
\vdots   &    \ddots  &  \vdots  &  \vdots  & \ddots & \vdots  \\
 M_{2 \Delta_1,1}  &    \ddots  &  M_{2 \Delta_1,\Delta_1}  &   M_{2\Delta_1,\Delta_1+1}  & \dots & M_{2\Delta_1,2 \Delta_1}
\end{bmatrix},\\
\end{equation}
with
\begin{equation}
  \begin{aligned}
 \beta_{1} &=(\alpha_{1,1}e^{\theta_{1}}, \dots , \alpha_{1,\Delta_1}e^{\theta_{\Delta_1}},\alpha_{1,\Delta_1+1 }e^{\theta_{\Delta_1+1}}, \dots, \alpha_{1,2\Delta_1}e^{\theta_{2\Delta_1}} )^{T}    \\
  &=(\alpha_{1,1}e^{\theta_{1}},\dots , \alpha_{1,\Delta_1}e^{\theta_{\Delta_1}},\alpha_{1,1}^{*}e^{\theta_{1}^{*}}, \dots, \alpha_{1,\Delta_1}^{*}e^{\theta_{\Delta_1}^{*}} )^{T},  \\
 \beta_{k+4}& =(\alpha_{k+4,1}^{*}e^{-\theta_{1}^{*}},\dots , \alpha_{k+4,\Delta_1}^{*}e^{-\theta_{\Delta_1}^{*}},\alpha_{k+4,\Delta_1+1 }^{*}e^{-\theta_{\Delta_1+1}^{*}}, \dots, \alpha_{k+4,2\Delta_1}^{*}e^{-\theta_{2\Delta_1}^{*}} )^{T}   \\
  & =(\alpha_{k+4,1}^{*}e^{-\theta_{1}^{*}},\dots , \alpha_{k+4,\Delta_1}^{*}e^{-\theta_{\Delta_1}^{*}},\alpha_{k+4,1}e^{-\theta_{1}}, \dots, \alpha_{k+4,\Delta_1}e^{-\theta_{\Delta_1}} )^{T}.  \\
   \end{aligned}
\end{equation}

For case 2, we set
\begin{equation}
  \vartheta_{j}=(\alpha_{1,j}e^{\theta_j},\dots,\alpha_{4,j}e^{\theta_j},\alpha_{5,j}e^{-\theta_j},\dots,\alpha_{8,j}e^{-\theta_j})^{T}, \qquad  1 \leq j \leq \Delta_2,
\end{equation}
then we can obtain $\Delta_2$-bell solutions
\begin{equation}
    u_{k}=2 i
    \frac{\det(F_{k+4})}{\det(M)},
\end{equation}
where
\begin{equation}
\begin{aligned}
F_{k+4}=
 \begin{bmatrix}
   0 &  \beta_{1}^{T}   \\
   \beta_{k+4} & M
 \end{bmatrix},
 \qquad   1\leq k \leq 3,
\end{aligned}
 \end{equation}
\begin{equation}
M=
\begin{bmatrix}
  M_{1,1}  & M_{1,2}  & \dots & M_{1,\Delta_2}  \\
  M_{2,1} & M_{2,2}  &  \dots & M_{2,\Delta_2}  \\
  \vdots  &  \vdots  &  \ddots  & \vdots \\
   M_{\Delta_2,1} & M_{\Delta_2,2}  & \dots & M_{\Delta_2,\Delta_2},
\end{bmatrix},
\end{equation}
\begin{equation}
  \begin{aligned}
  \beta_{1}&=(\alpha_{1,1}e^{\theta_{1}},\alpha_{1,2}e^{\theta_{2}},\dots , \alpha_{1,\Delta_2}e^{\theta_{\Delta_2}})^{T},\\
  \beta_{k+4} &=(\alpha_{k+4,1}^{*}e^{-\theta_{1}^{*}},\alpha_{k+4,2}^{*}e^{-\theta_{2}^{*}},\dots,\alpha_{k+4,\Delta_2}^{*}e^{-\theta_{\Delta_2}^{*}})^{T}.
   \end{aligned}
\end{equation}

For case 3, combining the results of case 1 and  case 2, we can easily  obtain the $2\Delta_1 + \Delta_2$ soliton solutions.

\subsubsection{A variety of Rational Solutions and Physical Visions}

 For case 1, if we take   $\Delta_1=1$, we can obtain  single-breather solutions as follows
 \begin{equation} \label{sf2.1}
\left\{
\begin{aligned}
   u_1(x,t)=2  i \frac{    \begin{vmatrix}
      0  & \alpha_{1,1}e^{\theta_{1}}  &  \alpha_{1,1}^{*}e^{\theta_{1}^{*}}  \\
      \alpha_{5,1}^{*}e^{-\theta_{1}^{*}} & M_{1,1}   & M_{1,2}      \\
      \alpha_{5,1}e^{-\theta_{1}}  & M_{2,1}  & M_{2,2}
   \end{vmatrix}  }
   {     \begin{vmatrix}   M_{1,1} &M_{1,2}    \\    M_{2,1}  & M_{2,2}                 \end{vmatrix}   },\\
 u_2(x,t)=2  i \frac{    \begin{vmatrix}
      0  & \alpha_{1,1}e^{\theta_{1}}  &  \alpha_{1,1}^{*}e^{\theta_{1}^{*}}  \\
      \alpha_{6,1}^{*}e^{-\theta_{1}^{*}} & M_{1,1}   & M_{1,2}      \\
      \alpha_{6,1}e^{-\theta_{1}}  & M_{2,1}  & M_{2,2}
   \end{vmatrix}  }
   {         \begin{vmatrix}   M_{1,1} &M_{1,2}    \\    M_{2,1}  & M_{2,2}                 \end{vmatrix}  },\\
u_3(x,t)=2  i \frac{     \begin{vmatrix}
      0  & \alpha_{1,1}e^{\theta_{1}}  &  \alpha_{1,1}^{*}e^{\theta_{1}^{*}}  \\
      \alpha_{7,1}^{*}e^{-\theta_{1}^{*}} & M_{1,1}   & M_{1,2}      \\
      \alpha_{7,1}e^{-\theta_{1}}  & M_{2,1}  & M_{2,2}
   \end{vmatrix}  }
   {   \begin{vmatrix}   M_{1,1} &M_{1,2}    \\    M_{2,1}  & M_{2,2}                 \end{vmatrix}   },
\end{aligned}
\right.
 \end{equation}
where
\begin{equation}
\left\{
   \begin{aligned}
         M_{11}= &  \frac{(|\alpha_{1,1}|^2+|\alpha_{2,1}|^2+|\alpha_{3,1}|^2+|\alpha_{4,1}|^2)e^{\theta_{1}^{*}+\theta_{1}}}{\zeta_{1}-\zeta_{1}^{*}}  \\
           +& \frac{(|\alpha_{5,1}|^2+|\alpha_{6,1}|^2+|\alpha_{7,1}|^2+|\alpha_{8,1}|^2)e^{-\theta_{1}^{*}-\theta_{1}}}{\zeta_{1}-\zeta_{1}^{*}},\\
         M_{12}= &  \frac{(\alpha_{1,1}^{*}\alpha_{1,1}^{*}+\alpha_{2,1}^{*}\alpha_{2,1}^{*}+\alpha_{ 3,1}^{*}\alpha_{3,1}^{*}+\alpha_{4,1}^{*}\alpha_{4,1}^{*})e^{\theta_{1}^{*}+\theta_{1}^{*}}}{\zeta_{2}-\zeta_{1}^{*}} \\
         +&\frac{(\alpha_{ 5,1}^{*}\alpha_{5,1}^{*}+\alpha_{6,1}^{*}\alpha_{6,1}^{*}+\alpha_{ 7,1}^{*}\alpha_{7,1}^{*}+\alpha_{8,1}^{*}\alpha_{8,1}^{*})e^{-\theta_{1}^{*}-\theta_{1}^{*}}}{\zeta_{2}-\zeta_{1}^{*}},  \\
          M_{21}= &  \frac{(\alpha_{1,1}\alpha_{1,1}+\alpha_{2,1}\alpha_{2,1}+\alpha_{3,1}\alpha_{3,1}+\alpha_{4,1}\alpha_{4,1})e^{\theta_{1}+\theta_{1}}}{\zeta_{1}-\zeta_{2}^{*}}   \\
          +& \frac{ (\alpha_{5,1}\alpha_{5,1}+\alpha_{6,1}\alpha_{6,1}+\alpha_{ 7,1}\alpha_{7,1}+\alpha_{8,1}\alpha_{8,1})e^{-\theta_{1}-\theta_{1}}}{\zeta_{1}-\zeta_{2}^{*}},  \\
         M_{22}= &  \frac{(|\alpha_{1,1}|^2+|\alpha_{2,1}|^2+|\alpha_{3,1}|^2+|\alpha_{4,1}|^2)e^{\theta_{1}^{*}+\theta_{1}}}{\zeta_{2}-\zeta_{2}^{*}} \\
        +&\frac{(|\alpha_{5,1}|^2+|\alpha_{6,1}|^2+|\alpha_{7,1}|^2+|\alpha_{8,1}|^2)e^{-\theta_{1}^{*}-\theta_{1}}}{\zeta_{2}-\zeta_{2}^{*}}.
   \end{aligned}
\right.
\end{equation}
with $\zeta_1=a_1+i b_1 (a_1 \neq 0, b_1>0)$, $\theta_1=-i (\zeta_1 x + 4 \zeta_1^3 t)$ and $\zeta_2= -\zeta_1^{*}$. In addition, the localized structures and  dynamic behaviors of single-breather solutions are shown  in Fig. 3.

For case 2, if we take the $\Delta_2=1$, we can obtain  the single-bell soliton solution as follows
\begin{equation} \label{sf2.2}
\left\{
   \begin{aligned}
          u_{1}(x,t)=  & \frac{-2i\alpha_{1,1} \alpha_{5,1}^{*} (\zeta_{1}-\zeta_{1}^{*})e^{\theta_1-\theta_{1}^{*}}}{ \gamma_1 e^{\theta_{1}^{*}+\theta_{1}}+ \gamma_2 e^{-\theta_{1}^{*}-\theta_{1}}},  \\
          u_{2}(x,t)=  & \frac{-2i \alpha_{1,1} \alpha_{6,1}^{*} (\zeta_{1}-\zeta_{1}^{*})e^{\theta_1-\theta_{1}^{*}}}{ \gamma_1 e^{\theta_{1}^{*}+\theta_{1}}+ \gamma_2 e^{-\theta_{1}^{*}-\theta_{1}}},\\
          u_{3}(x,t)=  &\frac{ -2i \alpha_{1,1} \alpha_{7,1}^{*} (\zeta_{1}-\zeta_{1}^{*})e^{\theta_1-\theta_{1}^{*}}}{ \gamma_1 e^{\theta_{1}^{*}+\theta_{1}}+ \gamma_2 e^{-\theta_{1}^{*}-\theta_{1}}},
   \end{aligned}
\right.
\end{equation}
where
\begin{equation}
\begin{aligned}
\gamma_1=&|\alpha_{1,1}|^2+|\alpha_{2,1}|^2+|\alpha_{3,1}|^2+|\alpha_{4,1}|^2, \\
\gamma_2=&|\alpha_{5,1}|^2+|\alpha_{6,1}|^2+|\alpha_{7,1}|^2+|\alpha_{8,1}|^2, \\
\zeta_1=&i b_1(b_1>0),   \theta_1 = -i (\zeta_1 x +4 \zeta_1^3 t).
\end{aligned}
\end{equation}

\noindent
{\rotatebox{0}{\includegraphics[width=3.6cm,height=3.0cm,angle=0]{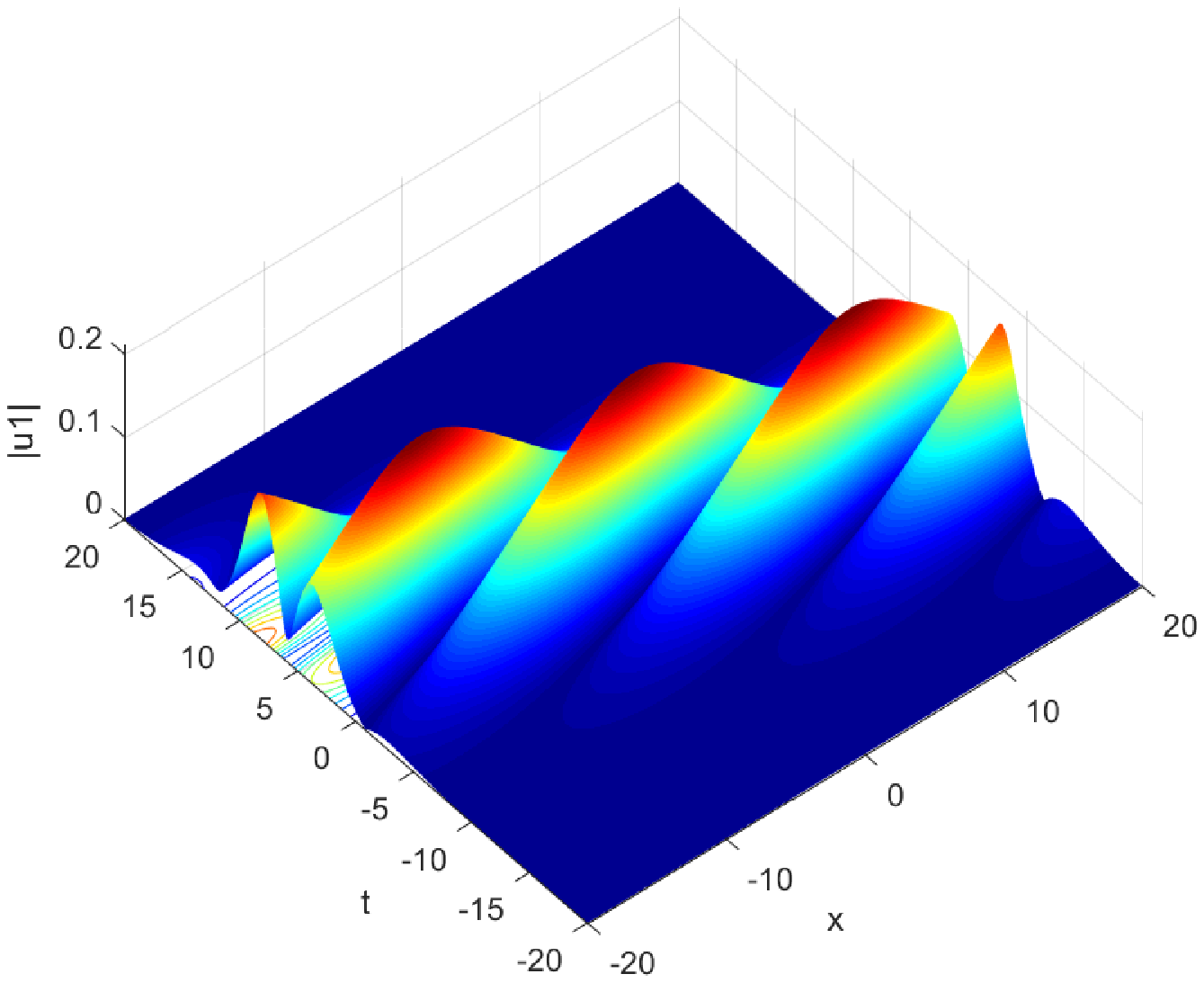}}}
~~~~
{\rotatebox{0}{\includegraphics[width=3.6cm,height=3.0cm,angle=0]{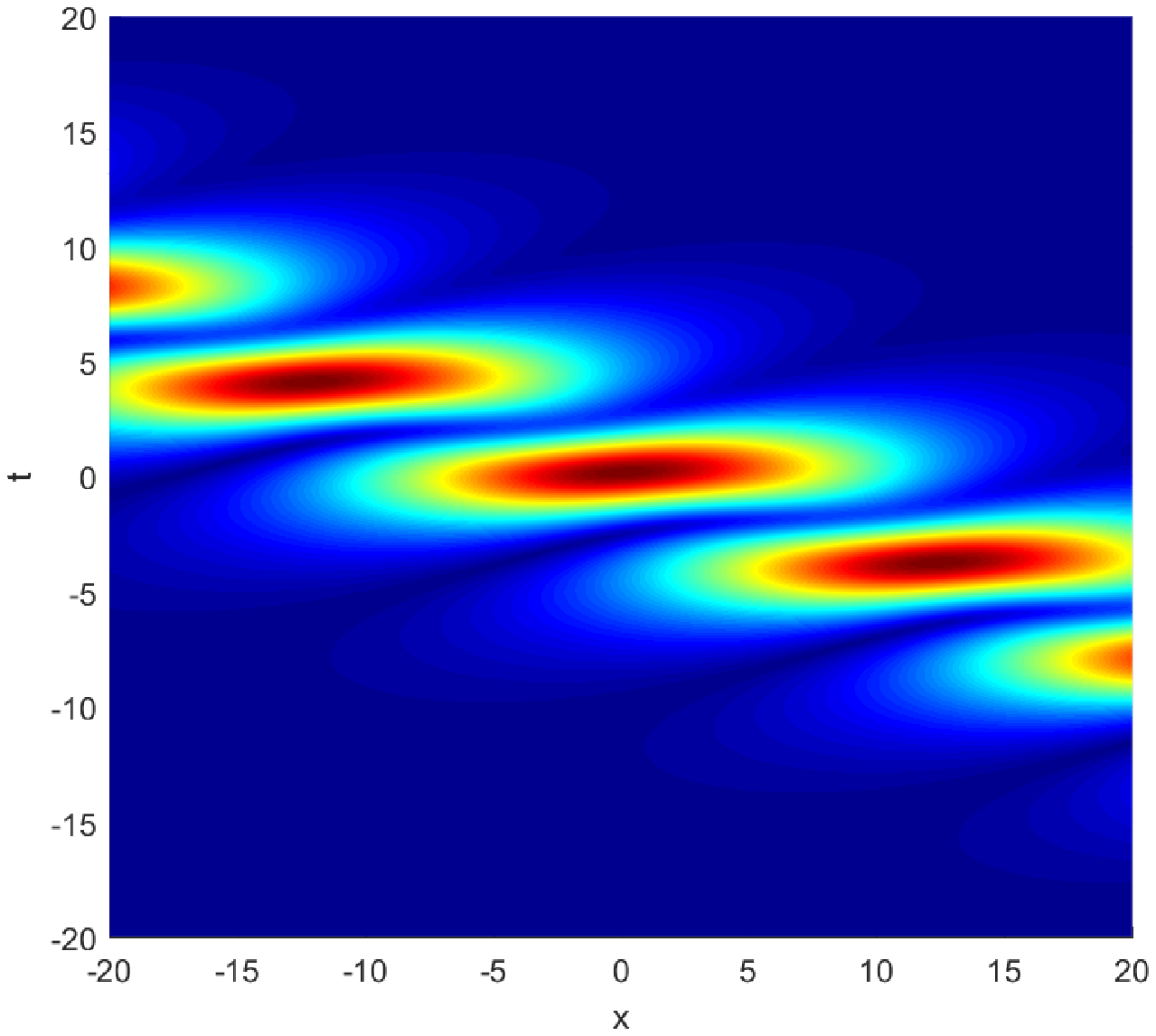}}}
~~~~
{\rotatebox{0}{\includegraphics[width=3.6cm,height=3.0cm,angle=0]{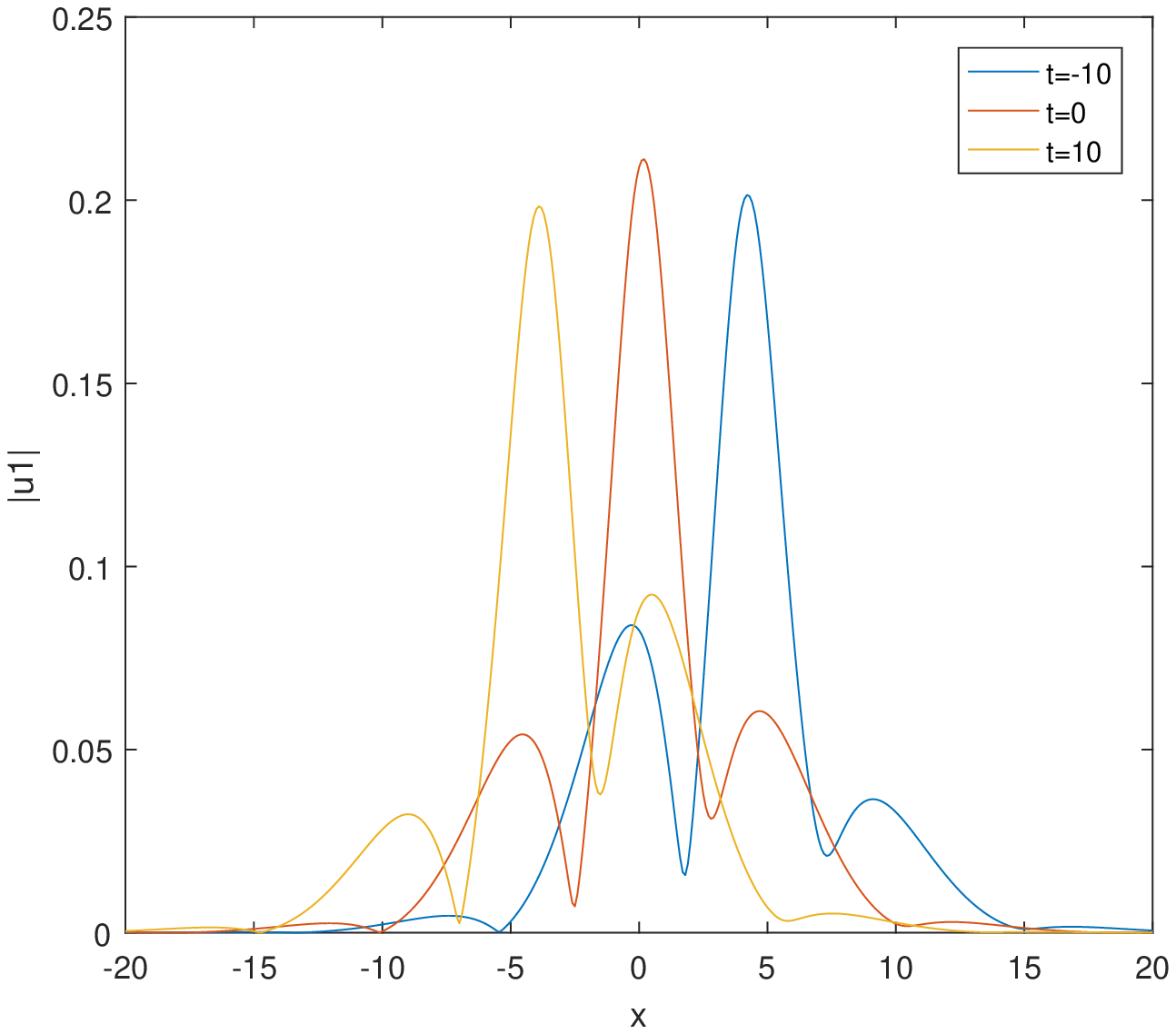}}}

$\ \qquad~~~~~~(\textbf{a})\qquad \ \qquad\qquad\qquad\qquad~(\textbf{b})
\ \qquad\qquad\qquad\qquad\qquad~(\textbf{c})$\\
\noindent
{\rotatebox{0}{\includegraphics[width=3.6cm,height=3.0cm,angle=0]{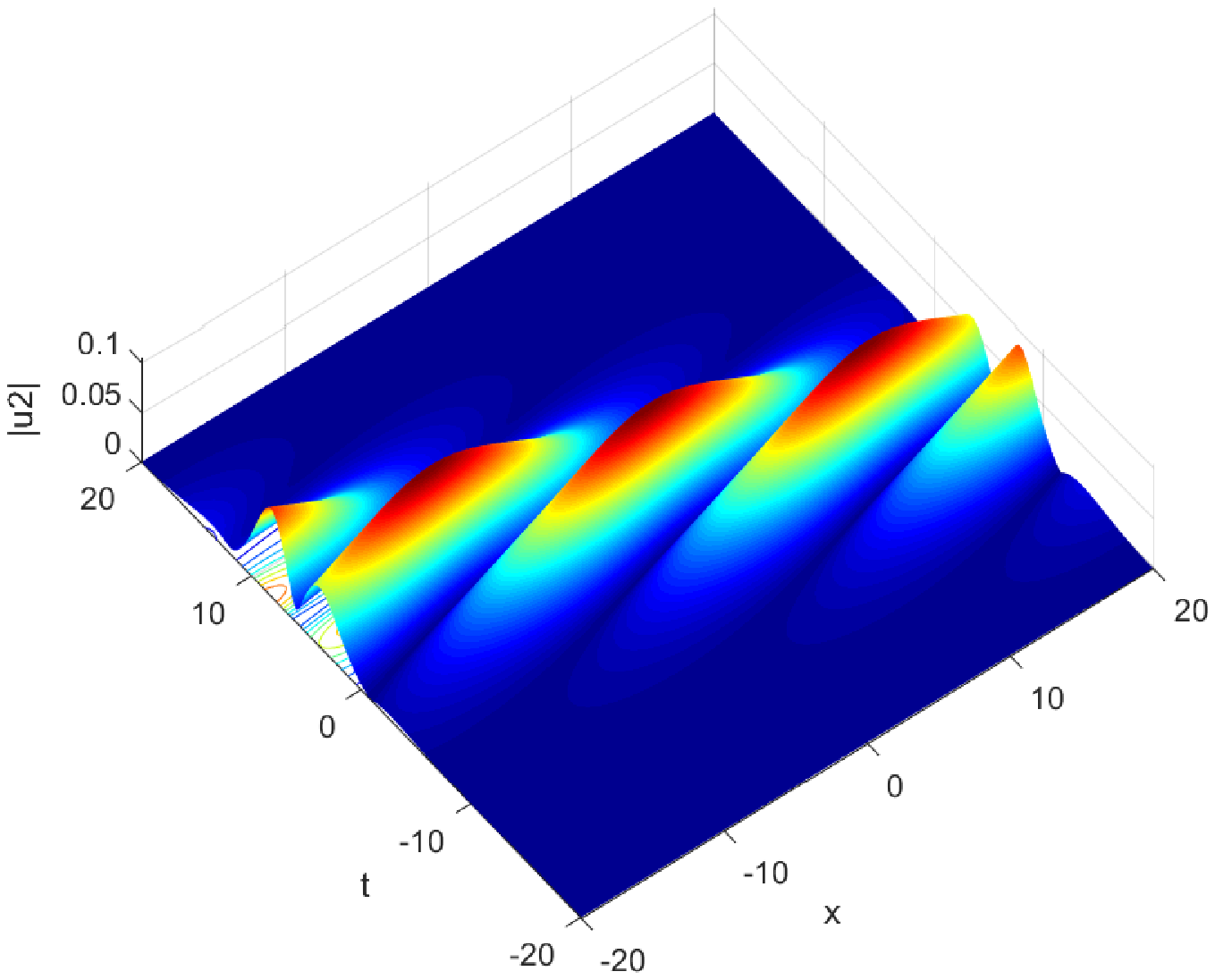}}}
~~~~
{\rotatebox{0}{\includegraphics[width=3.6cm,height=3.0cm,angle=0]{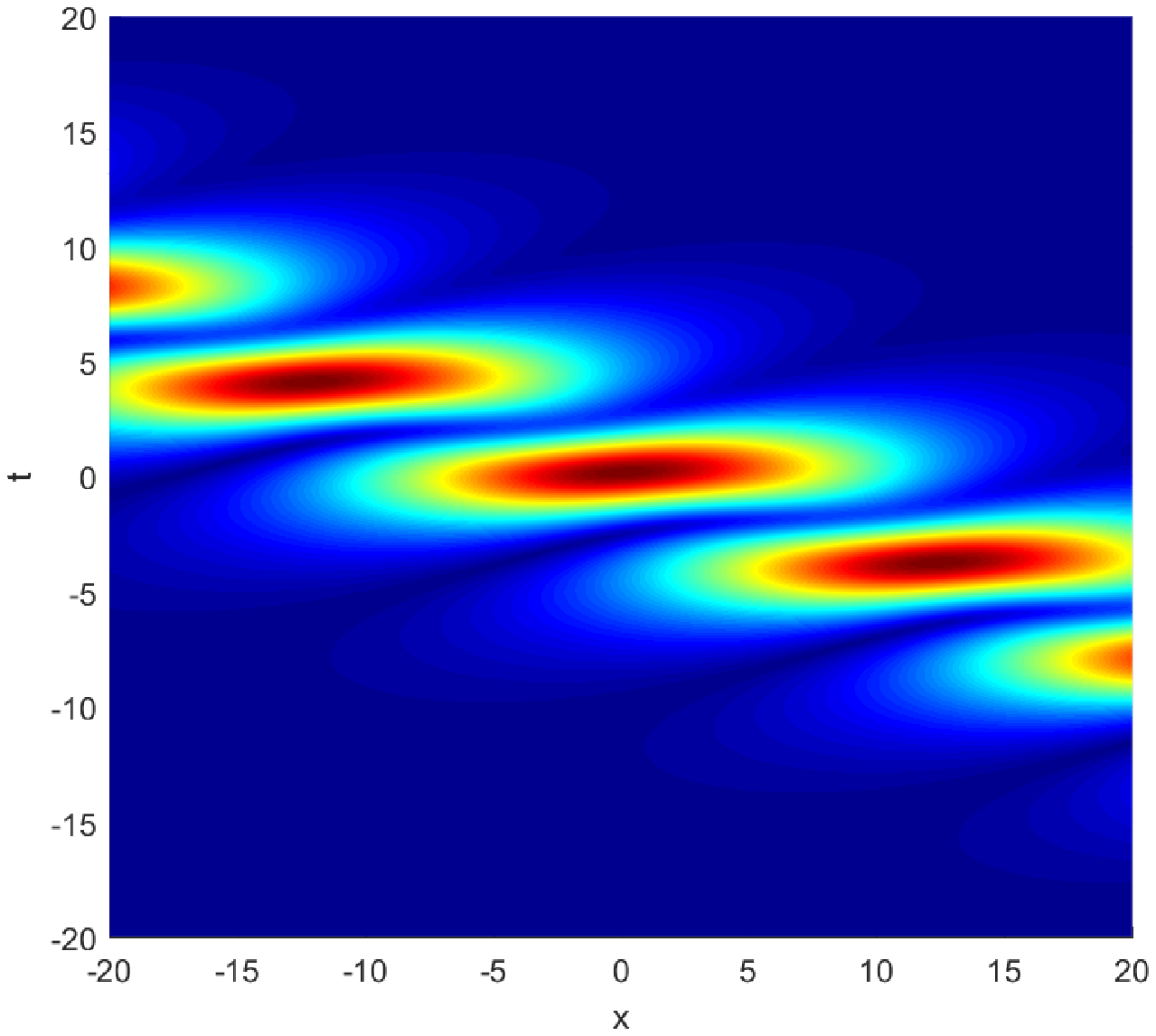}}}
~~~~
{\rotatebox{0}{\includegraphics[width=3.6cm,height=3.0cm,angle=0]{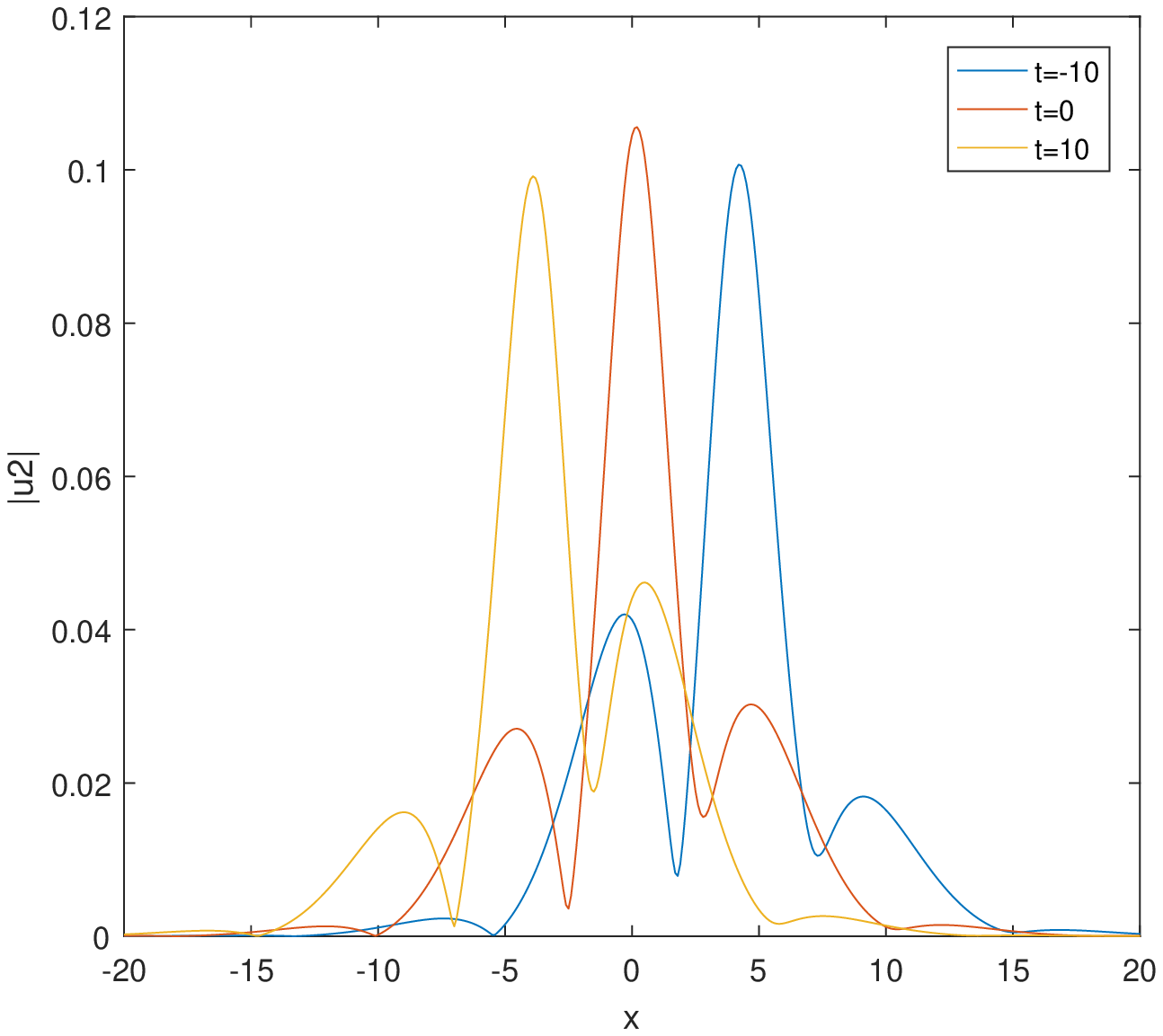}}}

$\ \qquad~~~~~~(\textbf{d})\qquad \ \qquad\qquad\qquad\qquad~(\textbf{e})
\ \qquad\qquad\qquad\qquad\qquad~(\textbf{f})$\\
\noindent
{\rotatebox{0}{\includegraphics[width=3.6cm,height=3.0cm,angle=0]{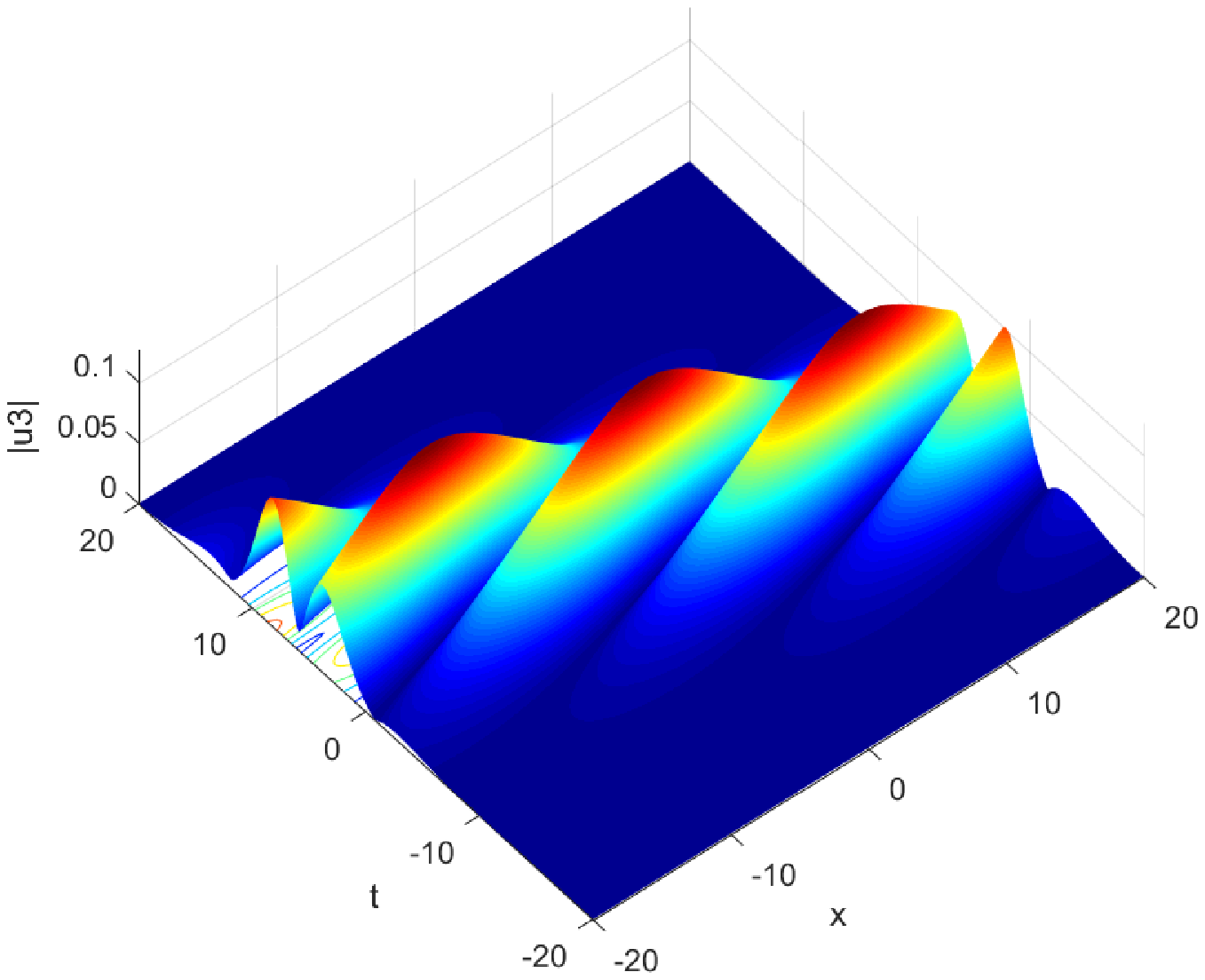}}}
~~~~
{\rotatebox{0}{\includegraphics[width=3.6cm,height=3.0cm,angle=0]{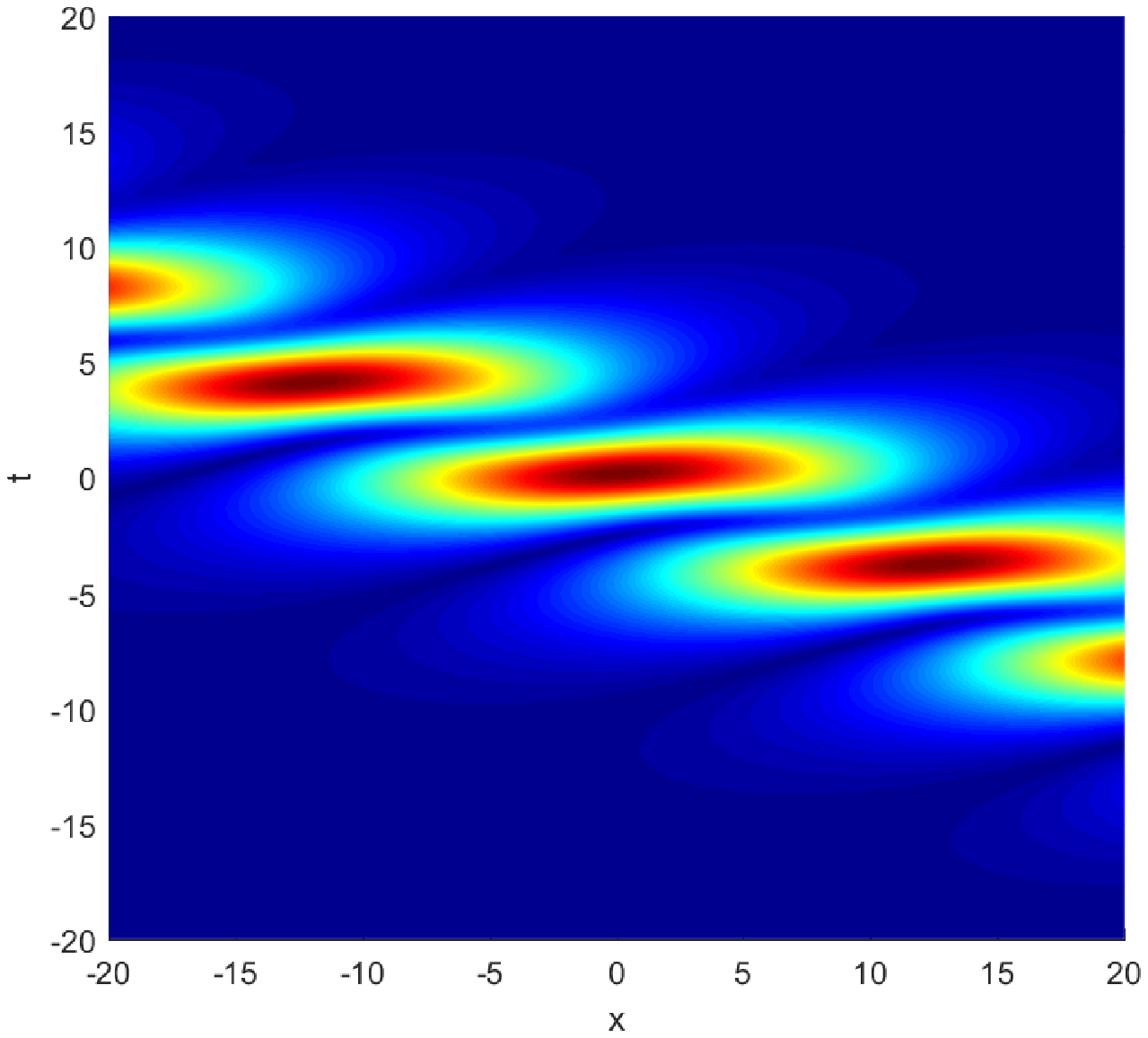}}}
~~~~
{\rotatebox{0}{\includegraphics[width=3.6cm,height=3.0cm,angle=0]{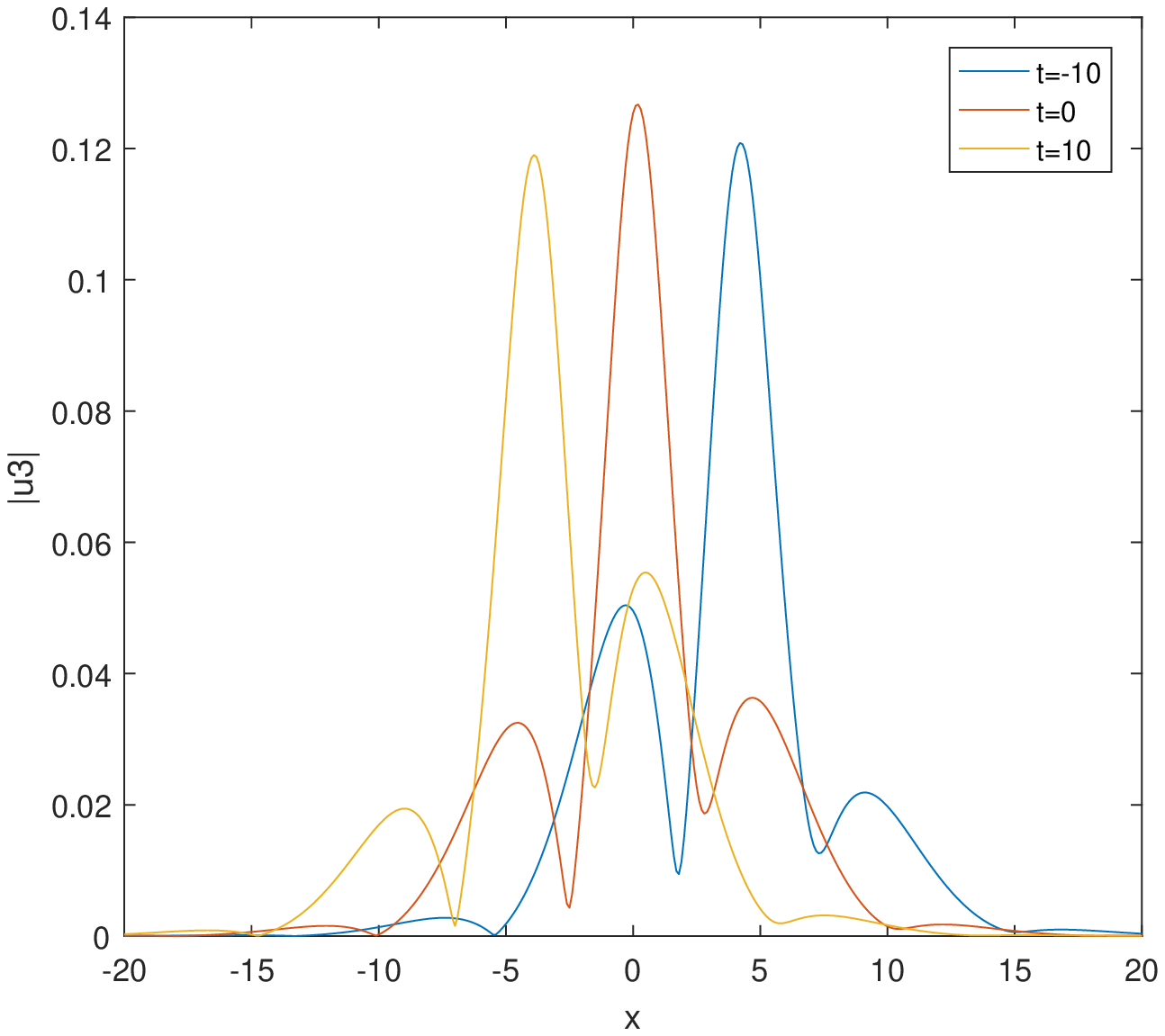}}}

$\ \qquad~~~~~~(\textbf{g})\qquad \ \qquad\qquad\qquad\qquad~(\textbf{h})
\ \qquad\qquad\qquad\qquad\qquad~(\textbf{i})$\\
\noindent { \small \textbf{Figure 3.} (Color online) Single-breather solutions to Eq. \eqref{sf2.1}  with the paremeters $a_1 =0.2$, $b_1=0.2$, $\alpha_{1,1}=\alpha_{3,1}=\alpha_{6,1}=0.1$, $\alpha_{2,1}=\alpha_{5,1}=0.2$, $\alpha_{4,1}=0.05$, $\alpha_{7,1}=0.12$, $\alpha_{8,1}=0.13$.
\label{fig2.1}
$\textbf{(a)(d)(g)}$: the structures of the single-breather solutions,
$\textbf{(b)(e)(h)}$: the density plot ,
$\textbf{(c)(f)(i)}$: the wave propagation of the single-breather solutions.}

From the 3D plot and density plot in Fig. 3, we can see that the waves corresponds to the solutions which travels left along the $x$-axis and their amplitude increases and decreases periodically. Moreover, it can be seen from the dynamic behaviors that the waveforms when time is negative are almost symmetrical with the waveform when time  is positive,  but they are not symmetrical about the central axis of themselves.  At $t=0$, the waveforms are symmetrical about the central axis.

\subsection{ $Q$ is taken as a special 1$\times$3 form}
In this subsection, we consider following special case for matrix $Q$
\begin{equation}
    Q= (u_1(x,t), \quad   u_2(x,t),  \quad   u_3(x,t))^{T},
\end{equation}
where $p=1, q=3$ and $R=-Q^{\dagger}$.
\subsubsection{A variety of Rational Solutions and Physical Visions}
If we set $N=1,2,3$, $p=1$,  and $q=3$ in Eq. \eqref{solution}, we can easily get the one-soliton solutions, two-soliton solutions and three-soliton solutions, respectively,
\begin{equation} \label{sf3.1}
\left\{
  \begin{aligned}
  u_1(x,t)=-2i \sum_{k=1}^{N} \sum_{j=1}^{N} \alpha_{1,k} \alpha_{2,j}^{*} e^{\theta_k -\theta_j^{*}}(M^{-1})_{k,j}, \quad N=1,2,3,  \\
  u_2(x,t)=-2i \sum_{k=1}^{N} \sum_{j=1}^{N} \alpha_{1,k} \alpha_{3,j}^{*} e^{\theta_k -\theta_j^{*}}(M^{-1})_{k,j}, \quad N=1,2,3,  \\
  u_3(x,t)=-2i \sum_{k=1}^{N} \sum_{j=1}^{N} \alpha_{1,k} \alpha_{4,j}^{*} e^{\theta_k -\theta_j^{*}}(M^{-1})_{k,j}, \quad N=1,2,3,
  \end{aligned}
\right.
\end{equation}
where
$
   M_{k,j}= \frac{\alpha_{1,k}^{*}\alpha_{1,j}e^{\theta_{k}^{*}+\theta_j}+\alpha_{2,k}^{*}\alpha_{2,j}e^{-\theta_{k}^{*}-\theta_j}+\alpha_{3,k}^{*}\alpha_{3,j}e^{-\theta_{k}^{*}-\theta_j}+\alpha_{4,k}^{*}\alpha_{4,j}e^{-\theta_{k}^{*}-\theta_j}}{\zeta_j -\zeta_k^{*}}.
$

\noindent
{\rotatebox{0}{\includegraphics[width=3.6cm,height=3.0cm,angle=0]{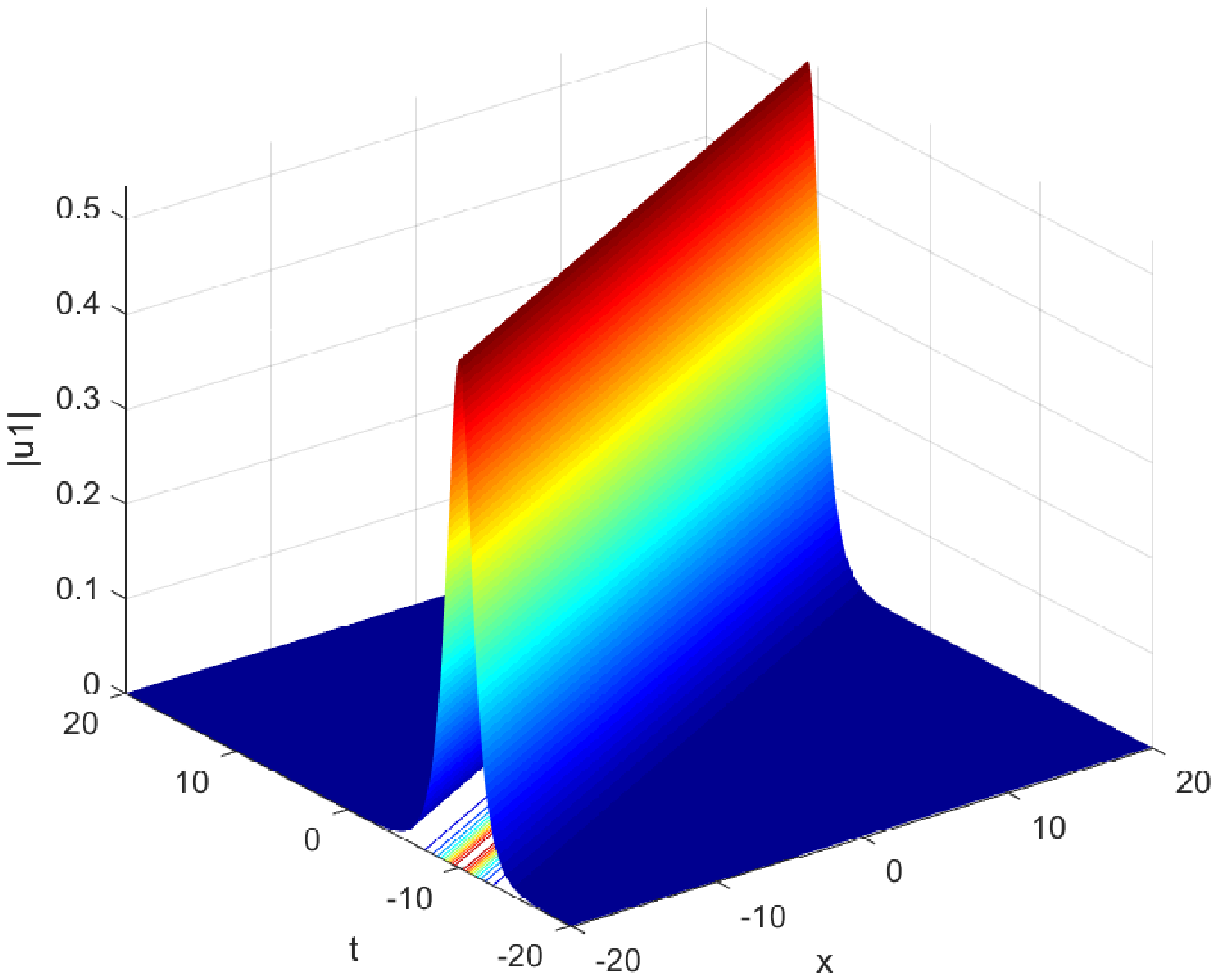}}}
~~~~
{\rotatebox{0}{\includegraphics[width=3.6cm,height=3.0cm,angle=0]{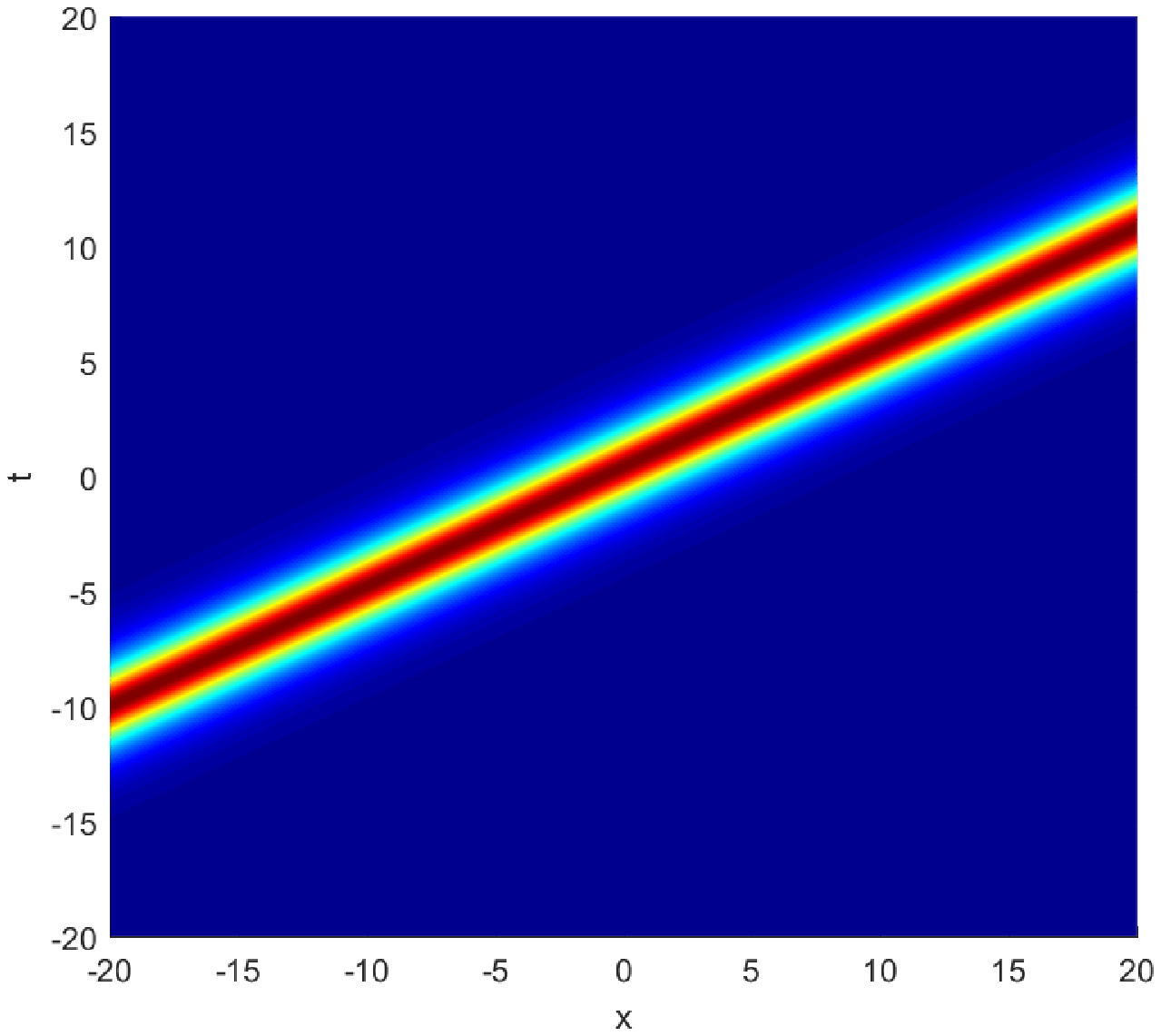}}}
~~~~
{\rotatebox{0}{\includegraphics[width=3.6cm,height=3.0cm,angle=0]{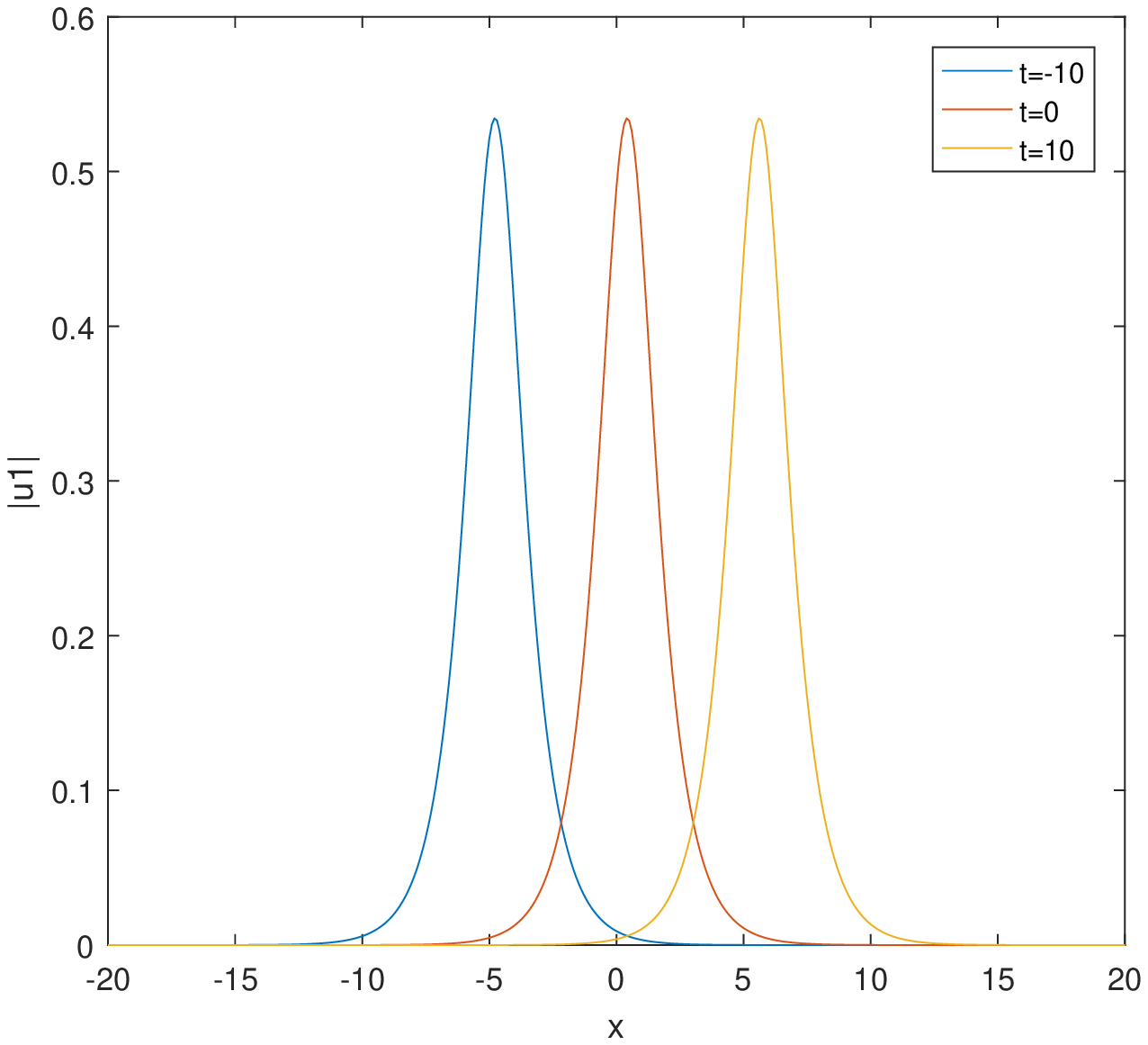}}}

$\ \qquad~~~~~~(\textbf{a})\qquad \ \qquad\qquad\qquad\qquad~(\textbf{b})
\ \qquad\qquad\qquad\qquad\qquad~(\textbf{c})$\\
\noindent
{\rotatebox{0}{\includegraphics[width=3.6cm,height=3.0cm,angle=0]{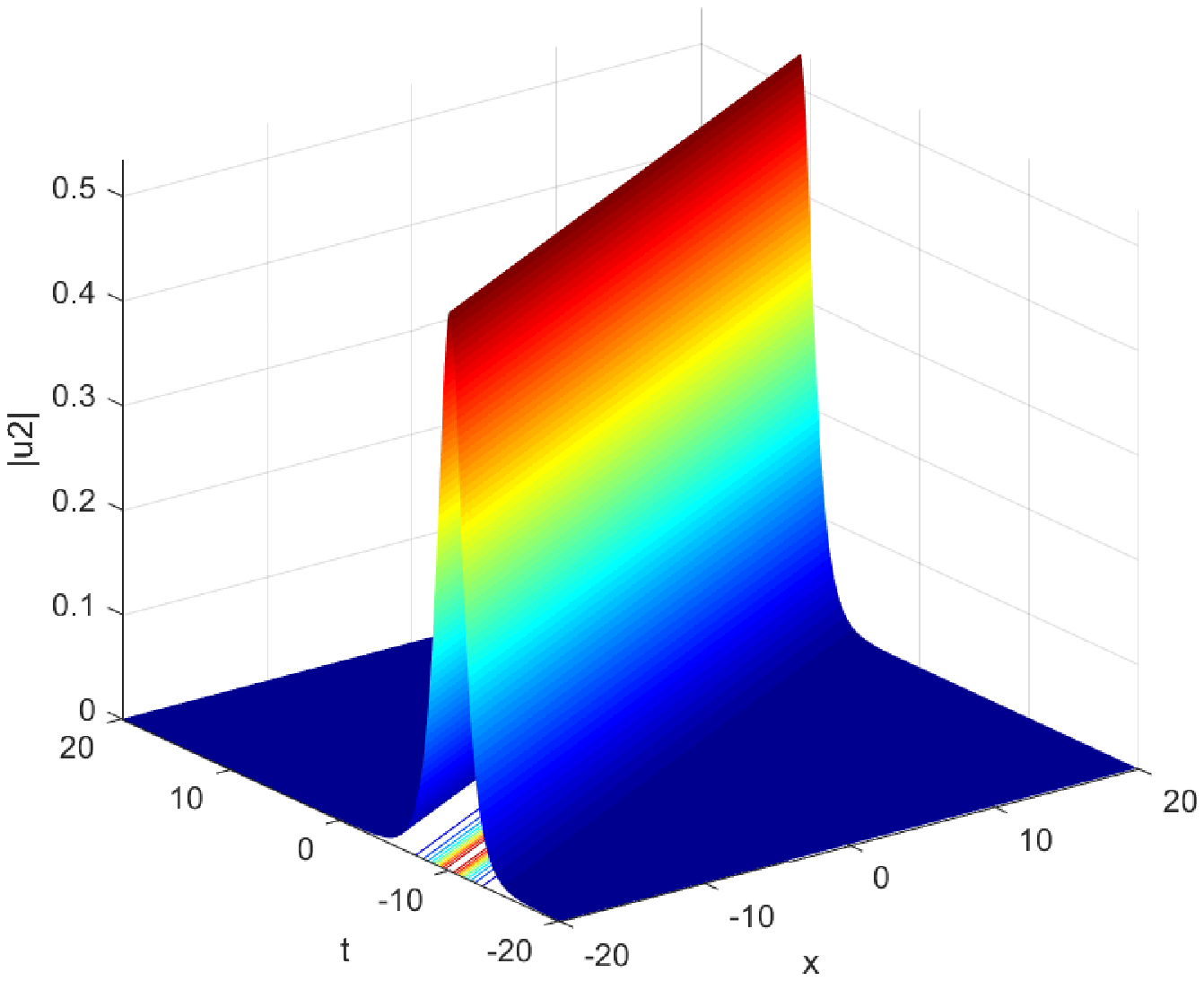}}}
~~~~
{\rotatebox{0}{\includegraphics[width=3.6cm,height=3.0cm,angle=0]{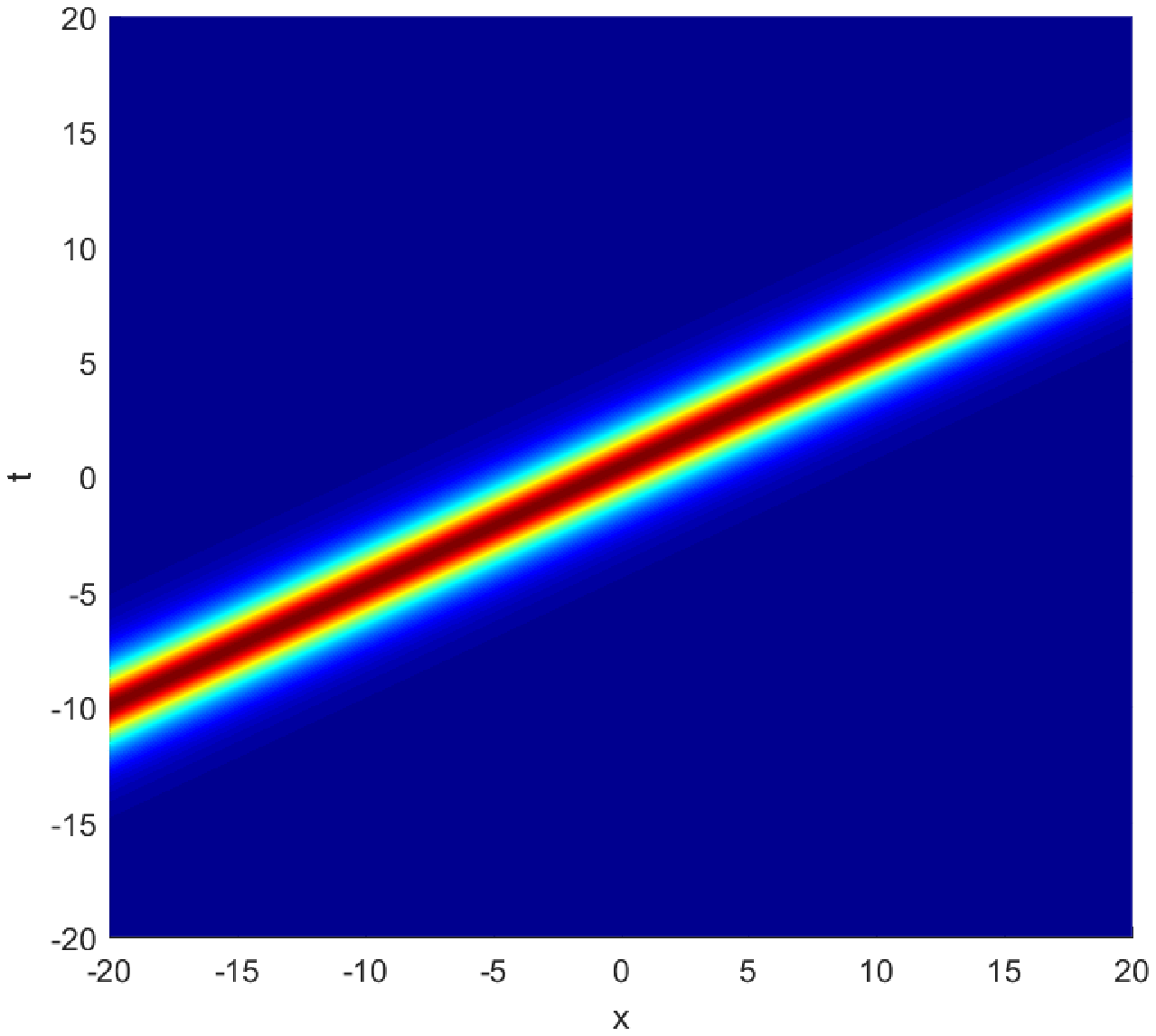}}}
~~~~
{\rotatebox{0}{\includegraphics[width=3.6cm,height=3.0cm,angle=0]{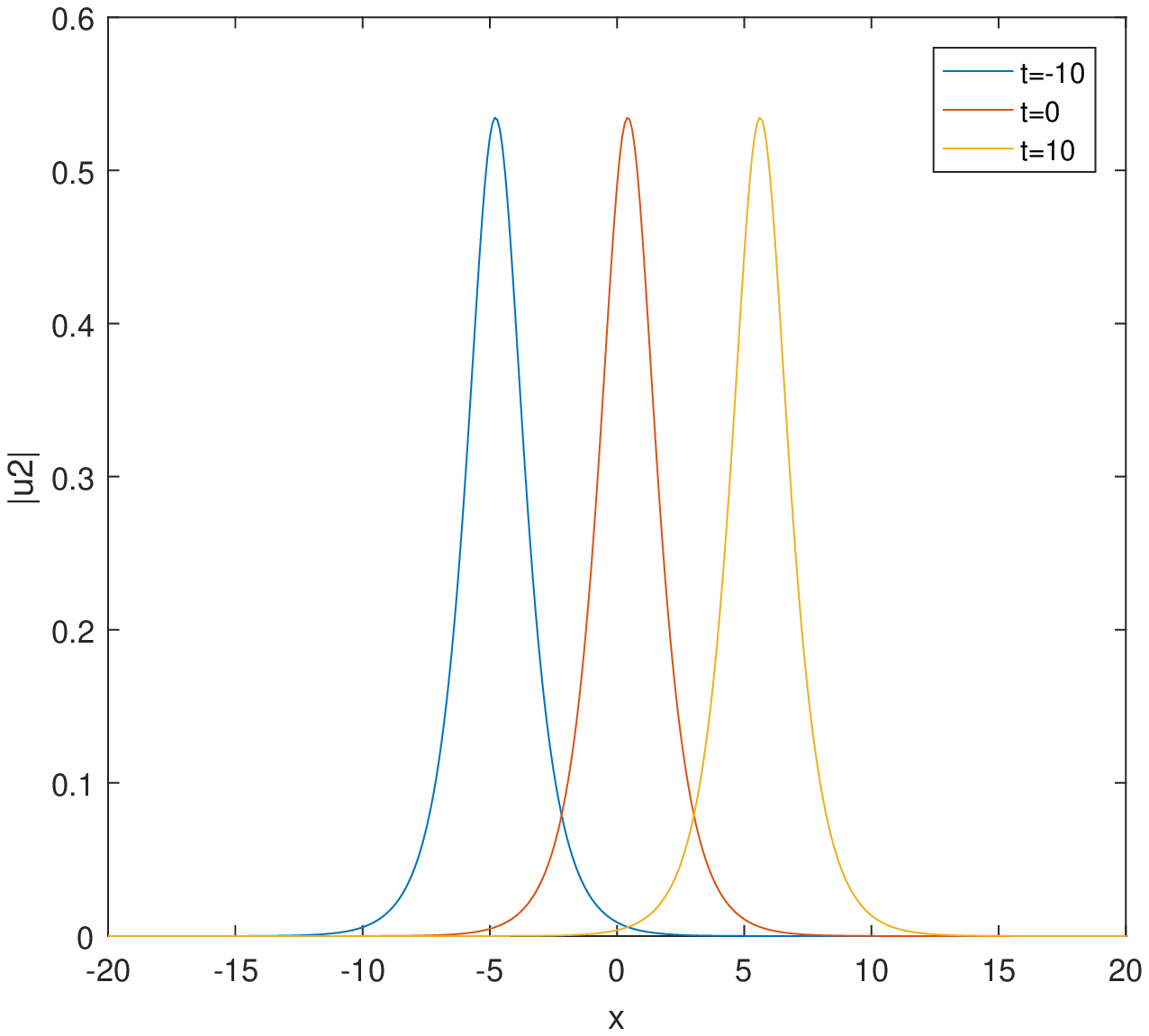}}}

$\ \qquad~~~~~~(\textbf{d})\qquad \ \qquad\qquad\qquad\qquad~(\textbf{e})
\ \qquad\qquad\qquad\qquad\qquad~(\textbf{f})$\\
\noindent
{\rotatebox{0}{\includegraphics[width=3.6cm,height=3.0cm,angle=0]{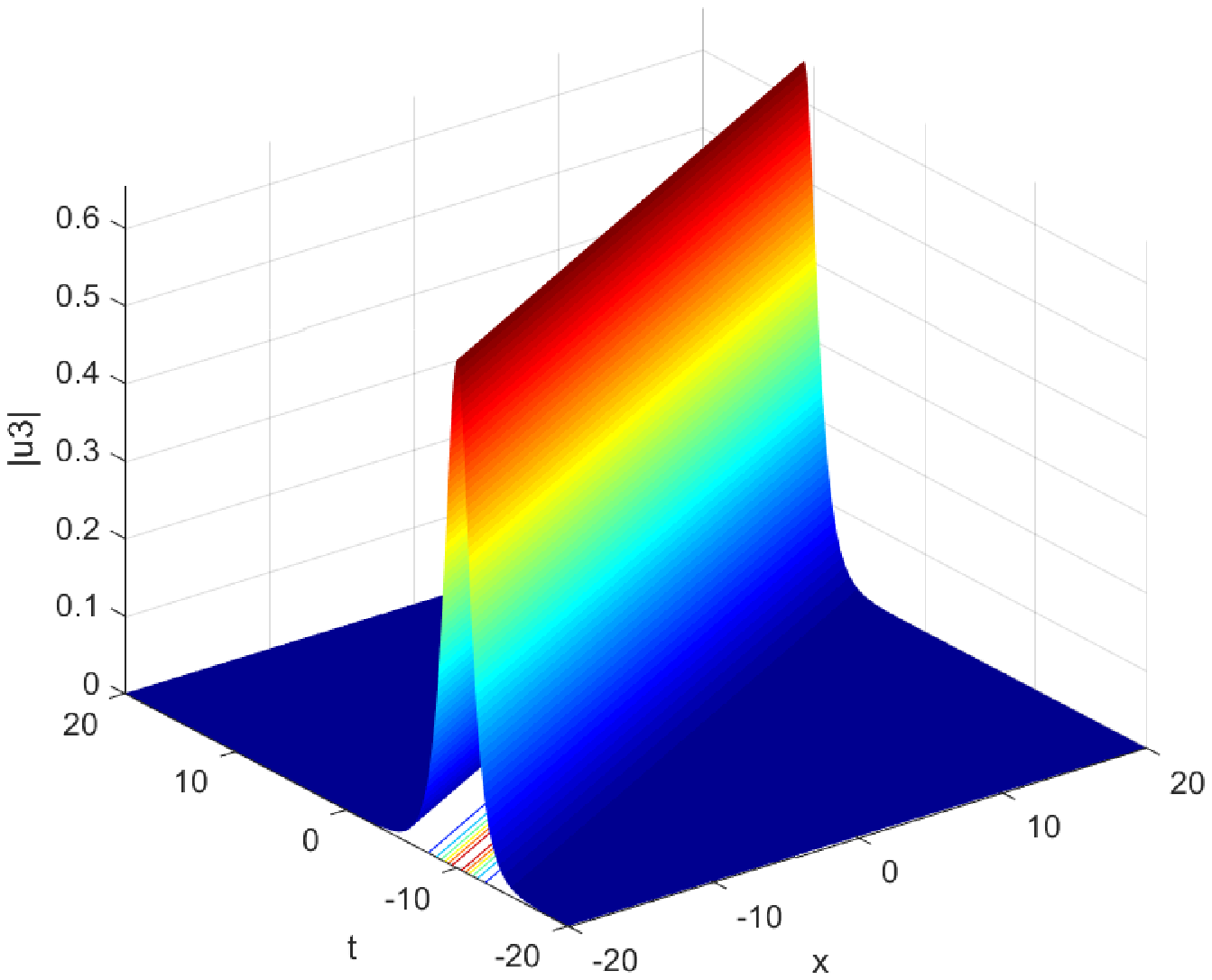}}}
~~~~
{\rotatebox{0}{\includegraphics[width=3.6cm,height=3.0cm,angle=0]{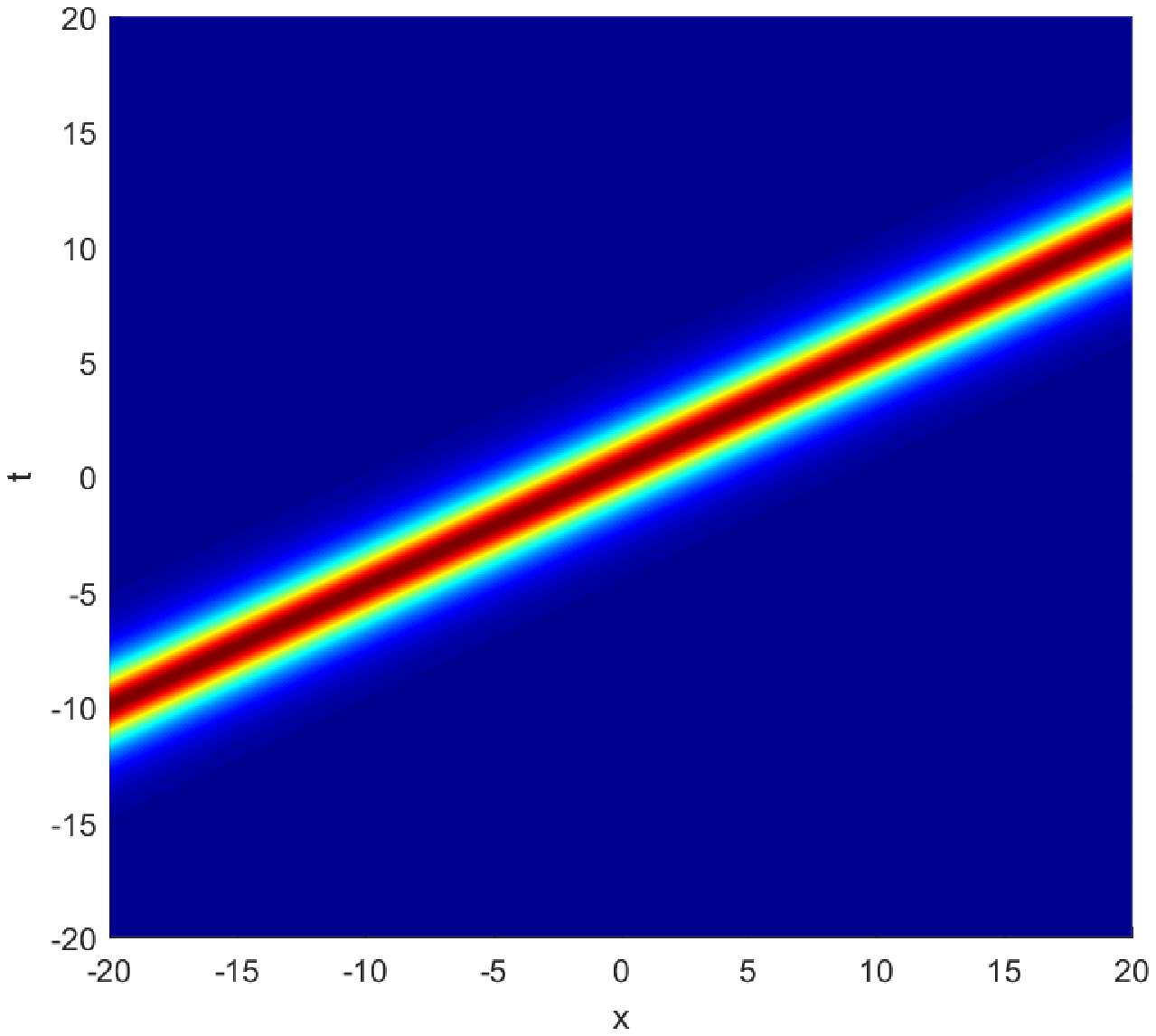}}}
~~~~
{\rotatebox{0}{\includegraphics[width=3.6cm,height=3.0cm,angle=0]{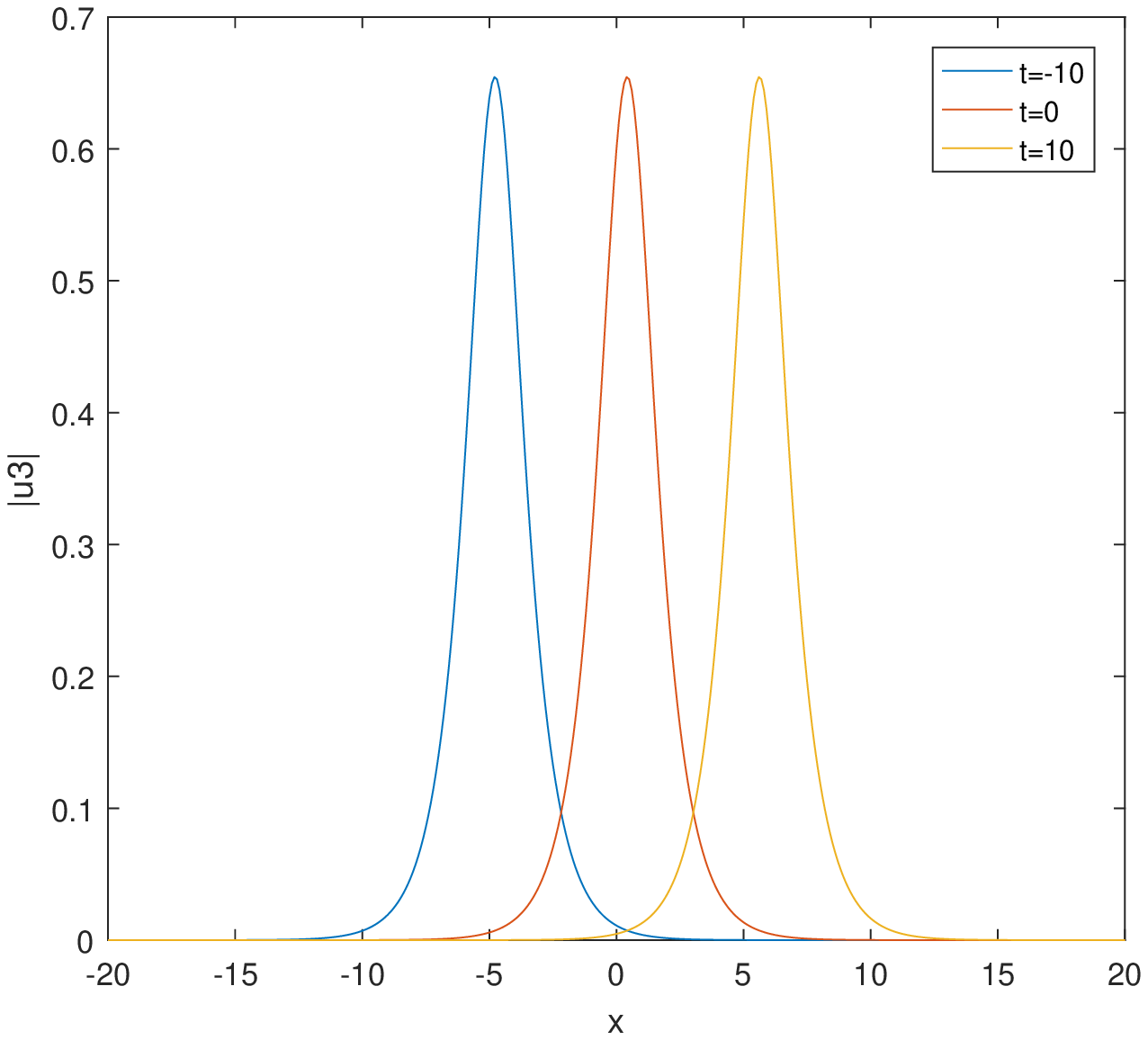}}}

$\ \qquad~~~~~~(\textbf{g})\qquad \ \qquad\qquad\qquad\qquad~(\textbf{h})
\ \qquad\qquad\qquad\qquad\qquad~(\textbf{i})$\\
\noindent { \small \textbf{Figure 4.} (Color online) One-soliton solutions to Eq. \eqref{sf3.1} with the parameters $a_1 =0.2$, $b_1=0.5$, $\alpha_{1,1}=\frac{\sqrt{3}}{2}$, $\alpha_{2,1}=\frac{\sqrt{2}}{2}$, $\alpha_{3,1}=-\frac{\sqrt{2}}{2}$, $\alpha_{4,1}=-\frac{\sqrt{3}}{2}$.
\label{fig3.11}
$\textbf{(a)(d)(g)}$: the structures of the one-soliton solutions,
$\textbf{(b)(e)(h)}$: the density plot,
$\textbf{(c)(f)(i)}$: the wave propagation of the one-soliton solutions.}\\

 Fig. 4   shows  the localized structures and  dynamic propagation behaviors of one-soliton solutions  by choosing appropriate parameters.  As shown in Fig. 4, we can see that three one-soliton soultions have the same direction of propagation, at the same time, they have similar properties in amplitude, peak and so forth.

\noindent
{\rotatebox{0}{\includegraphics[width=3.6cm,height=3.0cm,angle=0]{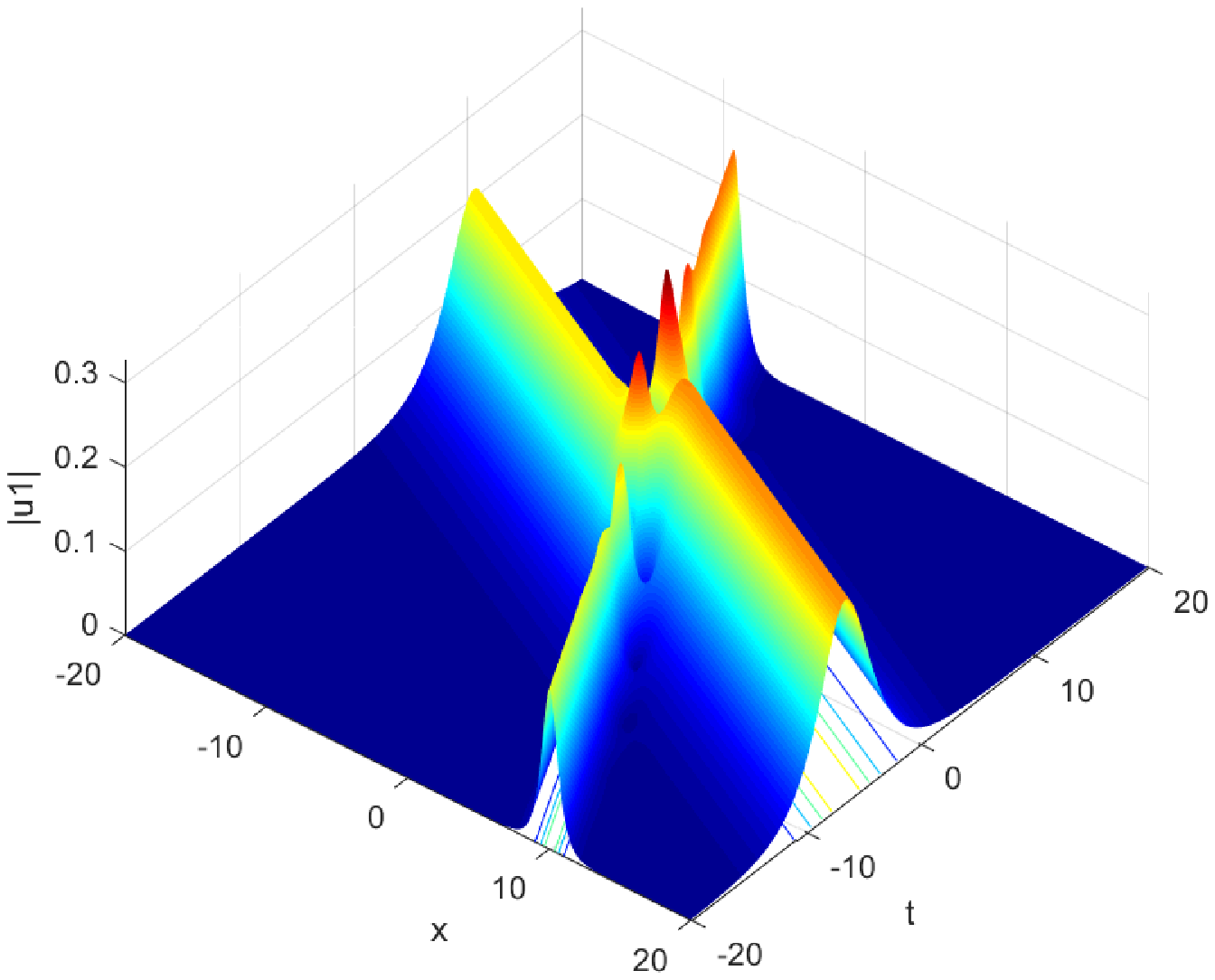}}}
~~~~
{\rotatebox{0}{\includegraphics[width=3.6cm,height=3.0cm,angle=0]{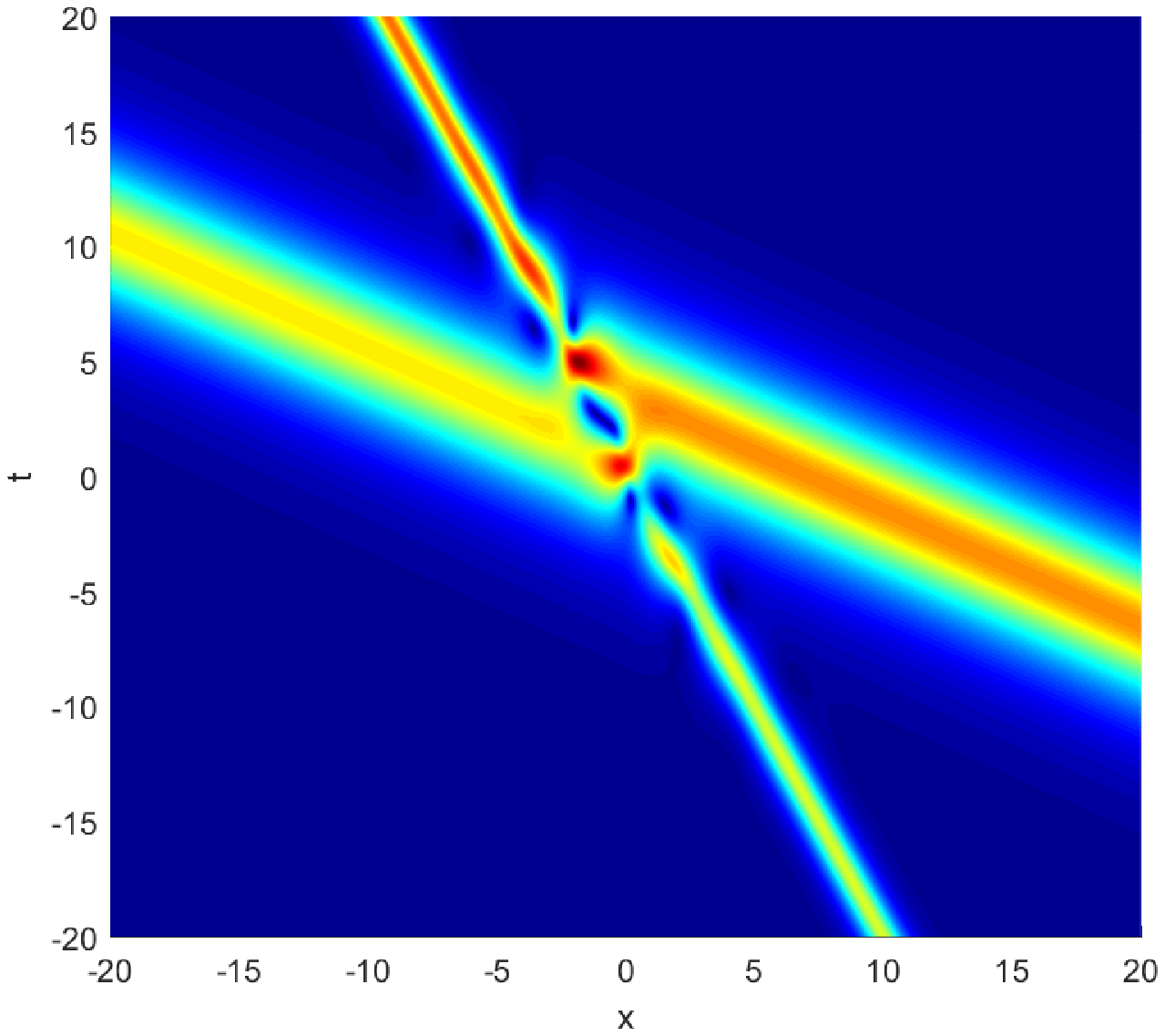}}}
~~~~
{\rotatebox{0}{\includegraphics[width=3.6cm,height=3.0cm,angle=0]{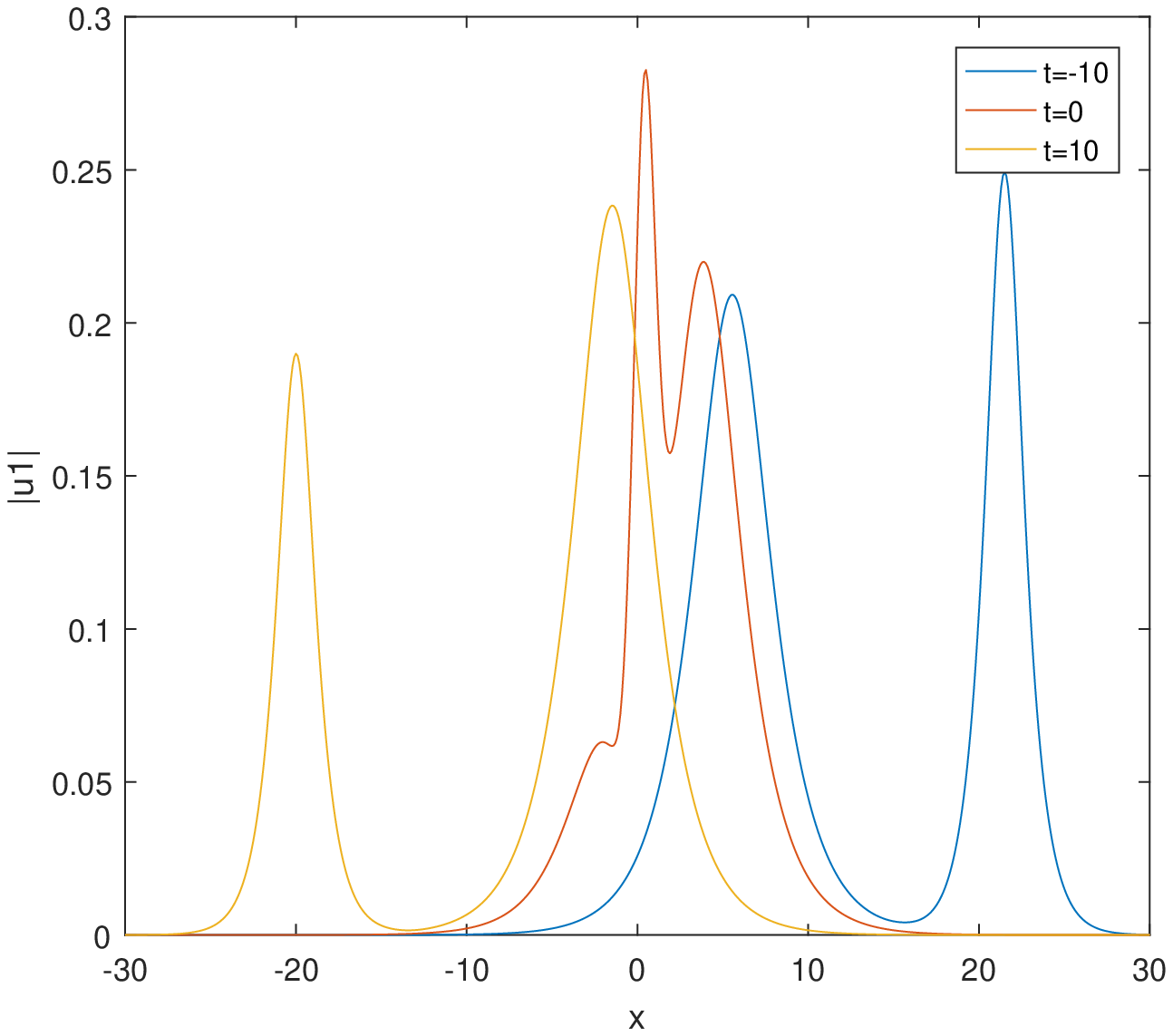}}}

$\ \qquad~~~~~~(\textbf{a})\qquad \ \qquad\qquad\qquad\qquad~(\textbf{b})
\ \qquad\qquad\qquad\qquad\qquad~(\textbf{c})$\\
\noindent
{\rotatebox{0}{\includegraphics[width=3.6cm,height=3.0cm,angle=0]{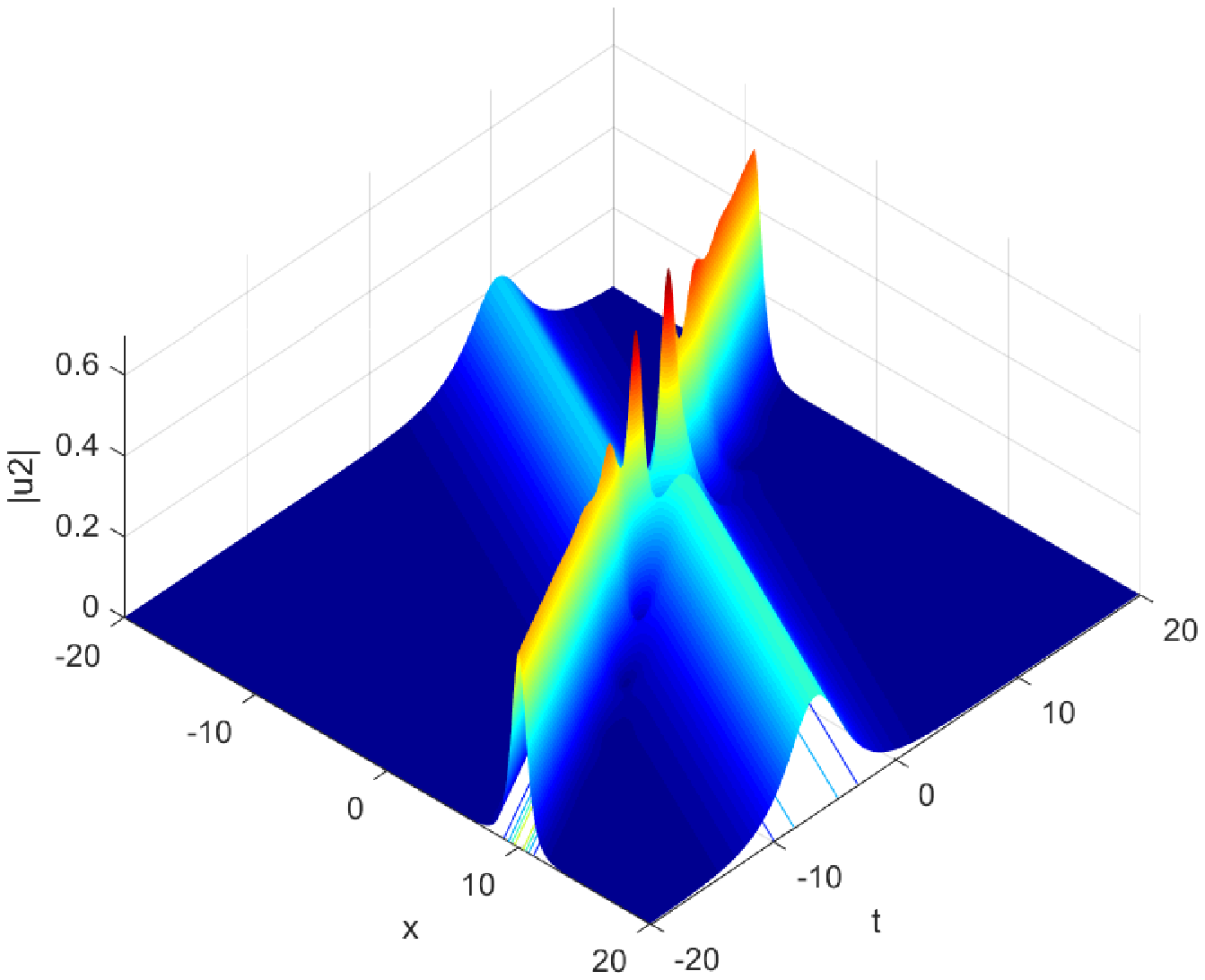}}}
~~~~
{\rotatebox{0}{\includegraphics[width=3.6cm,height=3.0cm,angle=0]{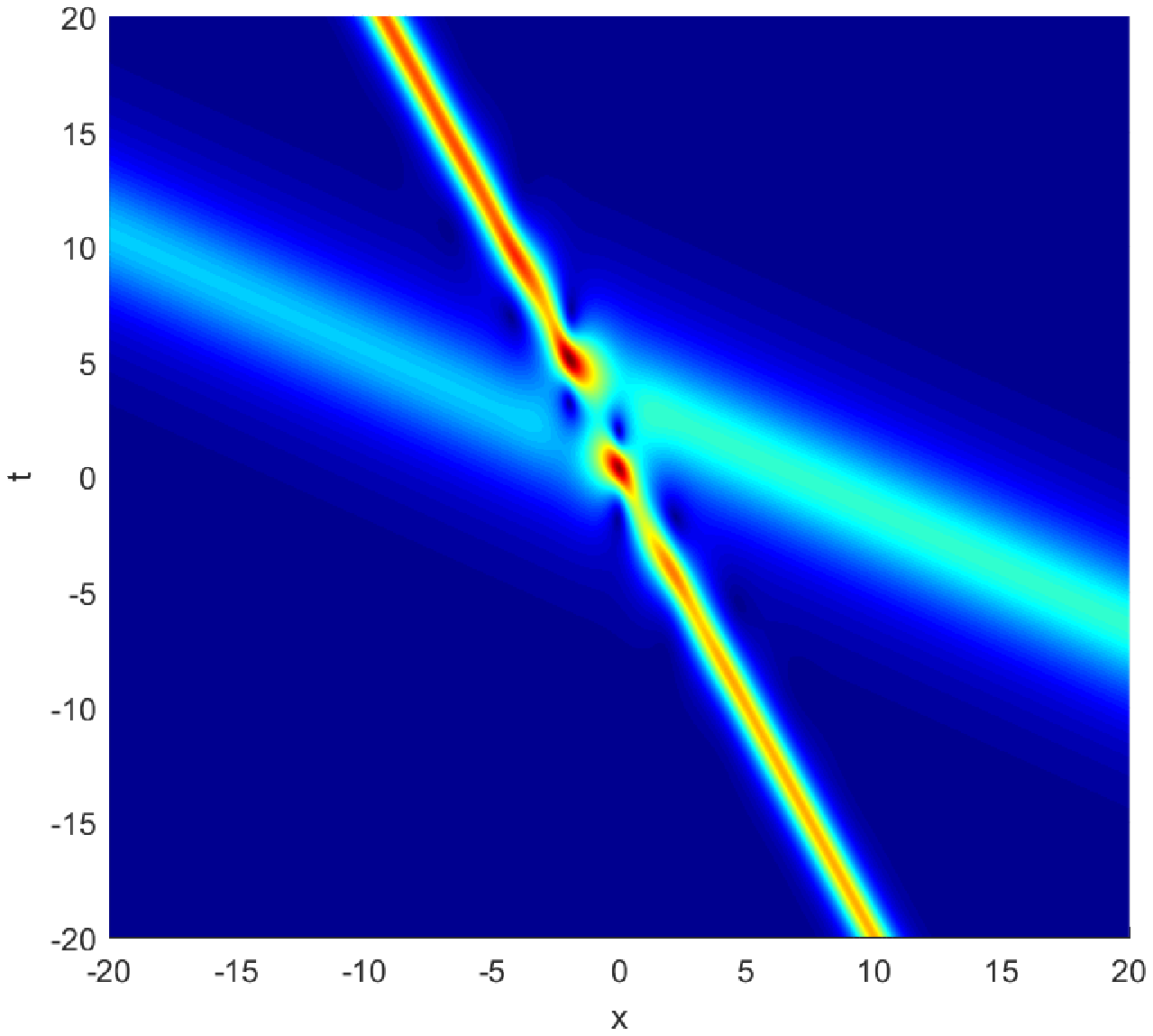}}}
~~~~
{\rotatebox{0}{\includegraphics[width=3.6cm,height=3.0cm,angle=0]{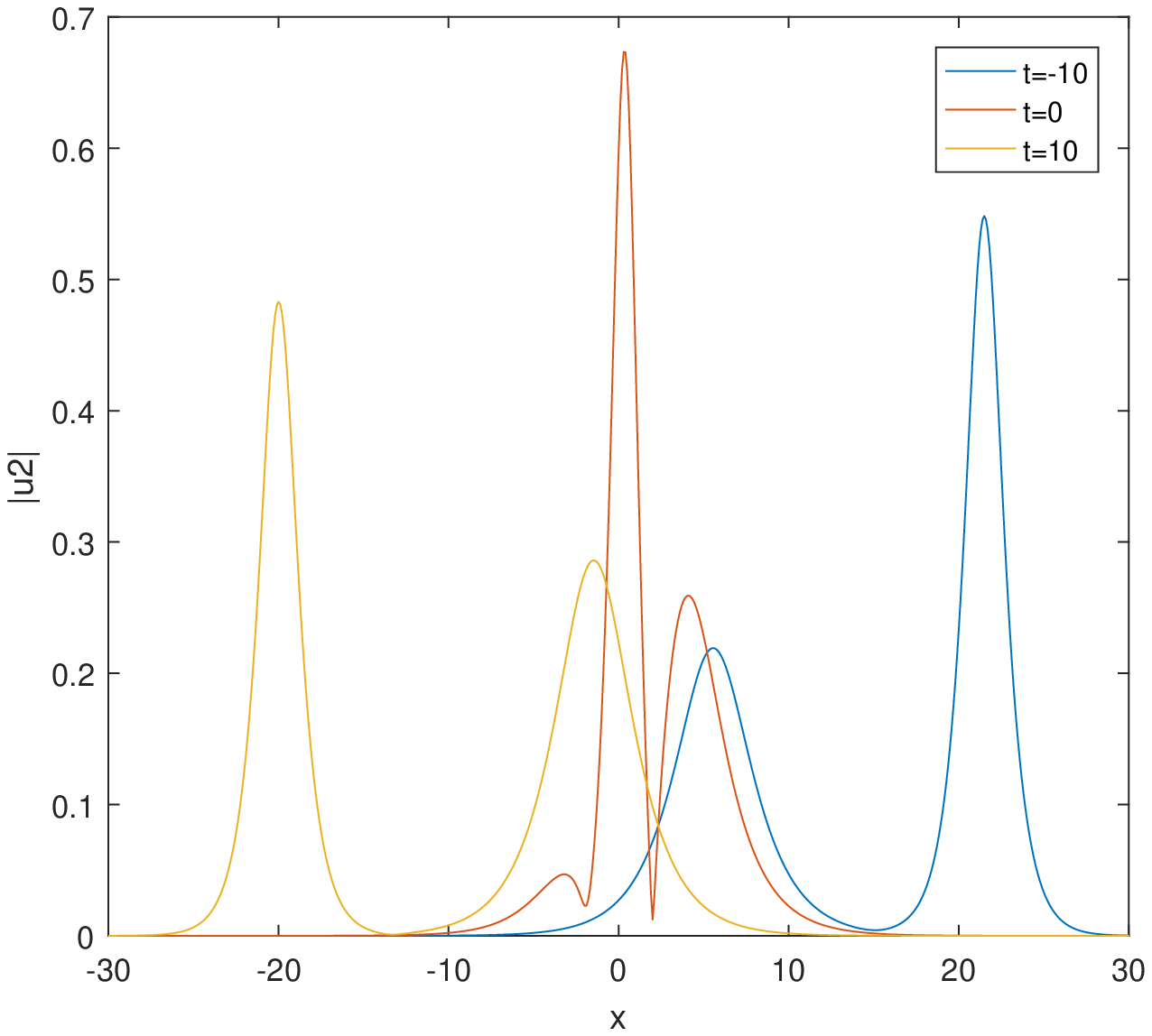}}}

$\ \qquad~~~~~~(\textbf{d})\qquad \ \qquad\qquad\qquad\qquad~(\textbf{e})
\ \qquad\qquad\qquad\qquad\qquad~(\textbf{f})$\\
\noindent
{\rotatebox{0}{\includegraphics[width=3.6cm,height=3.0cm,angle=0]{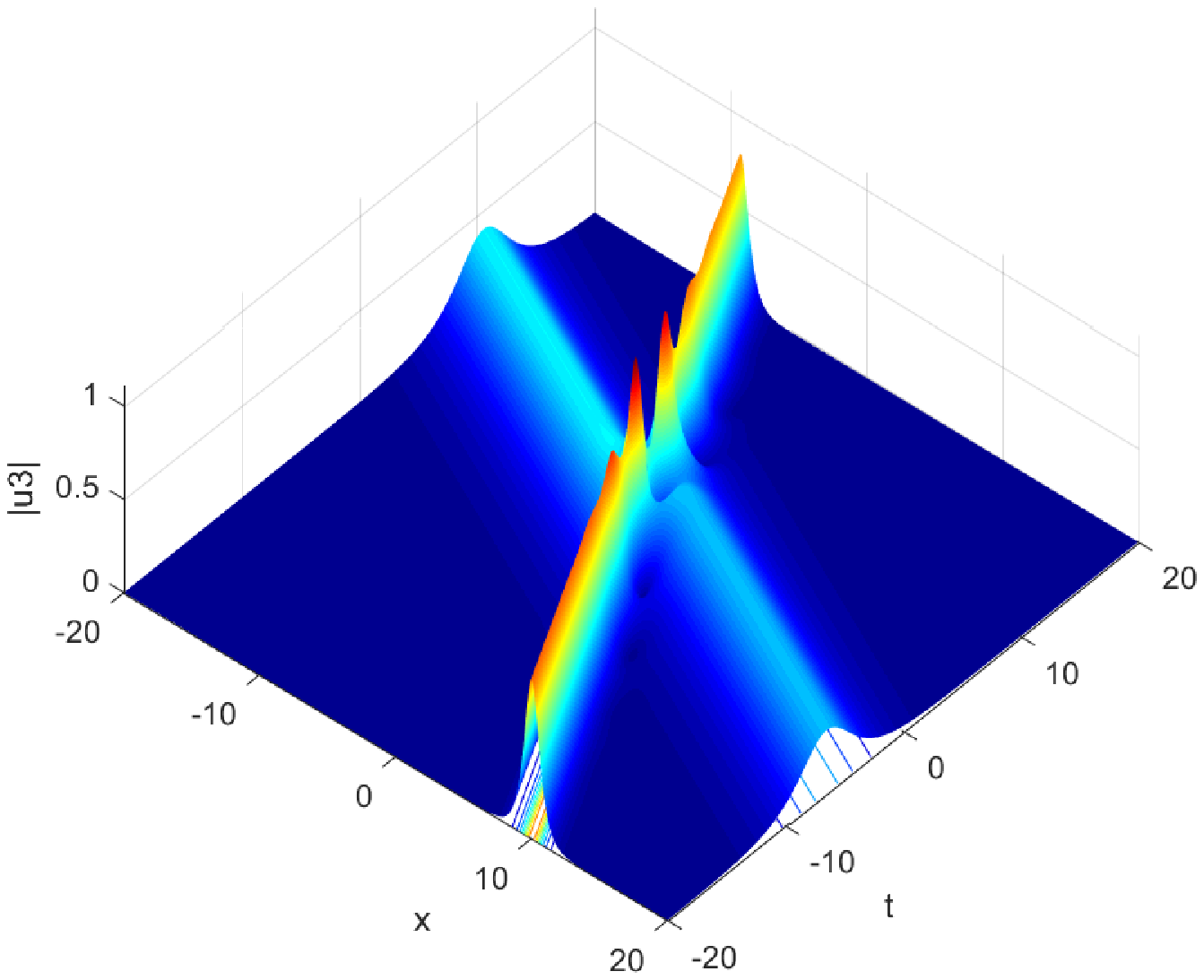}}}
~~~~
{\rotatebox{0}{\includegraphics[width=3.6cm,height=3.0cm,angle=0]{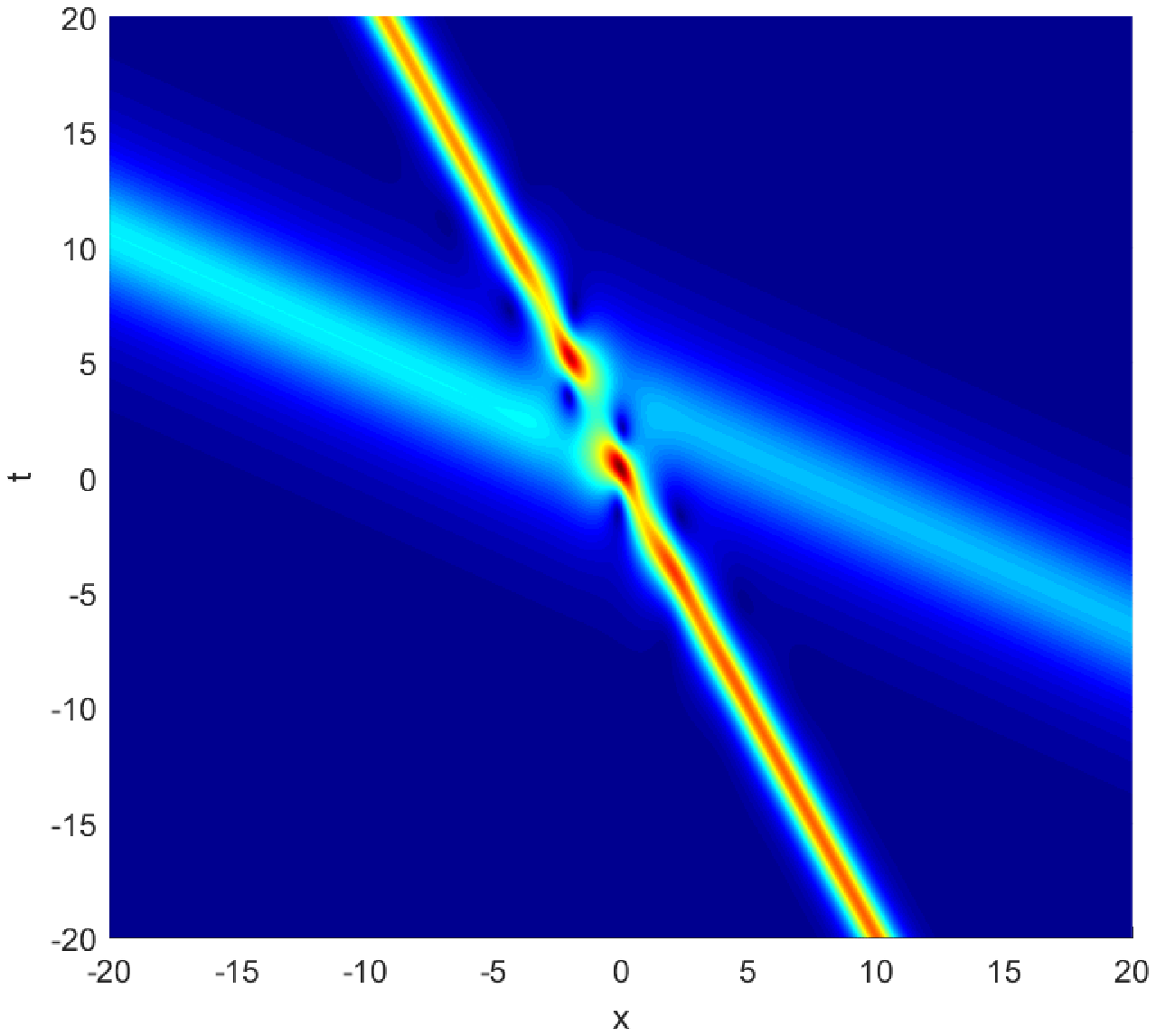}}}
~~~~
{\rotatebox{0}{\includegraphics[width=3.6cm,height=3.0cm,angle=0]{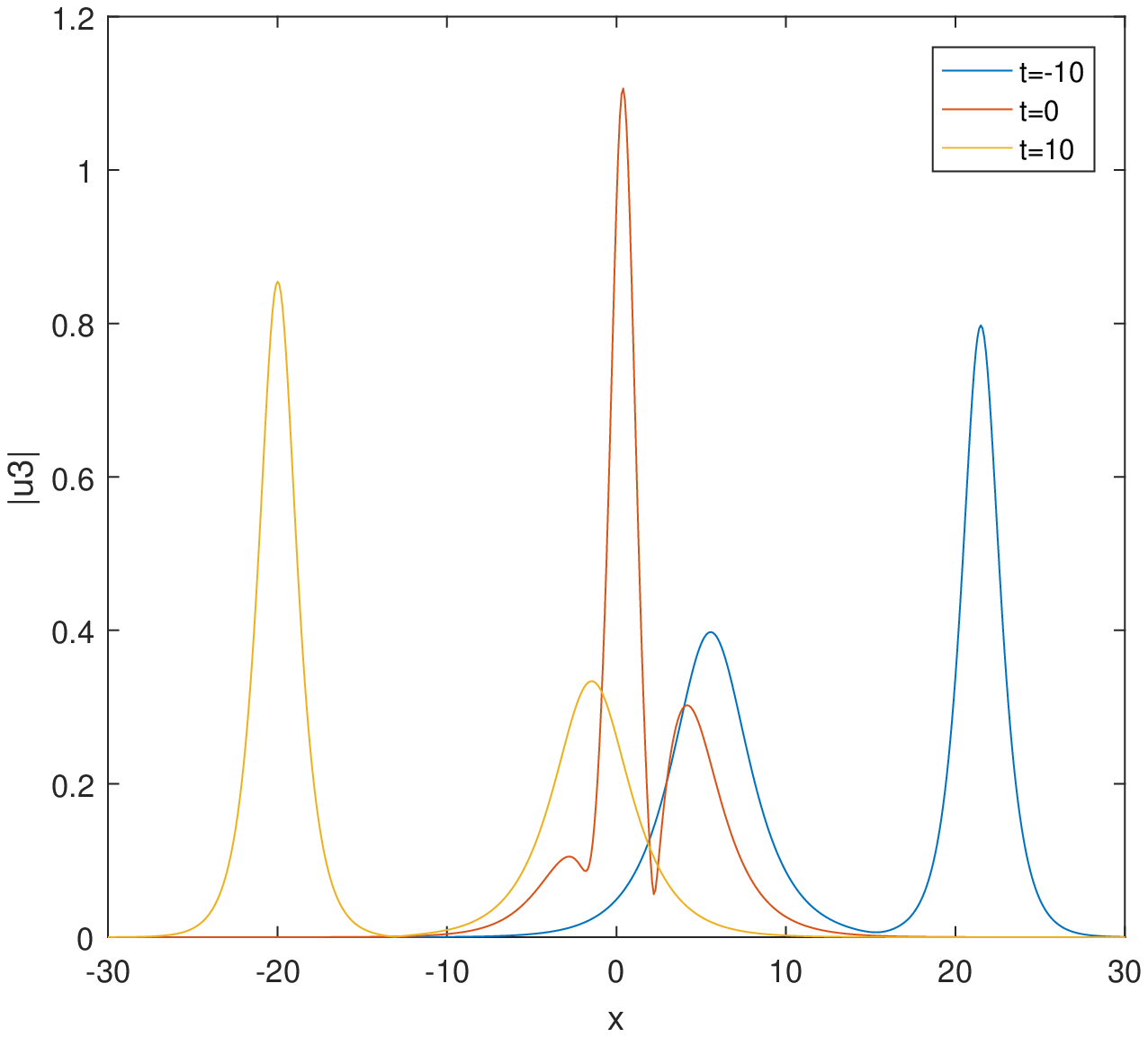}}}

$\ \qquad~~~~~~(\textbf{g})\qquad \ \qquad\qquad\qquad\qquad~(\textbf{h})
\ \qquad\qquad\qquad\qquad\qquad~(\textbf{i})$\\
\noindent { \small \textbf{Figure 5.} (Color online) Two-soliton solutions  to Eq. \eqref{sf3.1} with the parameters  $a_1 =b_1 =0.25$, $a_2=b_2=0.5$, $\alpha_{1,1}=0.4$, $\alpha_{2,1}=0.5$, $\alpha_{3,1}=0.6$, $\alpha_{4,1}=0.7$, $\alpha_{1,2}=0.5$, $\alpha_{2,2}=0.25$, $\alpha_{3,2}=0.55$, $\alpha_{4,2}=0.8$.
\label{fig3.21}
$\textbf{(a)(d)(g)}$: the structures of the two-soliton solutions,
$\textbf{(b)(e)(h)}$: the density plot,
$\textbf{(c)(f)(i)}$: the wave propagation of the two-soliton solutions.} \\

Fig. 5  presents  the localized structures and  dynamic propagation behaviors of  two-soliton solutions by choosing appropriate parameters. Through careful observation of $u_1$ in Figs. $5(a)-5(c)$, we can see that the propagtion direction of  lower left soliton is unexchanged before and after two solitons collide with each other, but its energy has been improved,  while nother soliton in the upper right has exchanged its position and decreased its energy. Similar results for $u_2$ and $u_3$ can be obtained by observing Figs. $5(d)-5(f)$ and Figs. $5(g)-5(i)$,  respectively.

Fig. 6   displays  the localized structures and  dynamic propagation behaviors  of  three-soliton solutions by choosing appropriate parameters.  Via analyzing the figures carefully, we can obtain the information about the energy, propagation direction and position of three solitons. For example, as for $u_1$, before three solitons collide with each other,  two of the three solitons travel in nearly parallel directions and the third soliton travel in another different direction. After collision,  the propagation directions of all three solitons are unexchanged, but their positions have shifted and their energies have been increased  or decreased.

\noindent
{\rotatebox{0}{\includegraphics[width=3.6cm,height=3.0cm,angle=0]{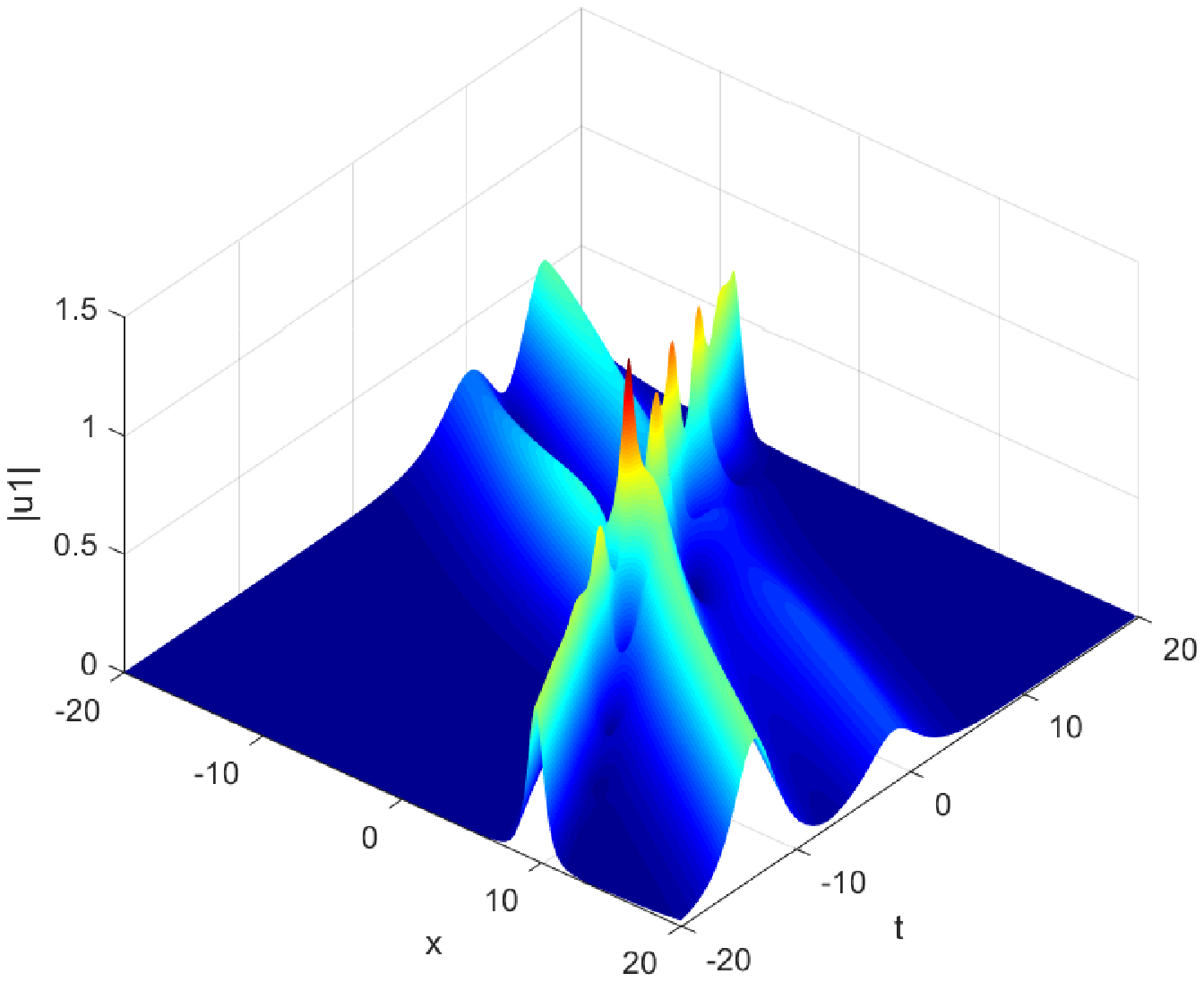}}}
~~~~
{\rotatebox{0}{\includegraphics[width=3.6cm,height=3.0cm,angle=0]{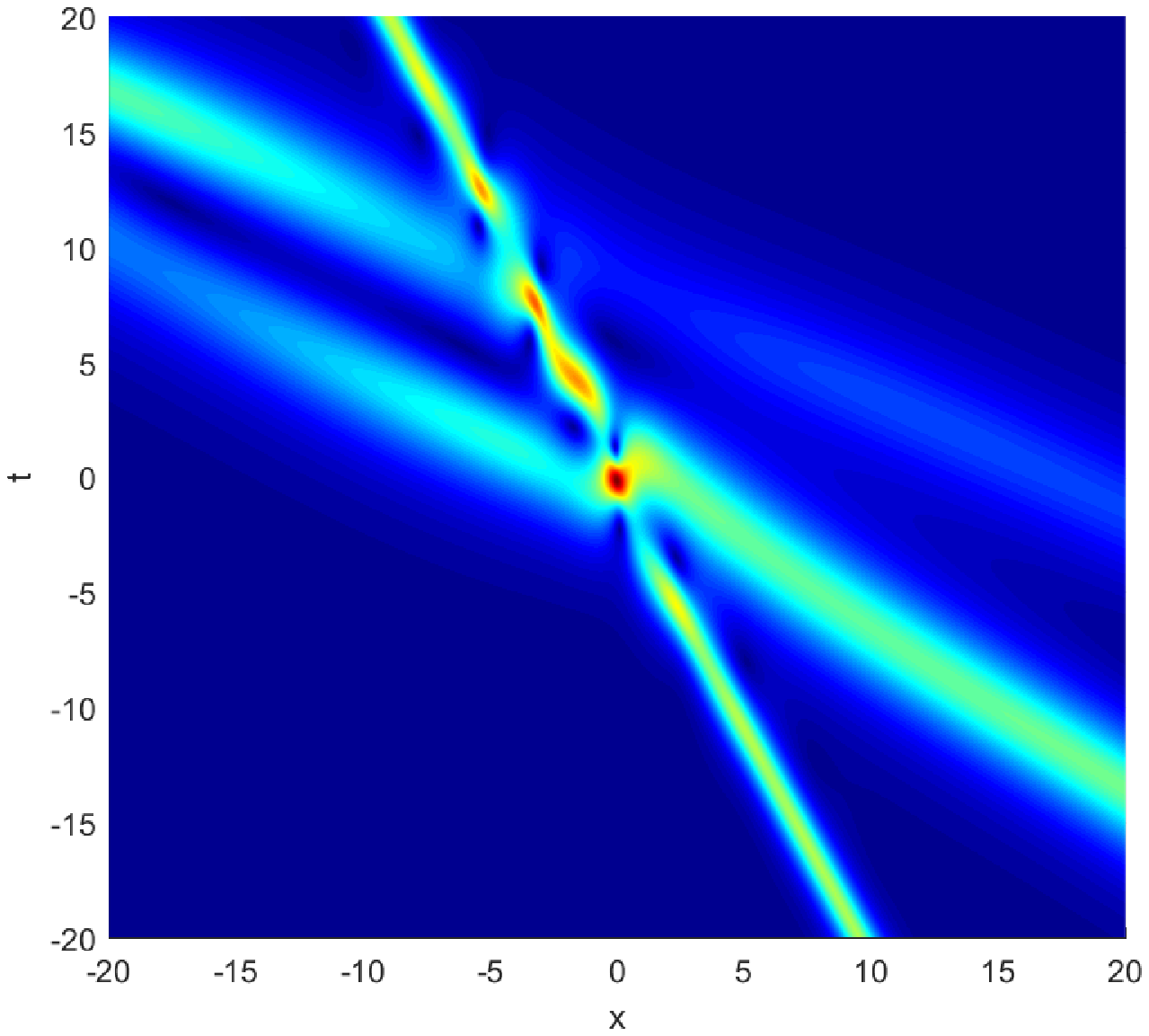}}}
~~~~
{\rotatebox{0}{\includegraphics[width=3.6cm,height=3.0cm,angle=0]{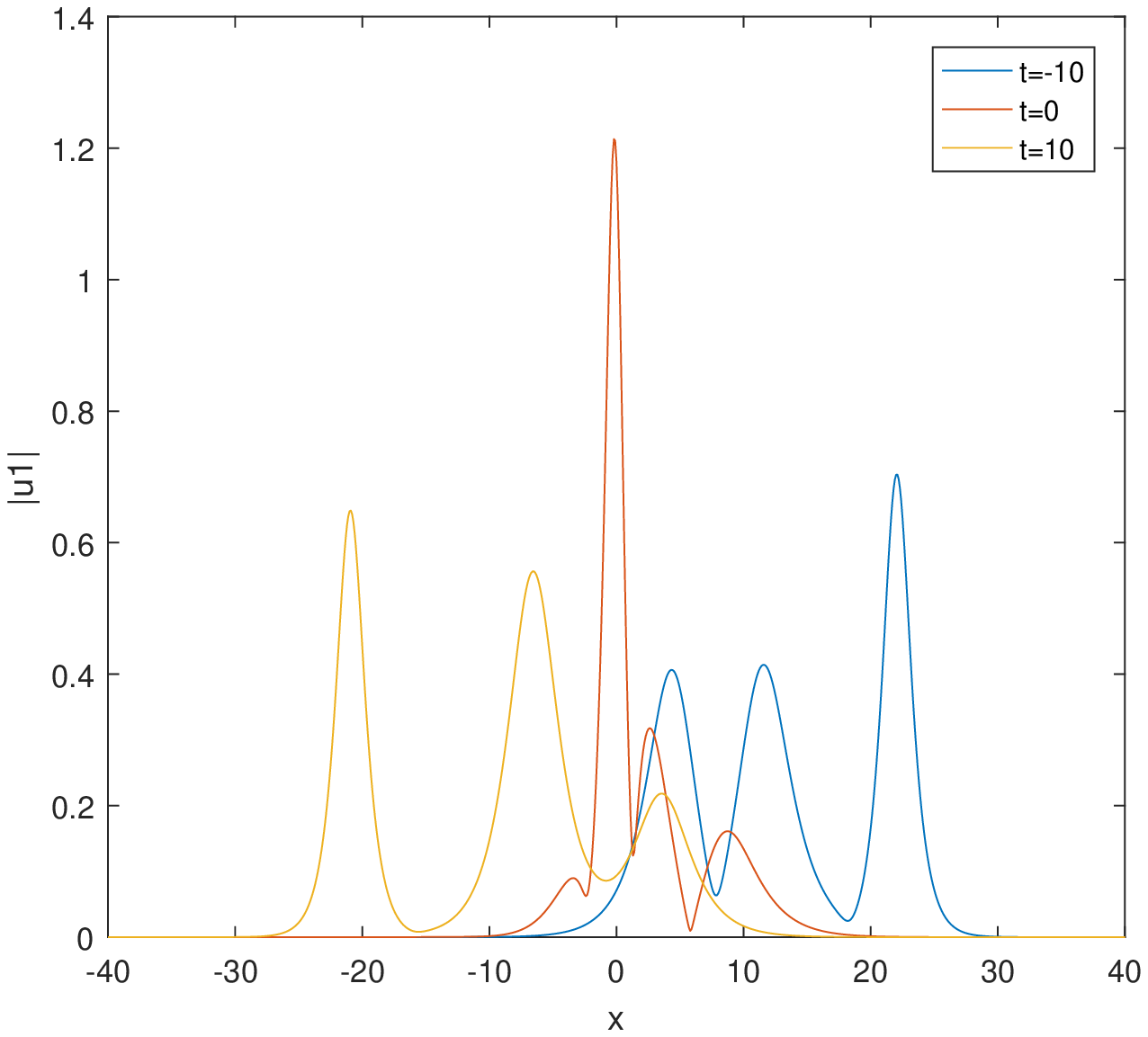}}}

$\ \qquad~~~~~~(\textbf{a})\qquad \ \qquad\qquad\qquad\qquad~(\textbf{b})
\ \qquad\qquad\qquad\qquad\qquad~(\textbf{c})$\\
\noindent
{\rotatebox{0}{\includegraphics[width=3.6cm,height=3.0cm,angle=0]{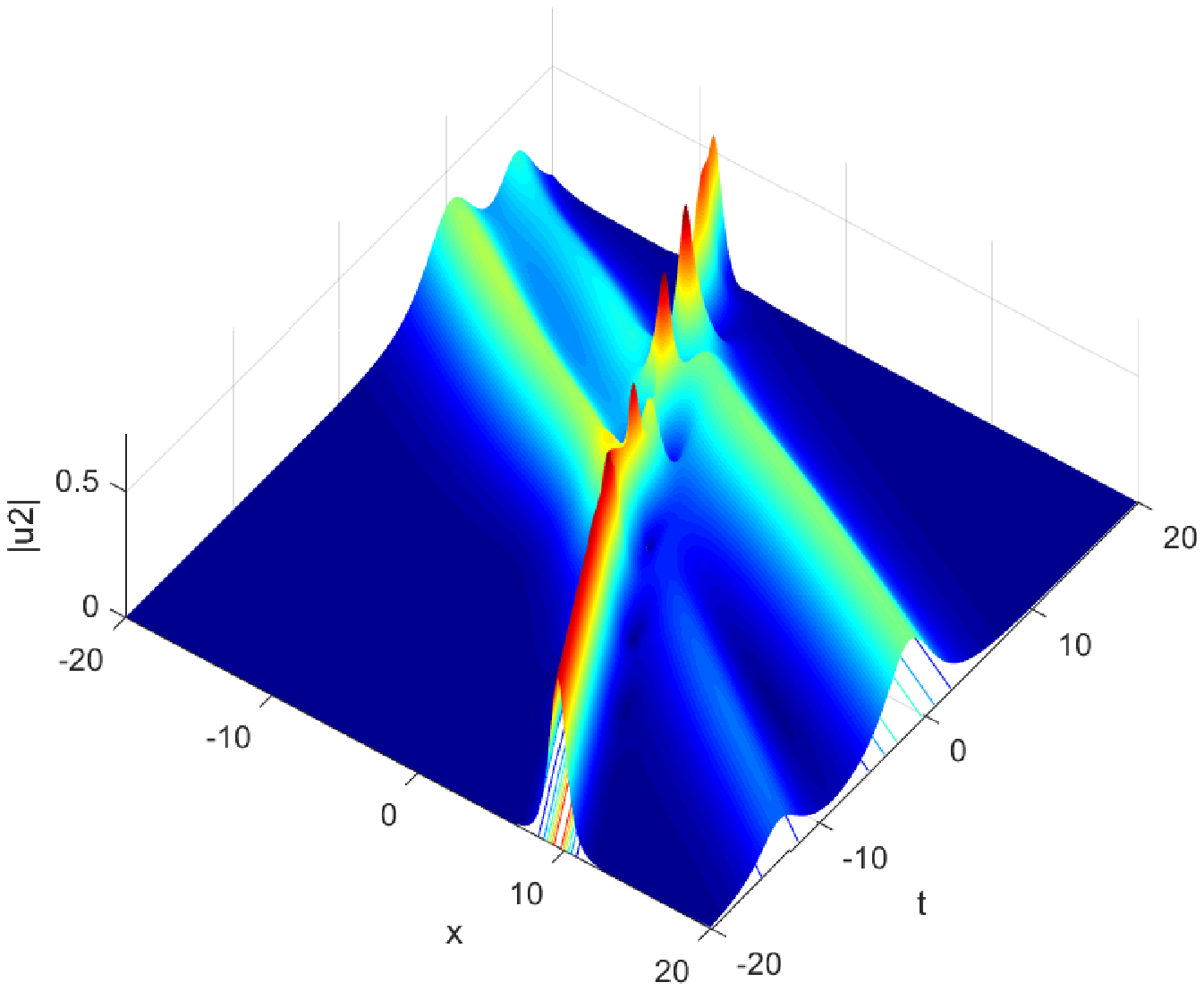}}}
~~~~
{\rotatebox{0}{\includegraphics[width=3.6cm,height=3.0cm,angle=0]{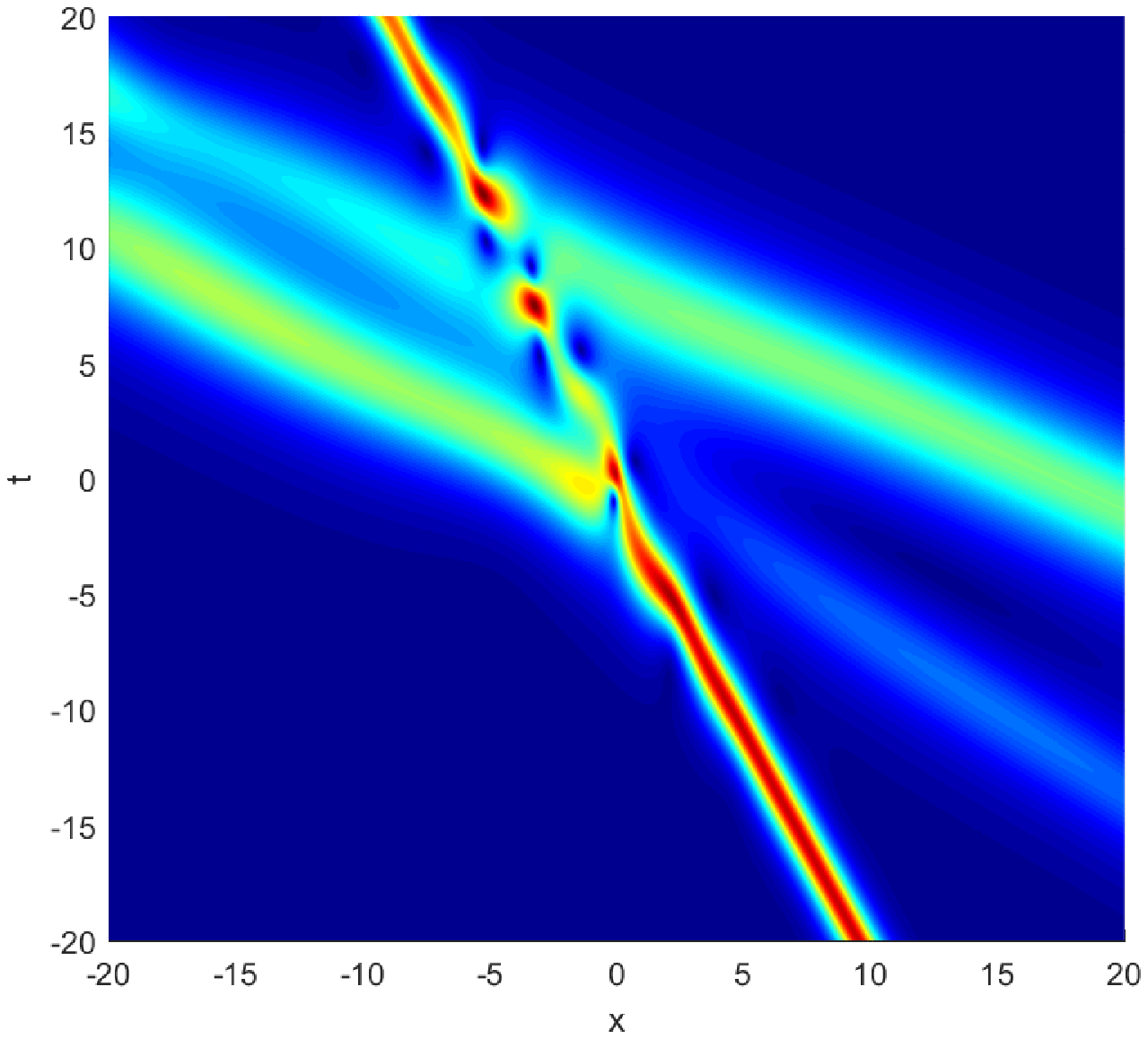}}}
~~~~
{\rotatebox{0}{\includegraphics[width=3.6cm,height=3.0cm,angle=0]{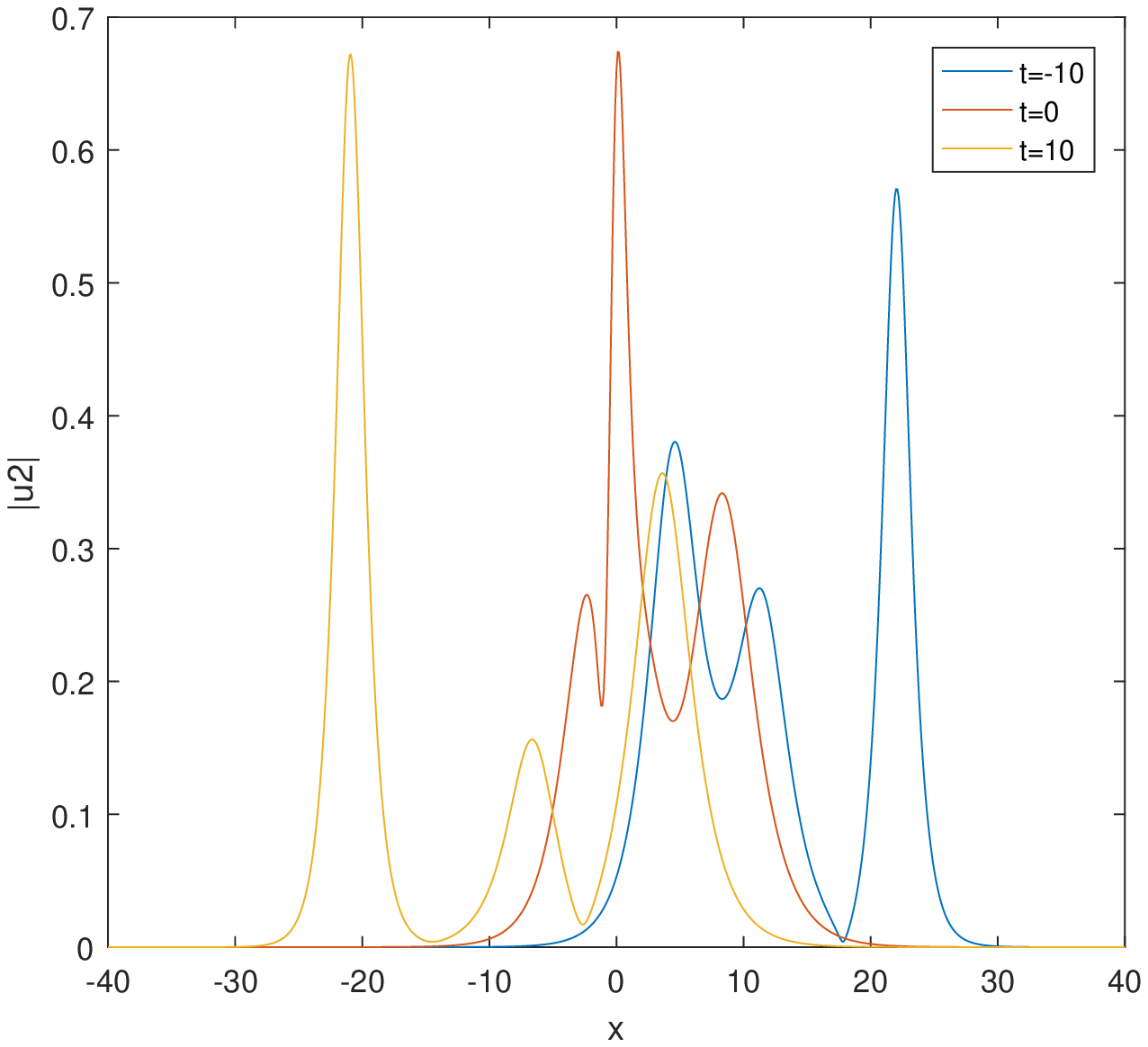}}}

$\ \qquad~~~~~~(\textbf{d})\qquad \ \qquad\qquad\qquad\qquad~(\textbf{e})
\ \qquad\qquad\qquad\qquad\qquad~(\textbf{f})$\\
\noindent
{\rotatebox{0}{\includegraphics[width=3.6cm,height=3.0cm,angle=0]{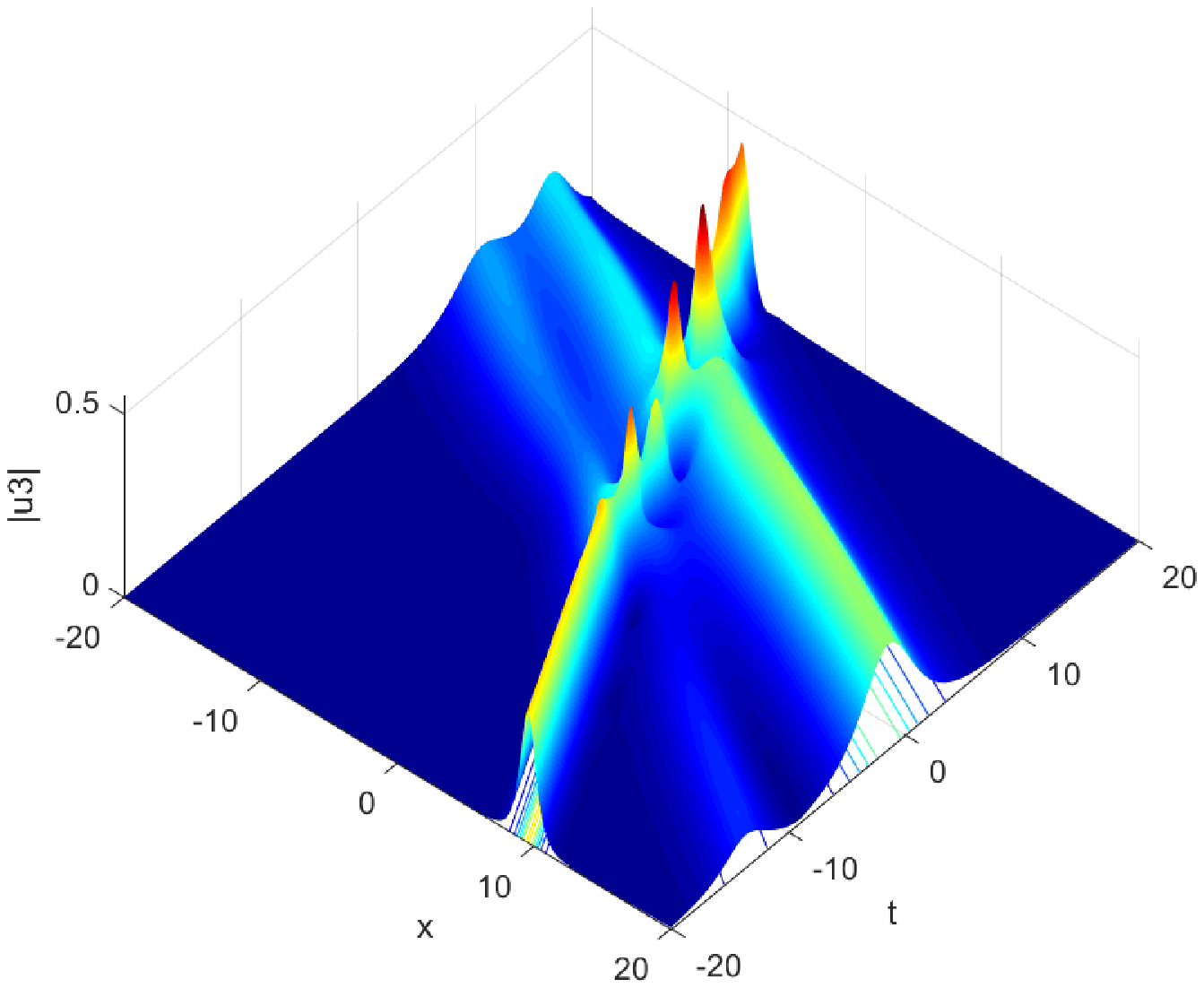}}}
~~~~
{\rotatebox{0}{\includegraphics[width=3.6cm,height=3.0cm,angle=0]{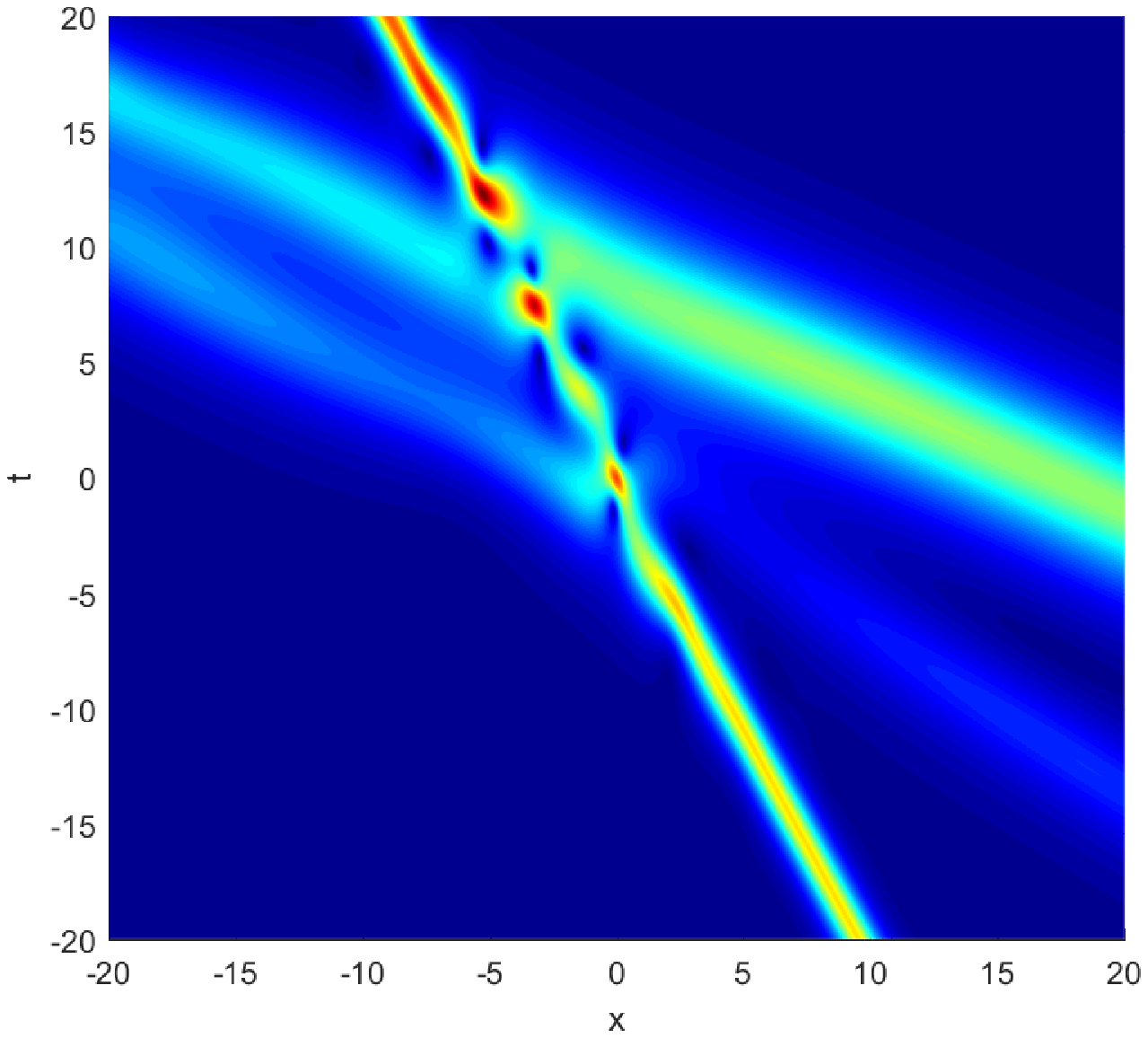}}}
~~~~
{\rotatebox{0}{\includegraphics[width=3.6cm,height=3.0cm,angle=0]{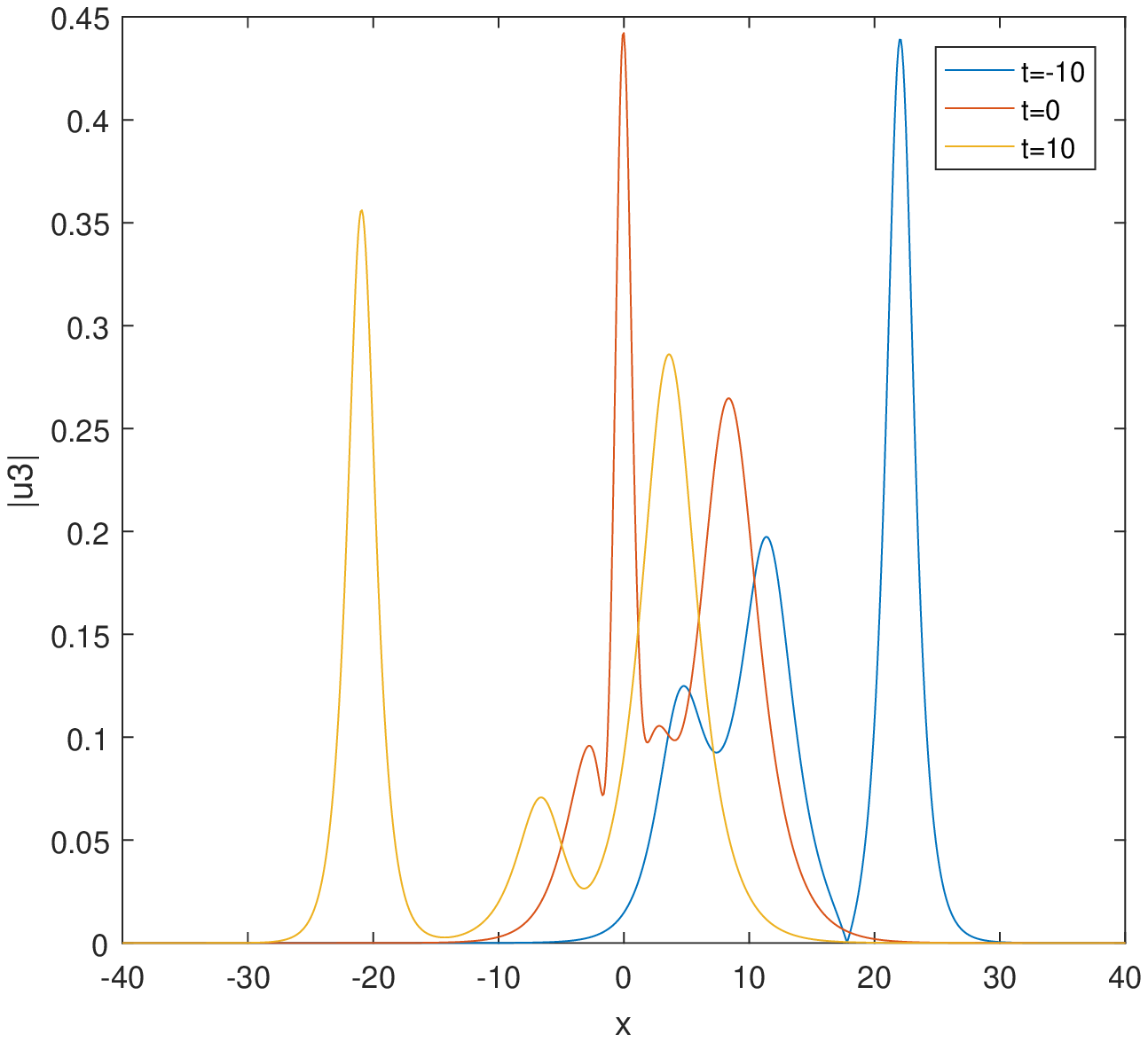}}}

$\ \qquad~~~~~~(\textbf{g})\qquad \ \qquad\qquad\qquad\qquad~(\textbf{h})
\ \qquad\qquad\qquad\qquad\qquad~(\textbf{i})$\\
\noindent { \small \textbf{Figure 6.} (Color online) Three-soliton solutions to Eq. \eqref{sf3.1} with the parameters   $a_1 =b_1 =0.25$, $a_2=b_2=0.5$, $a_3=b_3=0.3$, $\alpha_{1,1}=0.25$, $\alpha_{2,1}=0.45$, $\alpha_{3,1}=0.75$, $\alpha_{4,1}=0.6$, $\alpha_{1,2}=1.9$, $\alpha_{2,2}=1.6$, $\alpha_{3,2}=1.3$, $\alpha_{4,2}=1.0$, $\alpha_{1,3}=2.5$, $\alpha_{2,3}=2.9$, $\alpha_{3,3}=2.3$, $\alpha_{4,3}=2.1$.
\label{fig3.21}
$\textbf{(a)(d)(g)}$: the structures of the three-soliton solutions,
$\textbf{(b)(e)(h)}$: the density plot,
$\textbf{(c)(f)(i)}$: the wave propagation of the three-soliton solutions.} \\

\subsection{ $Q$ is taken as a special 6$\times$1 form}
In last subsection, we take $Q$ as special case as follows
\begin{equation}
    Q=(u_1, u_1^{*}, u_2, u_2^{*}, u_3, u_3^{*})^{T},
\end{equation}
then
\begin{equation}
    U= \begin{bmatrix}
        0  &0&0&0& 0 &0 & u_1\\
        0  &0&0&0& 0 &0 & u_1^{*}\\
         0  &0&0&0& 0 &0 & u_2\\
         0  &0&0&0& 0 &0 &u_2^{*}\\
         0  &0&0&0& 0 &0 & u_3\\
         0  &0&0&0& 0 &0 &u_3^{*}\\
         -u_1^{*} & -u_1  &-u_2^{*} &-u_2 &-u_3^{*} & -u_3^{*} & 0
        \end{bmatrix},
\end{equation}
where $p=6$,  $q=1$. Specially, besides $U = - U^{\dagger}$,  we can derive another symmetry for $U$,
\begin{equation} \label{symme}
    U^{*} = \Lambda U \Lambda,
\end{equation}
where
\begin{equation}
   \Lambda = \begin{bmatrix}
     0  & 1  & 0  & 0 & 0  & 0  & 0 \\
     1  & 0  & 0  & 0 & 0  & 0  & 0 \\
     0  & 0  & 0  & 1 & 0  & 0  & 0 \\
     0  & 0  & 1  & 0 & 0  & 0  & 0 \\
     0  & 0  & 0  & 0 & 0  & 1  & 0 \\
     0  & 0 & 0  & 0  & 1  & 0  & 0 \\
     0  & 0  & 0  & 0 & 0  & 0  & 1
   \end{bmatrix}.
\end{equation}
With the aid of Eq. \eqref{symme}, we can derive
\begin{equation} \label{q}
    \Lambda \Psi_{1}^{*}(-\zeta^{*}) \Lambda =\Psi_{1}(\zeta), \quad  \Lambda \Psi_{2}^{*}(-\zeta^{*}) \Lambda =\Psi_{2}(\zeta),
\end{equation}
which leads to
\begin{equation}
  \Lambda S^{*}(-\zeta^{*}) \Lambda = S(\zeta).
\end{equation}
Moreover, Eq. \eqref{q} also yield a property,
\begin{equation}
   \Lambda \Gamma_{1}^{*}(-\zeta^{*})\Lambda = \Gamma_{1}(\zeta),
\end{equation}
which indicates that if $\zeta_{j}$ is one zero of $\det(\Gamma_{1})$,  $-\zeta_{j}^{*}$ is also one zero of $\det(\Gamma_1)$. Therefore, we consider the zeros of $\det(\Gamma_{1})$ in the following three cases:
\begin{itemize}
  \item Suppose that $\det(\Gamma_{1})$  has $2N_{1}$ simple zeros satisfying  $\mathrm{Re} (\zeta_j)\neq 0 $,  $\zeta_j = - \zeta_{j-N_1}^{*}$, $N_1+1 \leq j \leq 2 N_1$, which are all in $\mathbb{C}^{+}$;
  \item Suppose that $\det(\Gamma_{1})$  has $N_{2}$ simple zeros $\zeta_{j}$  which are all pure imaginary in $\mathbb{C}^{+}$;
  \item Suppose that $\det(\Gamma_{1})$  has $2N_{1}+N_{2}$ simple zeros $\zeta_{j}$,  where the first $2N_1$ zeros satisfy $\mathrm{Re} (\zeta_j)\neq 0 $ and $\zeta_j = - \zeta_{j-N_1}^{*}$, $N_1+1\leq j \leq 2N_1 $,  the last $N_2$ zeros are  pure imaginary, and all zeros $\zeta_j (1 \leq j \leq 2N_{1}+N_{2})$ are all in $\mathbb{C}^{+}$.
\end{itemize}

Similar  with the last subsection, we can derive the similar results. First, suppose that the discrete scattering  data consists of $\{ \zeta_{j}, \hat{\zeta_{j}}, \vartheta_{j},  \hat{\vartheta_{j}}   \} $ satisfy
\begin{equation}
  \Gamma_{1}(\zeta_{j}) \vartheta_{j}= 0,
\end{equation}
\begin{equation}
  \hat{\vartheta}_{j} \Gamma_{2}(\hat{\zeta}_{j})= 0.
\end{equation}

For case 1,  we can obtain
\begin{equation}
  \begin{aligned}
      \hat{\vartheta}_{j}= \vartheta_{j}^{\dagger}, & \quad 1 \leq j \leq 2N_1, \\
      \vartheta_{j} = \Lambda \vartheta_{j-N_1}^{*}, & \quad N_1+1 \leq j \leq 2N_1.
  \end{aligned}
\end{equation}
More specificially,

\begin{equation}
 \vartheta_{j}=
\left\{
\begin{array}{cc}
     e^{\theta_{j}\sigma} \vartheta_{j,0},  & \quad 1 \leq j \leq N_1,  \\
     \Lambda e^{\theta_{j-N_1}^{*} \sigma} \vartheta_{j-N_1,0}^{*},  & \quad N_1+1 \leq j \leq 2N_1,
\end{array} \right.
 \end{equation}

\begin{equation}
 \hat{\vartheta}_{j}=
\left\{
\begin{array}{cc}
    \vartheta_{j,0}^{\dagger} e^{\theta_{j}^{*}\sigma},  & \quad 1 \leq j \leq N_1,  \\
   \vartheta_{j-N_1,0}^{T} e^{\theta_{j-N_1} \sigma}  \Lambda  ,  & \quad N_1+1 \leq j \leq 2N_1.
\end{array} \right.
 \end{equation}
Moreover, the solutions of RH problem in Eq. \eqref{solu} can be rewritten as follows,
\begin{equation}
\left\{
  \begin{aligned}
       \Gamma_{1}(\zeta) = \mathbb{I} -\sum_{k=1}^{2 N_1 } \sum_{j=1}^{2 N_1} \frac{\vartheta_{k}\widehat{\vartheta}_{j}(M^{-1})_{k,j}}{\zeta-\zeta_{j}^{*}},  \\
       \Gamma_{2}(\zeta) = \mathbb{I} + \sum_{k=1}^{2 N_1} \sum_{j=1}^{2 N_1} \frac{\vartheta_{k}\widehat{\vartheta}_{j}(M^{-1})_{k,j}}{\zeta-\zeta_{k}},
  \end{aligned}
\right.
\end{equation}
then we have
\begin{equation}
       \Gamma_{1}^{(1)}(\zeta) = \sum_{k=1}^{2 N_1 } \sum_{j=1}^{2 N_1}\vartheta_{k}\widehat{\vartheta}_{j}(M^{-1})_{k,j},
\end{equation}
finally, the solutions are as follows
\begin{equation}
\left\{
  \begin{aligned}
           u_1(x,t)=-2 i (\Gamma_{1}^{(1)})_{1,7},  \\
           u_2(x,t)=-2 i (\Gamma_{1}^{(1)})_{3,7},  \\
           u_3(x,t)=-2 i (\Gamma_{1}^{(1)})_{5,7}.
  \end{aligned}
\right.
\end{equation}

Similarly, for case 2, we have the following process
\begin{equation}
\left\{
  \begin{aligned}
       \vartheta_{j}=e^{\theta_{j}\sigma} \vartheta_{j,0}, & \quad 1\leq j \leq N_2, \\
       \hat{\vartheta}_{j}= \vartheta_{j,0}^{\dagger} e^{\theta_{j}^{*}\sigma},   & \quad 1\leq j \leq N_2.
  \end{aligned}
\right.
\end{equation}
Moreover, the solutions of RH problem in Eq.(\ref{solu}) can be rewritten as follows
\begin{equation}
\left\{
  \begin{aligned}
       \Gamma_{1}(\zeta) = \mathbb{I} -\sum_{k=1}^{N_2 } \sum_{j=1}^{N_2} \frac{\vartheta_{k}\widehat{\vartheta}_{j}(M^{-1})_{k,j}}{\zeta-\zeta_{j}^{*}},   \\
       \Gamma_{2}(\zeta) = \mathbb{I} + \sum_{k=1}^{N_2} \sum_{j=1}^{N_2} \frac{\vartheta_{k}\widehat{\vartheta}_{j}(M^{-1})_{k,j}}{\zeta-\zeta_{k}},
  \end{aligned}
\right.
\end{equation}
then we have
\begin{equation}
       \Gamma_{1}^{(1)}(\zeta) = \sum_{k=1}^{N_2} \sum_{j=1}^{N_2}\vartheta_{k}\widehat{\vartheta}_{j}(M^{-1})_{k,j}.
\end{equation}
Finally, the solutions are as follows
\begin{equation}
\left\{
  \begin{aligned}
           u_1(x,t)=-2 i (\Gamma_{1}^{(1)})_{1,7}, \\
           u_2(x,t)=-2 i (\Gamma_{1}^{(1)})_{3,7},  \\
           u_3(x,t)=-2 i (\Gamma_{1}^{(1)})_{5,7}.
  \end{aligned}
\right.
\end{equation}

Equally, for case 3, we have the following process:

\begin{equation}
 \vartheta_{j}=
\left\{
\begin{array}{cc}
     e^{\theta_{j}\sigma} \vartheta_{j,0},  & \quad 1 \leq j \leq N_1,  \\
     \Lambda e^{\theta_{j-N_1}^{*} \sigma} \vartheta_{j-N_1,0}^{*},  & \quad N_1+1 \leq j \leq 2N_1, \\
      e^{\theta_{j}\sigma} \vartheta_{j,0},         &  \quad 2 N_1 + 1 \leq j \leq 2N_1 +N_2,
\end{array} \right.
 \end{equation}

\begin{equation}
 \hat{\vartheta}_{j}=
\left\{
\begin{array}{cc}
    \vartheta_{j,0}^{\dagger} e^{\theta_{j}^{*}\sigma},  & \quad 1 \leq j \leq N_1,  \\
   \vartheta_{j-N_1,0}^{T} e^{\theta_{j-N_1} \sigma}  \Lambda  ,  & \quad N_1+1 \leq j \leq 2N_1, \\
   \vartheta_{j,0}^{\dagger} e^{\theta_{j}^{*}\sigma},  & \quad 2 N_1 + 1 \leq j \leq 2 N_1 +N_2.
\end{array} \right.
 \end{equation}

Moreover, the solutions of RH problem in Eq.(\ref{solu}) can be rewritten as follows
\begin{equation}
\left\{
  \begin{aligned}
       \Gamma_{1}(\zeta) = \mathbb{I} -\sum_{k=1}^{2 N_1 +N_2 } \sum_{j=1}^{2 N_1 +N_2 } \frac{\vartheta_{k}\widehat{\vartheta}_{j}(M^{-1})_{k,j}}{\zeta-\zeta_{j}^{*}},   \\
       \Gamma_{2}(\zeta) = \mathbb{I} + \sum_{k=1}^{2 N_1 +N_2 } \sum_{j=1}^{2 N_1 +N_2 } \frac{\vartheta_{k}\widehat{\vartheta}_{j}(M^{-1})_{k,j}}{\zeta-\zeta_{k}},
  \end{aligned}
\right.
\end{equation}
then we have
\begin{equation}
       \Gamma_{1}^{(1)}(\zeta) = \sum_{k=1}^{2 N_1 +N_2 } \sum_{j=1}^{2 N_1 +N_2 }\vartheta_{k}\widehat{\vartheta}_{j}(M^{-1})_{k,j}.
\end{equation}
Finally, the solutions are as follows:
\begin{equation}
\left\{
  \begin{aligned}
           u_1(x,t)=-2 i (\Gamma_{1}^{(1)})_{1,7}, \\
           u_2(x,t)=-2 i (\Gamma_{1}^{(1)})_{3,7},  \\
           u_3(x,t)=-2 i (\Gamma_{1}^{(1)})_{5,7}.
  \end{aligned}
\right.
\end{equation}

\subsubsection{Multi-soliton solutions}
For case 1, we set $\vartheta_{j,0}=(\alpha_{1,j}, \alpha_{2,j},\alpha_{3,j},\alpha_{4,j},\alpha_{5,j},\alpha_{6,j},1)^{T}$,
 then
 \begin{equation}
\left\{
 \begin{aligned}
   u_1(x,t)=& 2i\sum_{k=1}^{N_1}\sum_{j=1}^{N_1} \alpha_{1,k} e^{\theta_k -\theta_j^{*}}(M^{-1})_{k,j} + 2i\sum_{k=1}^{N_1}\sum_{j=N_1+1}^{2 N_1} \alpha_{1,k} e^{\theta_k -\theta_{j-N_1}}(M^{-1})_{k,j}   \\
      & +2i\sum_{k=N_1+1}^{2 N_1}\sum_{j=1}^{  N_1} \alpha_{2,k-N_1}^{*} e^{\theta_{k-N_1}^{*} -\theta_{j}^{*}}(M^{-1})_{k,j} \\
        &+2i\sum_{k=N_1+1}^{2 N_1}\sum_{j=N_1+1}^{ 2 N_1} \alpha_{2,k-N_1}^{*} e^{\theta_{k-N_1}^{*} -\theta_{j-N_1}}(M^{-1})_{k,j},  \\
        u_2(x,t)=& 2i\sum_{k=1}^{N_1}\sum_{j=1}^{N_1} \alpha_{3,k} e^{\theta_k -\theta_j^{*}}(M^{-1})_{k,j} + 2i\sum_{k=1}^{N_1}\sum_{j=N_1+1}^{2 N_1} \alpha_{3,k} e^{\theta_k -\theta_{j-N_1}}(M^{-1})_{k,j}   \\
      & +2i\sum_{k=N_1+1}^{2 N_1}\sum_{j=1}^{  N_1} \alpha_{4,k-N_1}^{*} e^{\theta_{k-N_1}^{*} -\theta_{j}^{*}}(M^{-1})_{k,j}\\
                &+2i\sum_{k=N_1+1}^{2 N_1}\sum_{j=N_1+1}^{ 2 N_1} \alpha_{4,k-N_1}^{*} e^{\theta_{k-N_1}^{*} -\theta_{j-N_1}}(M^{-1})_{k,j},   \\
        u_3(x,t) =& 2i\sum_{k=1}^{N_1}\sum_{j=1}^{N_1} \alpha_{5,k} e^{\theta_k -\theta_j^{*}}(M^{-1})_{k,j} + 2i \sum_{k=1}^{N_1}\sum_{j=N_1+1}^{2 N_1} \alpha_{5,k} e^{\theta_k -\theta_{j-N_1}}(M^{-1})_{k,j}   \\
      & + 2i\sum_{k=N_1+1}^{2 N_1}\sum_{j=1}^{  N_1} \alpha_{6,k-N_1}^{*} e^{\theta_{k-N_1}^{*} -\theta_{j}^{*}}(M^{-1})_{k,j}\\
   & +2i\sum_{k=N_1+1}^{2 N_1}\sum_{j=N_1+1}^{ 2 N_1} \alpha_{6,k-N_1}^{*} e^{\theta_{k-N_1}^{*} -\theta_{j-N_1}}(M^{-1})_{k,j},   \\
 \end{aligned}
\right.
 \end{equation}
where $M$ is an $2N_1  \times 2N_1$ matrix. For simplicity, we define
\begin{equation}
\left\{
\begin{aligned}
   M^{(1)}_{k,j}= &  \alpha_{1,k}^{*} \alpha_{1,j}+\alpha_{2,k}^{*} \alpha_{2,j}+ \dots +\alpha_{6,k}^{*} \alpha_{6,j},  \\
   M^{(2)}_{k,j}= &  \alpha_{1,k}^{*} \alpha_{2,j-N_1}^{*}+\alpha_{2,k}^{*} \alpha_{1,j-N_1}^{*}+ \dots +\alpha_{6,k}^{*} \alpha_{5,j-N_1}^{*},   \\
    M^{(3)}_{k,j}= &  \alpha_{2,k-N_1}  \alpha_{1,j}+\alpha_{1,k-N_1}  \alpha_{2,j}+ \dots +\alpha_{5,k-N_1} \alpha_{6,j},   \\
    M^{(4)}_{k,j}= &  \alpha_{1,k-N_1}  \alpha_{1,j-N_1}^{*} + \alpha_{2,k-N_1} \alpha_{2,j-N_1}^{*}+ \dots +\alpha_{6,k-N_1} \alpha_{6,j-N_1}^{*},
\end{aligned}
\right.
\end{equation}
then
\begin{equation}
 m_{k,j}=
\left\{
\begin{array}{l r}
     \frac{M^{(1)}_{k,j}e^{\theta_{k}^{*}+\theta_{j}}+ e^{-\theta_{k}^{*}-\theta_{j}} }{\zeta_{j}-\zeta_{k}^{*}},  &     \quad 1 \leq k,j \leq N_1          \\
     \frac{M^{(2)}_{k,j}e^{\theta_{k}^{*}+\theta_{j-N_1}^{*}}+ e^{-\theta_{k}^{*}- \theta_{j-N_1}^{*}} }{\zeta_{j}-\zeta_{k}^{*}},  &     \quad 1 \leq k \leq N_1 , N_1+1 \leq  j \leq 2 N_1   \\
      \frac{M^{(3)}_{k,j}e^{\theta_{k-N_1} +\theta_{j }}+ e^{-\theta_{k-N_1} -\theta_{j }} }{\zeta_{j}-\zeta_{k}^{*}},  &     \quad N_1+1 \leq  k \leq 2 N_1, 1 \leq j \leq N_1 ,  \\
       \frac{M^{(4)}_{k,j}e^{\theta_{k-N_1} +\theta_{j-N_1 }^{*}}+ e^{-\theta_{k-N_1} -\theta_{j-N_1 }^{*}} }{\zeta_{j}-\zeta_{k}^{*}},  &     \quad N_1+1 \leq k,j \leq  2 N_1.
\end{array} \right.
 \end{equation}

For case 2, we set $\vartheta_{j,0}=(\alpha_{1,j}, \alpha_{2,j},\alpha_{3,j},\alpha_{4,j},\alpha_{5,j},\alpha_{6,j},1)^{T}$,
 then
 \begin{equation}
\left\{
 \begin{aligned}
   u_1(x,t)&= 2i\sum_{k=1}^{N_2}\sum_{j=1}^{N_2} \alpha_{1,k} e^{\theta_k -\theta_j^{*}}(M^{-1})_{k,j}, \\
   u_2(x,t)&= 2i\sum_{k=1}^{N_2}\sum_{j=1}^{N_2} \alpha_{3,k} e^{\theta_k -\theta_j^{*}}(M^{-1})_{k,j}, \\
   u_3(x,t)&= 2i\sum_{k=1}^{N_2}\sum_{j=1}^{N_2} \alpha_{5,k} e^{\theta_k -\theta_j^{*}}(M^{-1})_{k,j}, \\
 \end{aligned}
\right.
 \end{equation}
where
\begin{equation}
 m_{k,j}= \frac{M^{(1)}_{k,j}e^{\theta_{k}^{*}+\theta_{j}} + e^{-\theta_{k}^{*}  - \theta_{j}}}{\zeta_j-\zeta_k^{*}}.
\end{equation}

For case 3, combine the results of case 1 and case 2 reasonably,  we can easily obtain the final soliton solutions.

\subsubsection{A variety of Rational Solutions and Physical Visions}
 For case 1, we usually are interested in the simple situation in $N_1=1$,
  \begin{equation} \label{sf4.1}
\left\{
 \begin{aligned}
   u_1(x,t)&=  2i \alpha_{1,1} e^{\theta_1 -\theta_1^{*}}(M^{-1})_{1,1} +   2i \alpha_{1,1} e^{\theta_1 -\theta_{1}}(M^{-1})_{1,2}   \\
     &  +  2i\alpha_{2,1}^{*} e^{\theta_{1}^{*} -\theta_{1}^{*}}(M^{-1})_{2,1}+  2i \alpha_{2,1}^{*} e^{\theta_{1}^{*}-\theta_{1}}(M^{-1})_{2,2},   \\
        u_2(x,t)&=    2i\alpha_{3,1} e^{\theta_1 -\theta_1^{*}}(M^{-1})_{1,1} +  2i  \alpha_{3,1} e^{\theta_1 -\theta_{1}}(M^{-1})_{1,2} \\
     &  +  2i\alpha_{4,1}^{*} e^{\theta_{1}^{*} -\theta_{1}^{*}}(M^{-1})_{2,1}+  2i \alpha_{4,1}^{*} e^{\theta_{1}^{*}-\theta_{1}}(M^{-1})_{2,2},   \\
             u_3(x,t)&=   2i \alpha_{5,1} e^{\theta_1 -\theta_1^{*}}(M^{-1})_{1,1} +  2i  \alpha_{5,1} e^{\theta_1 -\theta_{1}}(M^{-1})_{1,2}                 \\
     &  + 2i \alpha_{6,1}^{*} e^{\theta_{1}^{*} -\theta_{1}^{*}}(M^{-1})_{2,1}+  2i \alpha_{6,1}^{*} e^{\theta_{1}^{*}-\theta_{1}}(M^{-1})_{2,2},   \\
 \end{aligned}
\right.
 \end{equation}
where $\zeta_{1}=a_1 + i b_1$$(a_1 \neq 0, b_1 >0)$,  and  $\theta_1=-i (\zeta_1 x +4 \zeta_1^3 t)$.    If we take the parameters as $\alpha_{2,1}=\alpha_{1,1}^{*}$, $\alpha_{4,1}=\alpha_{3,1}^{*}$, $\alpha_{6,1}=\alpha_{5,1}^{*}$, then the breather-type  soliton solutions  are obtained
 \begin{equation}
\left\{
    \begin{aligned}
     u_1(x,t)= \frac{-2\sqrt{2}\alpha_{1,1}a_1  b_1}{\sqrt{|\alpha_{1,1}|^2+|\alpha_{3,1}|^2+|\alpha_{5,1}|^2}}\frac{a_1 \cosh X_1 \cos Y_1+ b_1 \sinh X_1 \sin Y_1}{a_1^2 \cosh^2 X_1 + b_1^2 \sin^2 Y_1},  \\
        u_2(x,t)= \frac{-2\sqrt{2}\alpha_{3,1}a_1  b_1}{\sqrt{|\alpha_{1,1}|^2+|\alpha_{3,1}|^2+|\alpha_{5,1}|^2}}\frac{a_1 \cosh X_1 \cos Y_1+ b_1 \sinh X_1 \sin Y_1}{a_1^2 \cosh^2 X_1 + b_1^2 \sin^2 Y_1}, \\
           u_3(x,t)= \frac{-2\sqrt{2}\alpha_{5,1}a_1  b_1}{\sqrt{|\alpha_{1,1}|^2+|\alpha_{3,1}|^2+|\alpha_{5,1}|^2}}\frac{a_1 \cosh X_1 \cos Y_1+ b_1 \sinh X_1 \sin Y_1}{a_1^2 \cosh^2 X_1 + b_1^2 \sin^2 Y_1},
     \end{aligned}
\right.
 \end{equation}
 where $X_1 =  2 b_1 [x+4(3 a_1^2- b_1^2)t] - \ln \sqrt{|\alpha_{1,1}|^2+|\alpha_{3,1}|^2+|\alpha_{5,1}|^2},   Y_1 =  2 a_1 [x+4(a_1^2-3 b_1^2)t]$. The breather-type soliton solutions  are plotted in Fig. 7.

\noindent
{\rotatebox{0}{\includegraphics[width=3.6cm,height=3.0cm,angle=0]{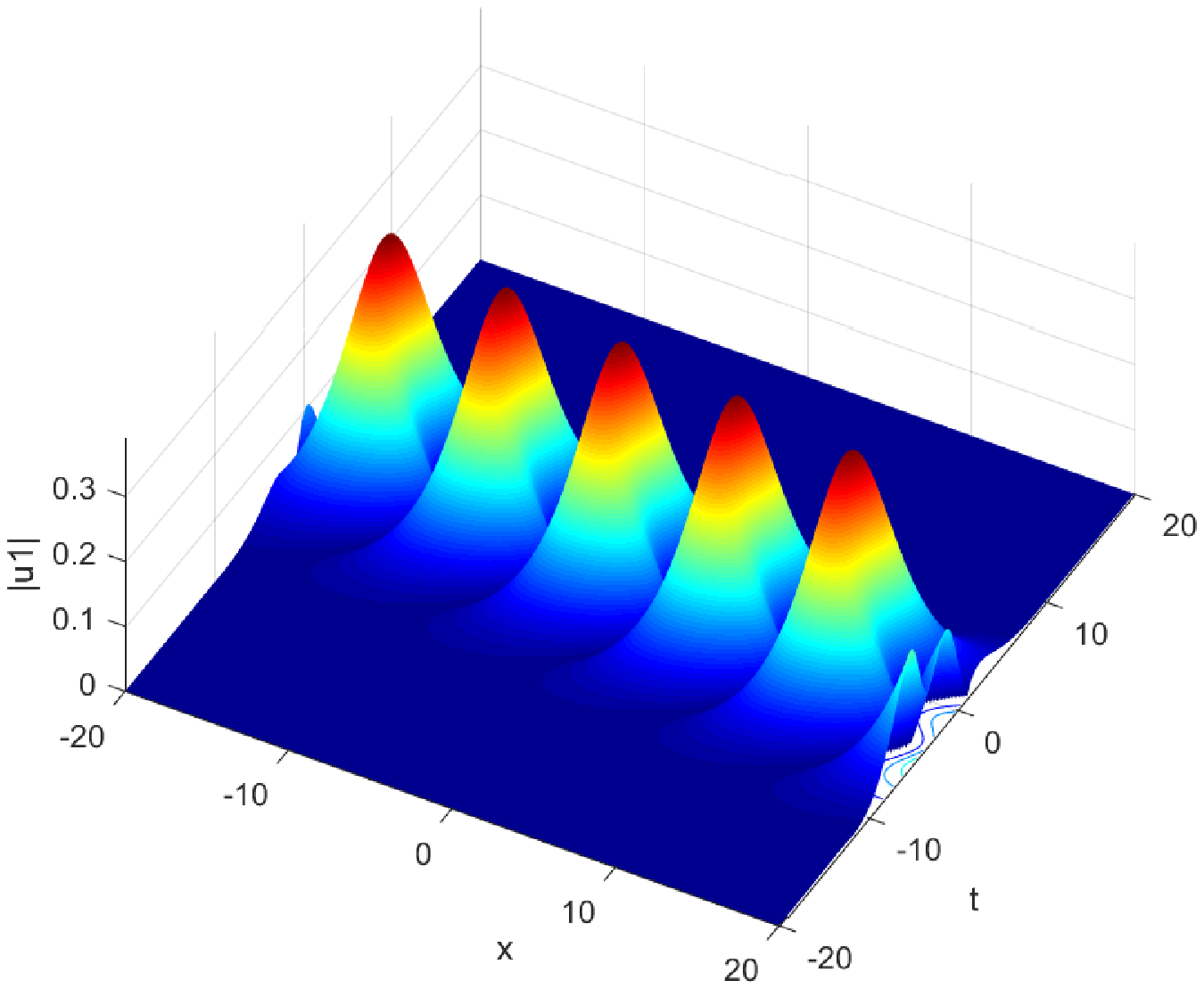}}}
~~~~
{\rotatebox{0}{\includegraphics[width=3.6cm,height=3.0cm,angle=0]{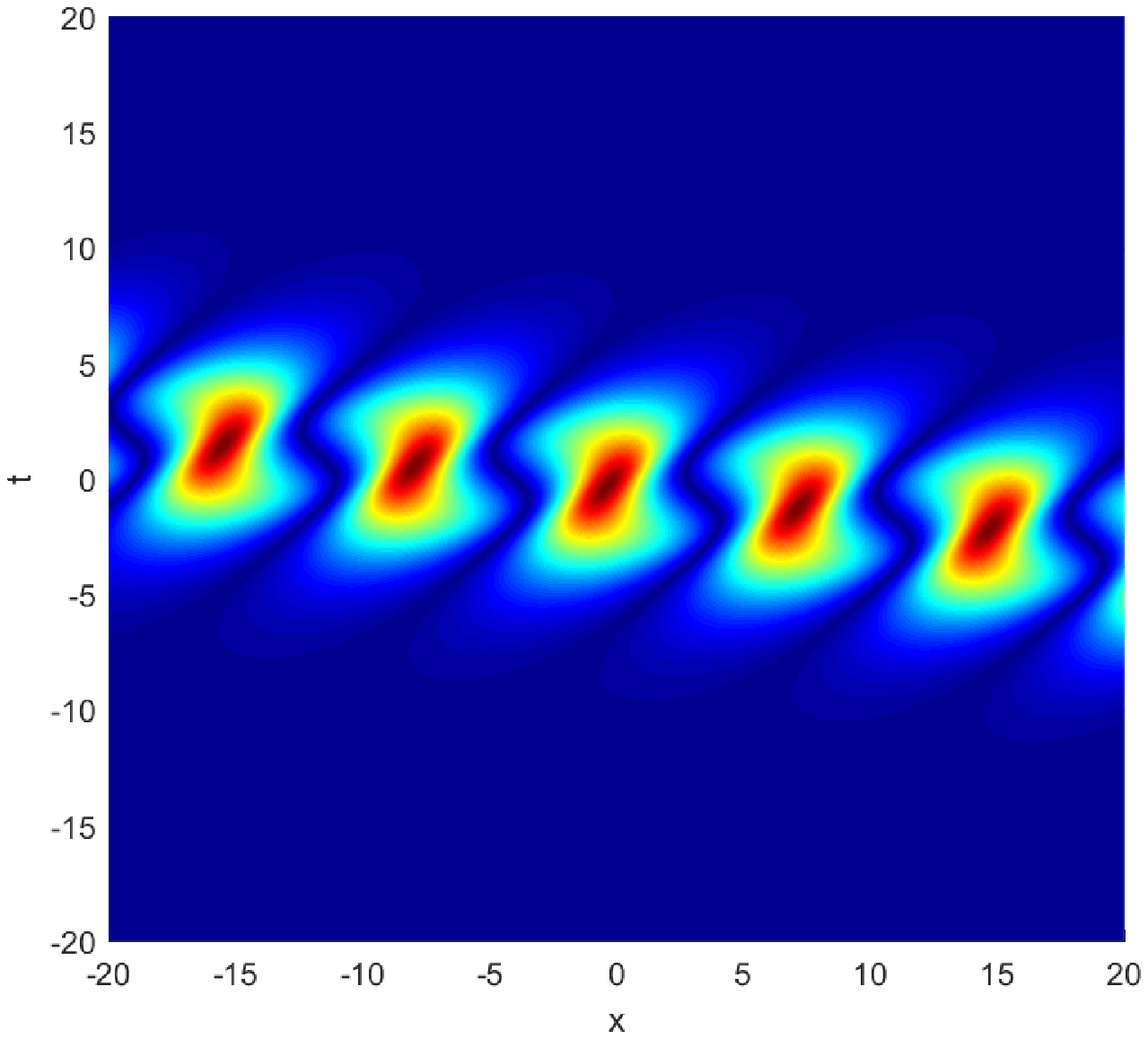}}}
~~~~
{\rotatebox{0}{\includegraphics[width=3.6cm,height=3.0cm,angle=0]{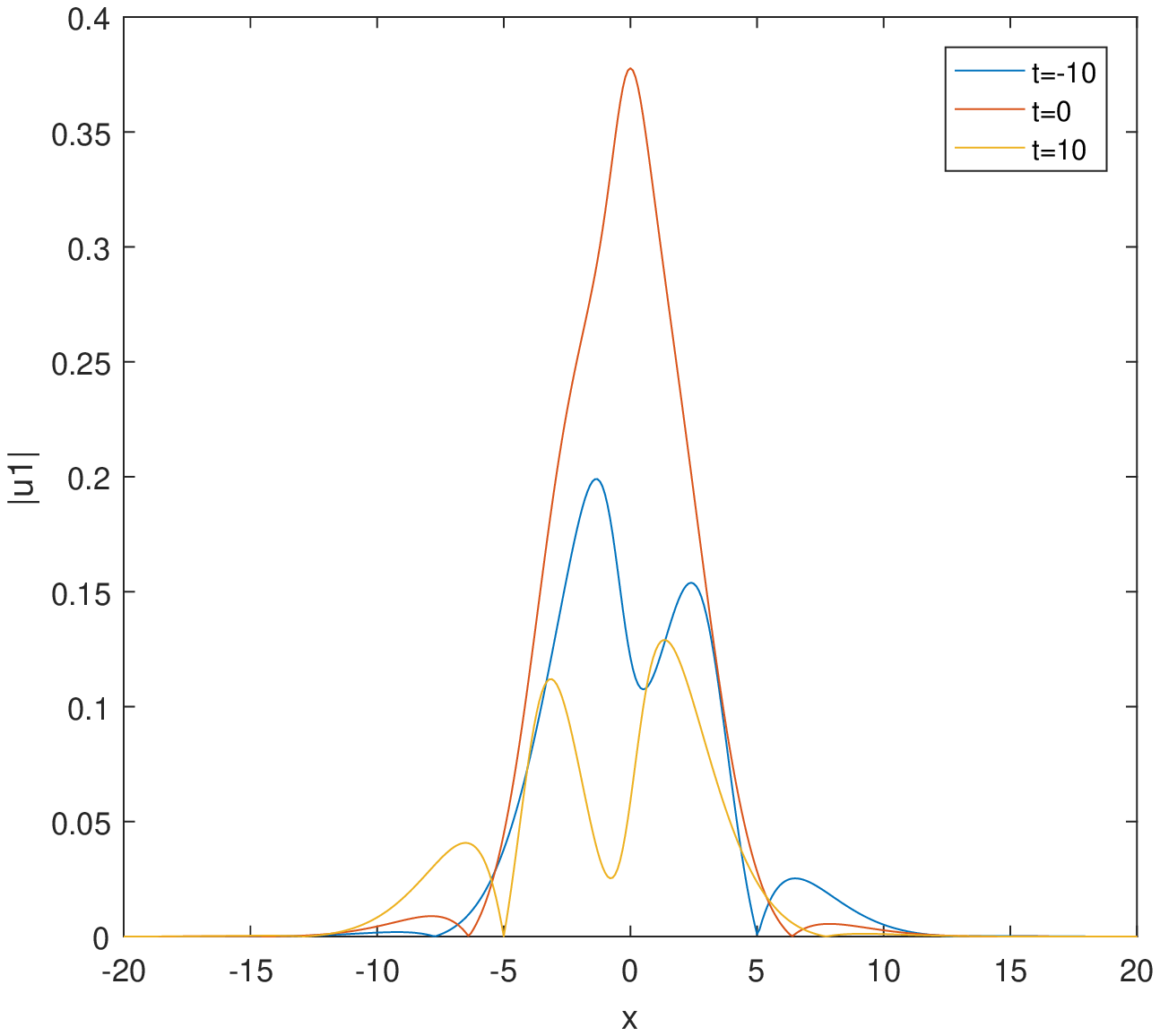}}}

$\ \qquad~~~~~~(\textbf{a})\qquad \ \qquad\qquad\qquad\qquad~(\textbf{b})
\ \qquad\qquad\qquad\qquad\qquad~(\textbf{c})$\\
\noindent
{\rotatebox{0}{\includegraphics[width=3.6cm,height=3.0cm,angle=0]{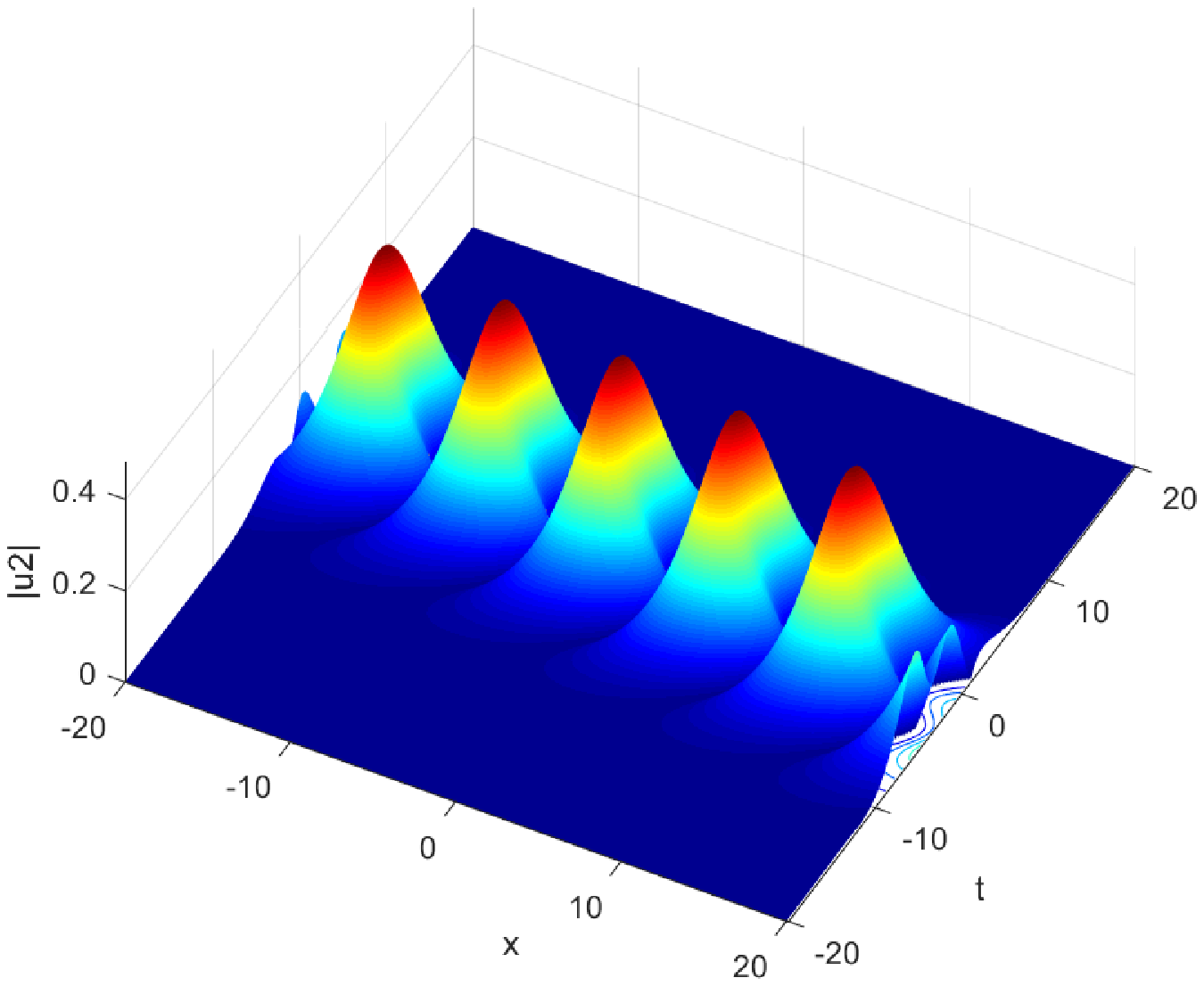}}}
~~~~
{\rotatebox{0}{\includegraphics[width=3.6cm,height=3.0cm,angle=0]{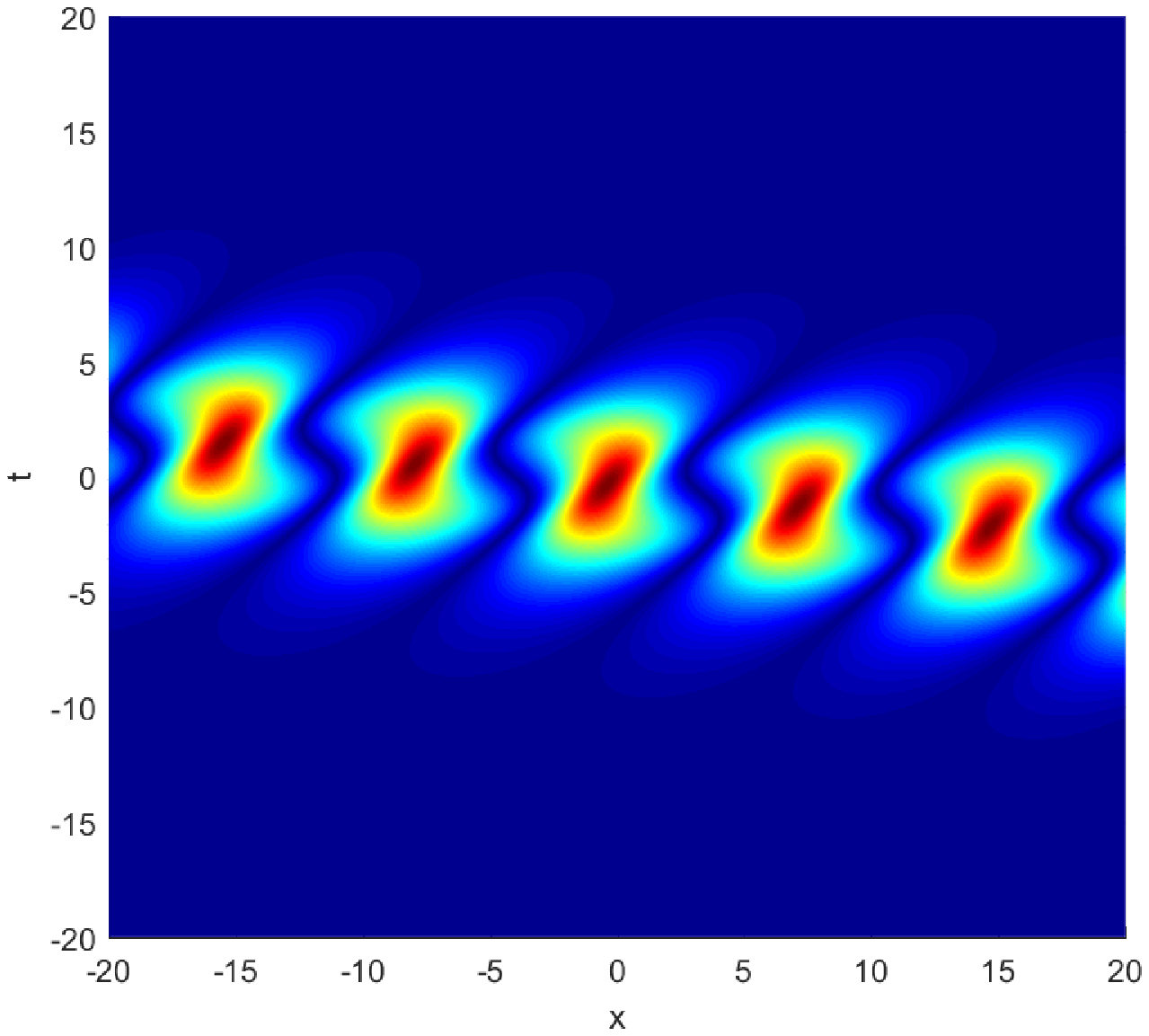}}}
~~~~
{\rotatebox{0}{\includegraphics[width=3.6cm,height=3.0cm,angle=0]{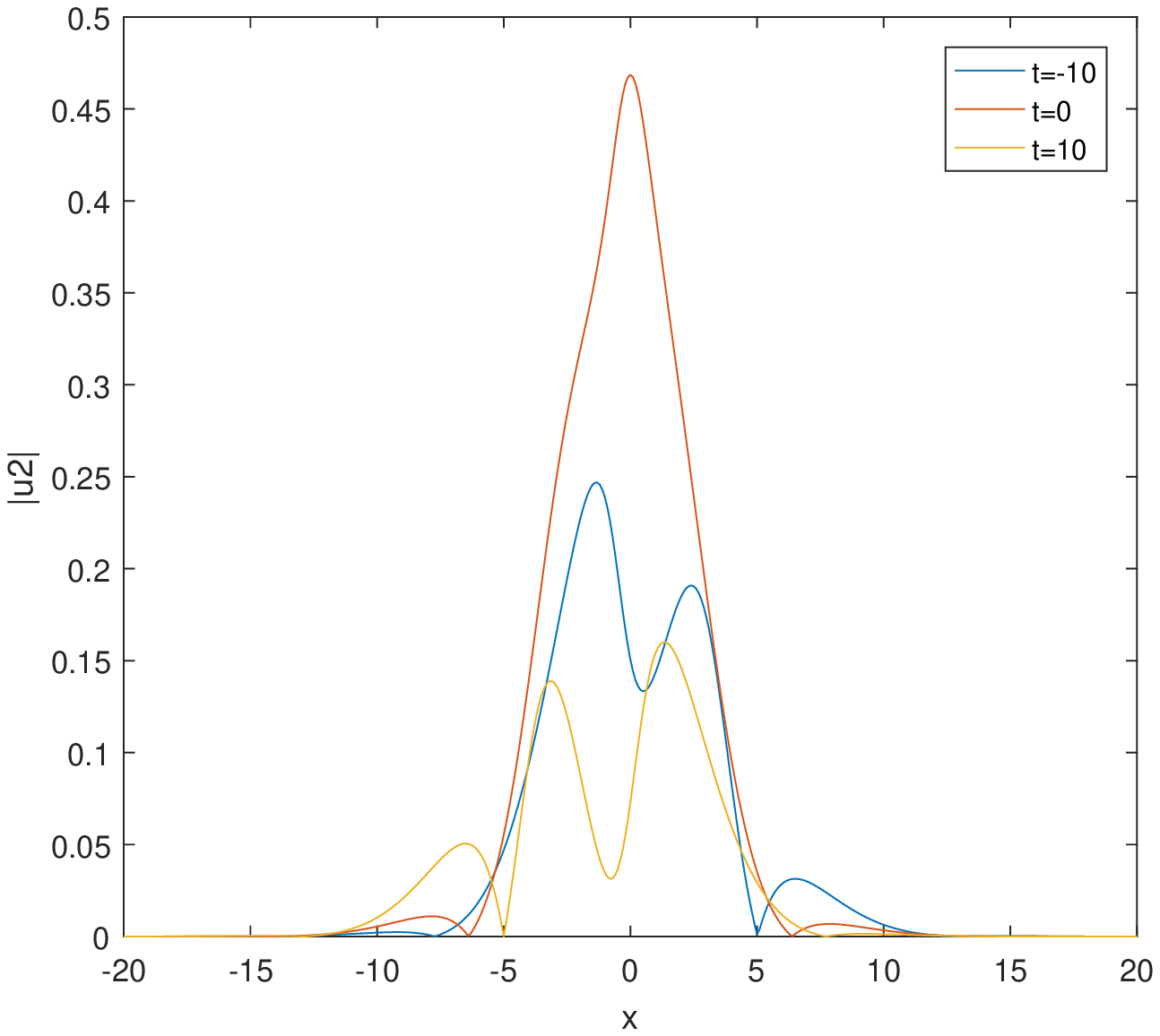}}}

$\ \qquad~~~~~~(\textbf{d})\qquad \ \qquad\qquad\qquad\qquad~(\textbf{e})
\ \qquad\qquad\qquad\qquad\qquad~(\textbf{f})$\\
\noindent
{\rotatebox{0}{\includegraphics[width=3.6cm,height=3.0cm,angle=0]{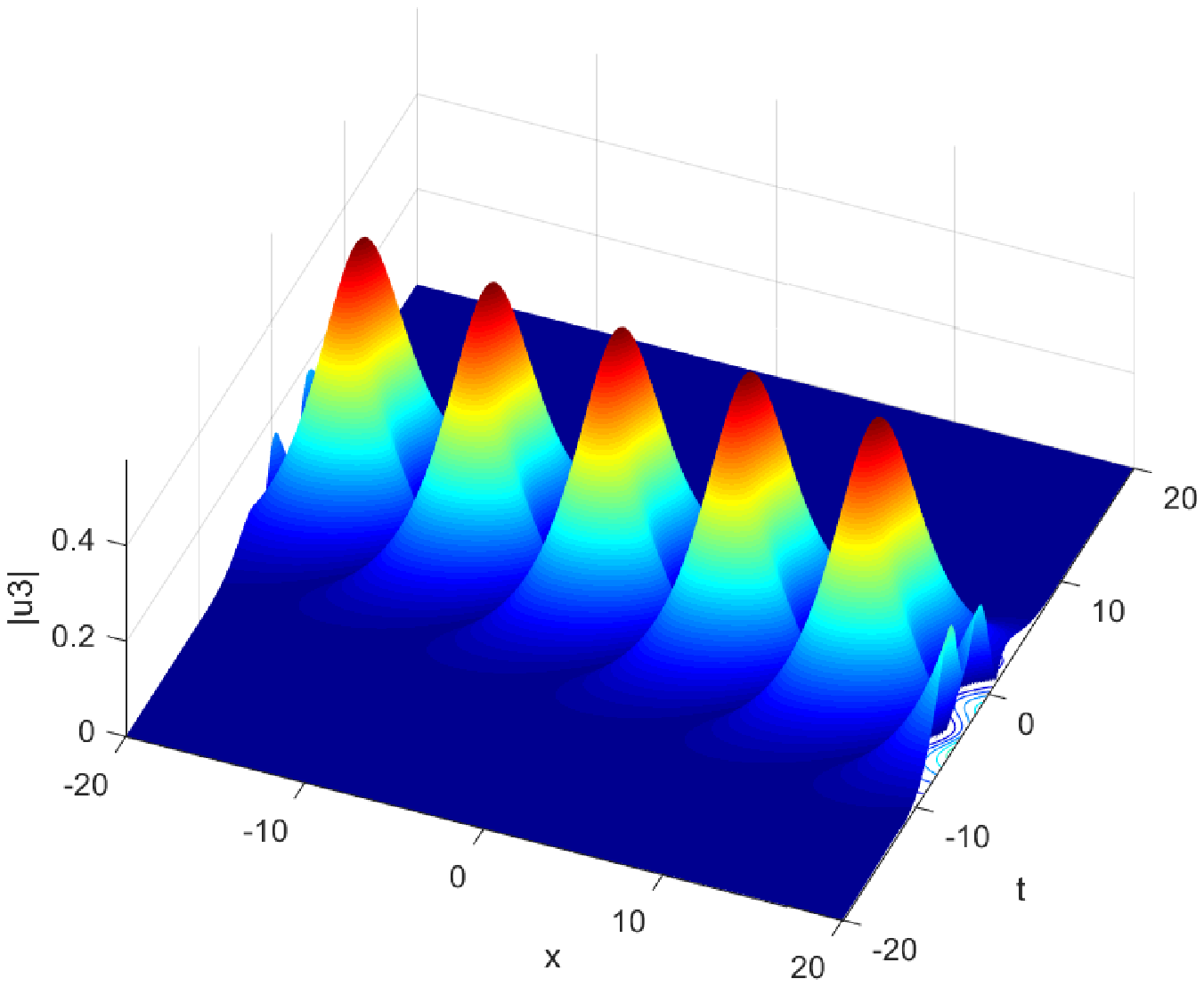}}}
~~~~
{\rotatebox{0}{\includegraphics[width=3.6cm,height=3.0cm,angle=0]{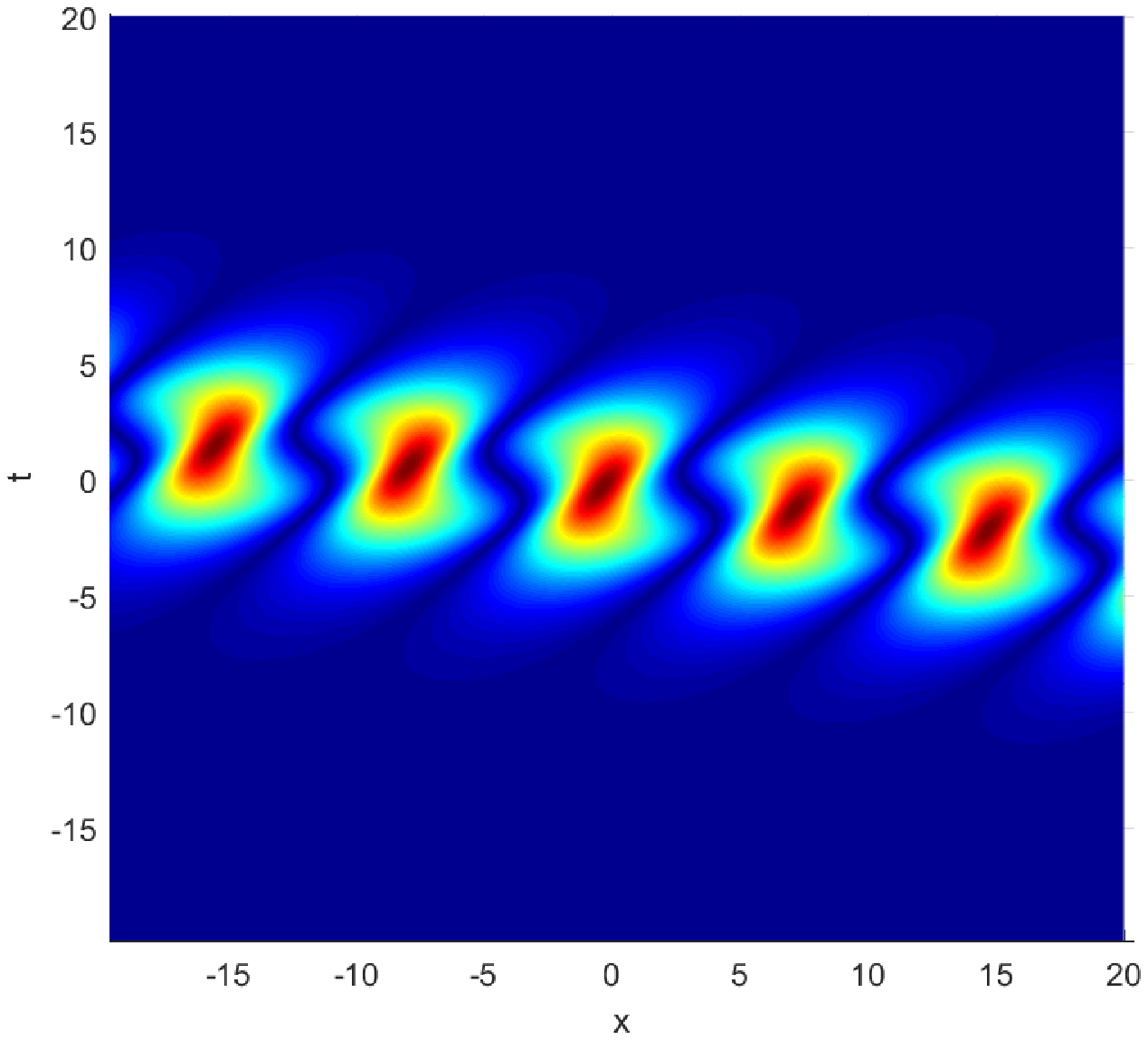}}}
~~~~
{\rotatebox{0}{\includegraphics[width=3.6cm,height=3.0cm,angle=0]{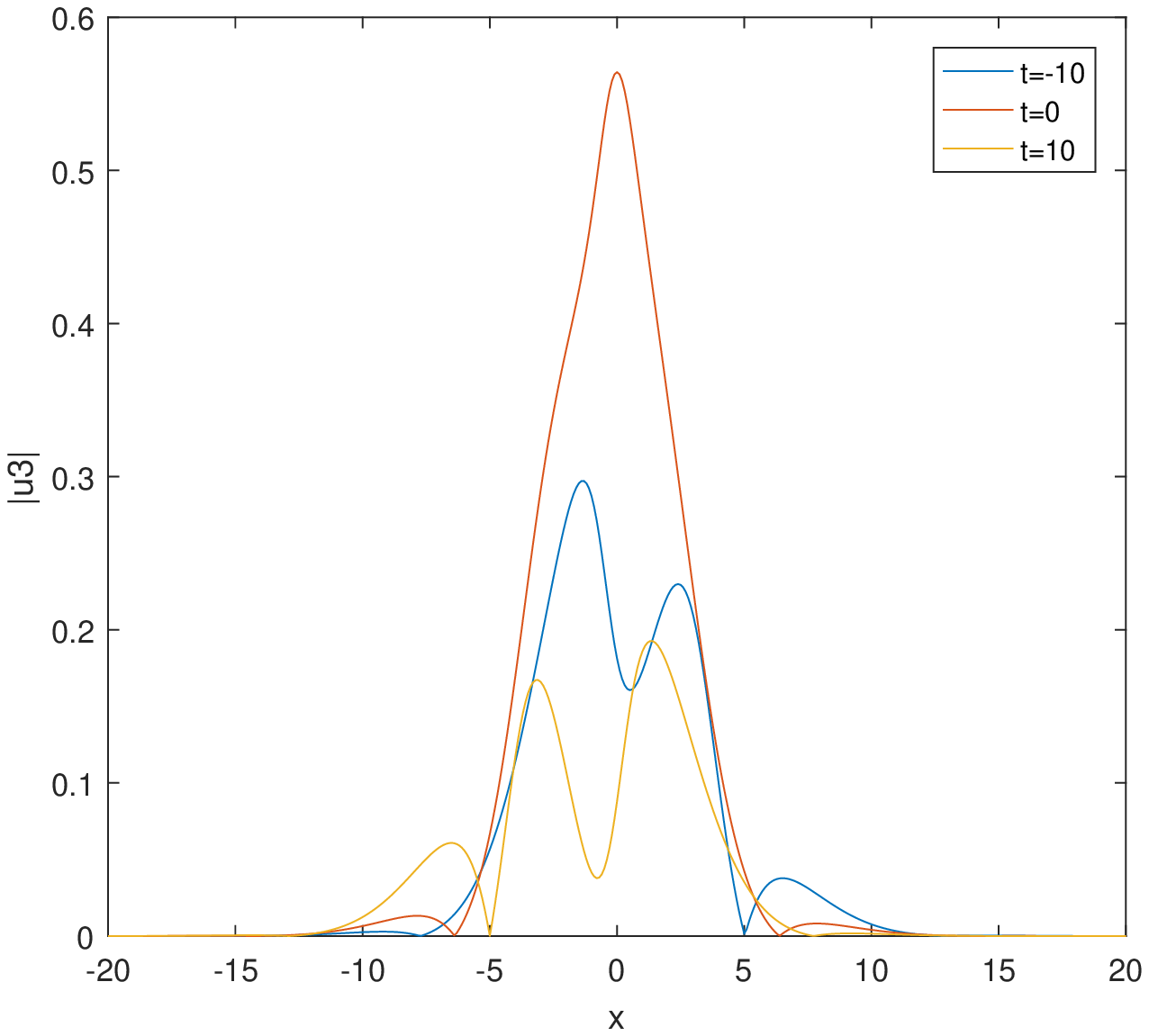}}}

$\ \qquad~~~~~~(\textbf{g})\qquad \ \qquad\qquad\qquad\qquad~(\textbf{h})
\ \qquad\qquad\qquad\qquad\qquad~(\textbf{i})$\\
\noindent { \small \textbf{Figure 7.} (Color online) Single-breather solutions to Eq. \eqref{sf4.1}   with the parameters   $\alpha_{1,1}=0.2+0.3 i$, $\alpha_{3,1}=0.2+0.4 i$, $\alpha_{5,1}=0.2+0.5 i$, $a_1=0.2$,  $b_1=0.3$.
\label{fig4.11}
$\textbf{(a)(d)(g)}$: the structures of the single-breather solutions,
$\textbf{(b)(e)(h)}$: the density plot,
$\textbf{(c)(f)(i)}$: the wave propagation of the single-breather solutions.}\\

For case 2, we usually are interested in the simple situation in $N_2=1$,
 \begin{equation} \label{solu4.2}
\left\{
 \begin{aligned}
   u_1(x,t)&= 2i  \frac{\alpha_{1,1} e^{\theta_1 -\theta_1^{*}}(\zeta_{1}-\zeta_{1}^{*}) }{M_{1,1}^{(1)}e^{\theta_{1}^{*}+\theta_1}+e^{-\theta_{1}^{*}-\theta_1}},\\
   u_2(x,t)&= 2i  \frac{\alpha_{3,1} e^{\theta_1 -\theta_1^{*}}(\zeta_{1}-\zeta_{1}^{*}) }{M_{1,1}^{(1)}e^{\theta_{1}^{*}+\theta_1}+e^{-\theta_{1}^{*}-\theta_1}},  \\
   u_3(x,t)&= 2i  \frac{\alpha_{5,1} e^{\theta_1 -\theta_1^{*}}(\zeta_{1}-\zeta_{1}^{*}) }{M_{1,1}^{(1)}e^{\theta_{1}^{*}+\theta_1}+e^{-\theta_{1}^{*}-\theta_1}},\\
 \end{aligned}
\right.
 \end{equation}
where $\zeta_{1}= i b_1(b_1>0)$, and $\theta_1= -i(\zeta_1 x +4 \zeta_1^3 t)$.  If we set $\alpha_{2,1}=\alpha_{1,1}^{*}$, $\alpha_{4,1}=\alpha_{3,1}^{*}$, and $\alpha_{6,1}=\alpha_{5,1}^{*}$, then Eq. \eqref{solu4.2} can be converted to the following form
\begin{equation}\label{sf4.2}
\left\{
   \begin{aligned}
          u_1(x,t)= \frac{-\sqrt{2}\alpha_{1,1}b_1}{\sqrt{|\alpha_{1,1}|^2+|\alpha_{3,1}|^2+|\alpha_{5,1}|^2}} \mathrm{sech} \left( 2 b_1 x -8 b_1^3 t + \ln \sqrt{2(|\alpha_{1,1}|^2+|\alpha_{3,1}|^2+|\alpha_{5,1}|^2)}   \right),  \\
          u_2(x,t)= \frac{-\sqrt{2}\alpha_{3,1}b_1}{\sqrt{|\alpha_{1,1}|^2+|\alpha_{3,1}|^2+|\alpha_{5,1}|^2}} \mathrm{sech} \left(2 b_1 x -8 b_1^3 t + \ln \sqrt{2(|\alpha_{1,1}|^2+|\alpha_{3,1}|^2+|\alpha_{5,1}|^2)}\right),  \\
          u_3(x,t)= \frac{-\sqrt{2}\alpha_{5,1}b_1}{\sqrt{|\alpha_{1,1}|^2+|\alpha_{3,1}|^2+|\alpha_{5,1}|^2}} \mathrm{sech} \left(2 b_1 x -8 b_1^3 t + \ln \sqrt{2(|\alpha_{1,1}|^2+|\alpha_{3,1}|^2+|\alpha_{5,1}|^2)}  \right).
   \end{aligned}
\right.
\end{equation}

In Fig. 8,  we display the localized structures and dynamic behaviors of one-bell solutions vividly.

\noindent
{\rotatebox{0}{\includegraphics[width=3.6cm,height=3.0cm,angle=0]{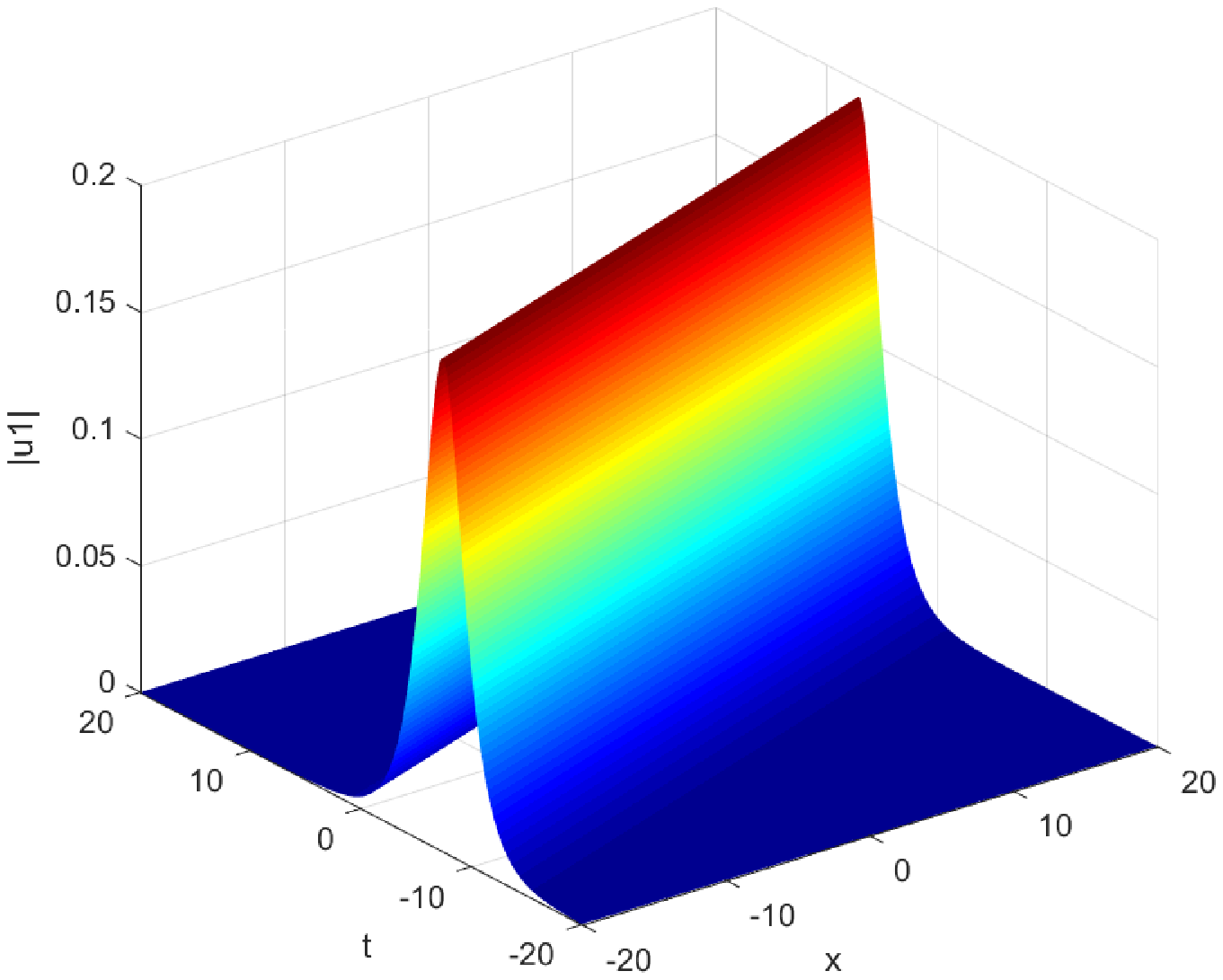}}}
~~~~
{\rotatebox{0}{\includegraphics[width=3.6cm,height=3.0cm,angle=0]{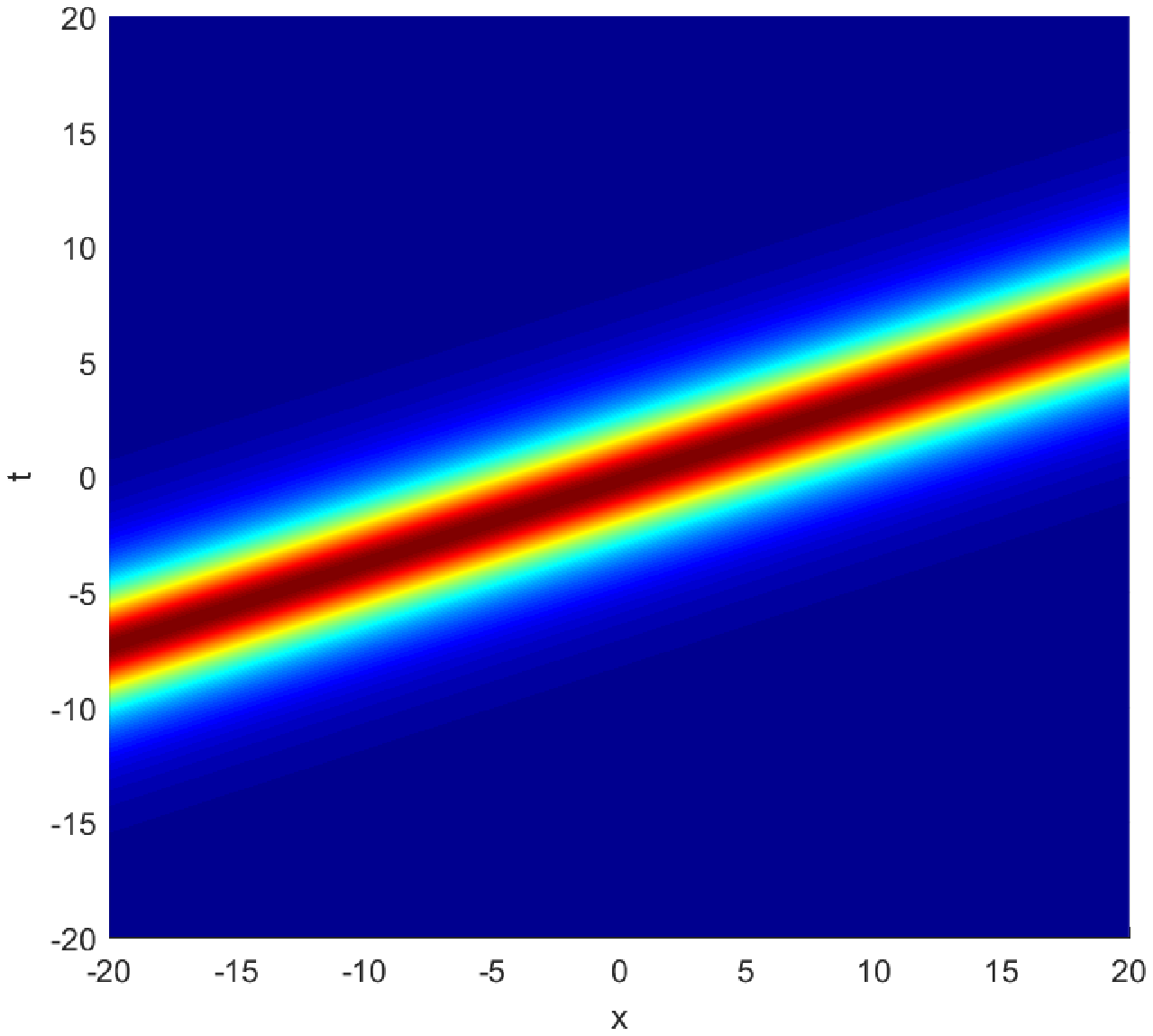}}}
~~~~
{\rotatebox{0}{\includegraphics[width=3.6cm,height=3.0cm,angle=0]{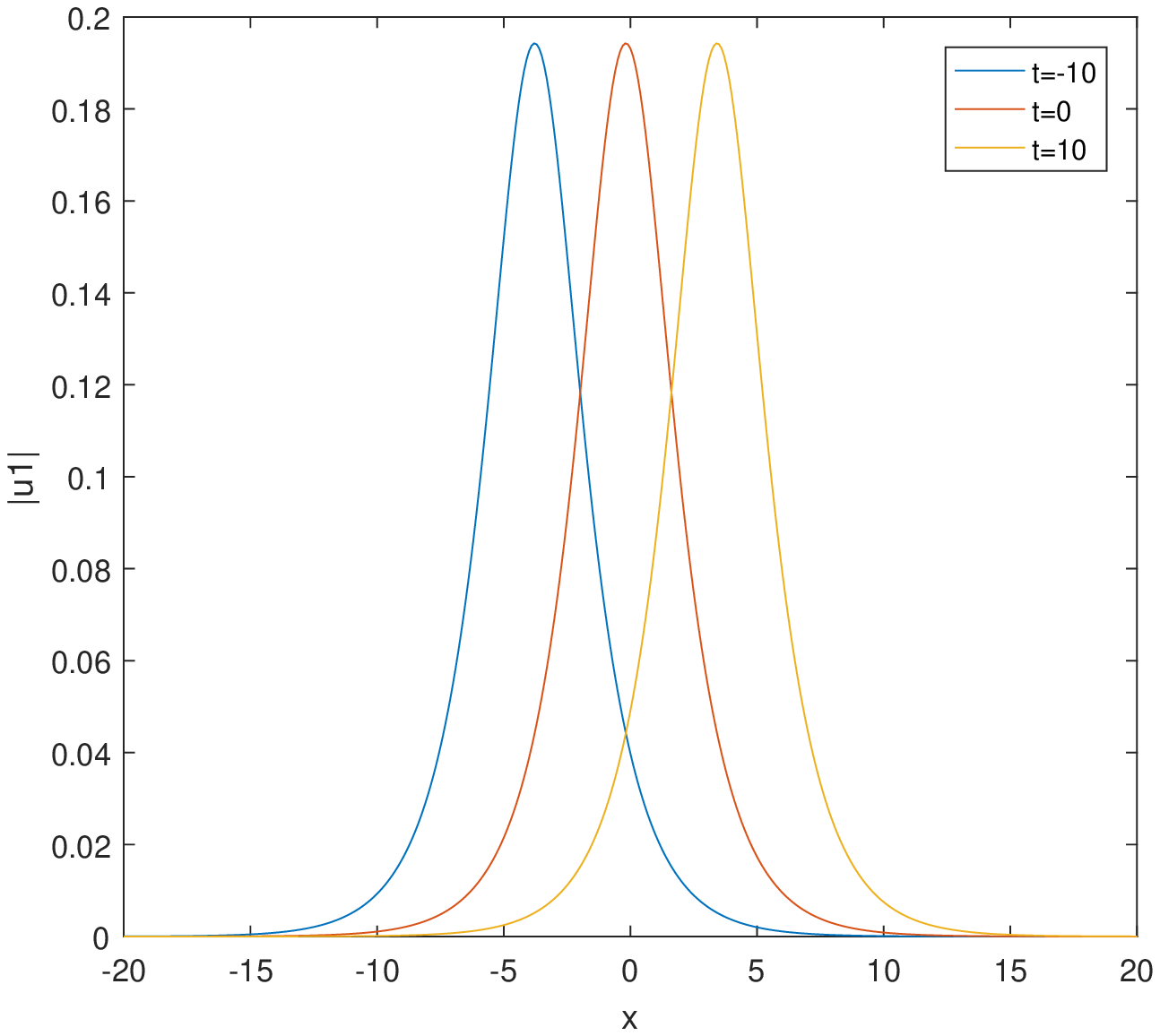}}}

$\ \qquad~~~~~~(\textbf{a})\qquad \ \qquad\qquad\qquad\qquad~(\textbf{b})
\ \qquad\qquad\qquad\qquad\qquad~(\textbf{c})$\\
\noindent
{\rotatebox{0}{\includegraphics[width=3.6cm,height=3.0cm,angle=0]{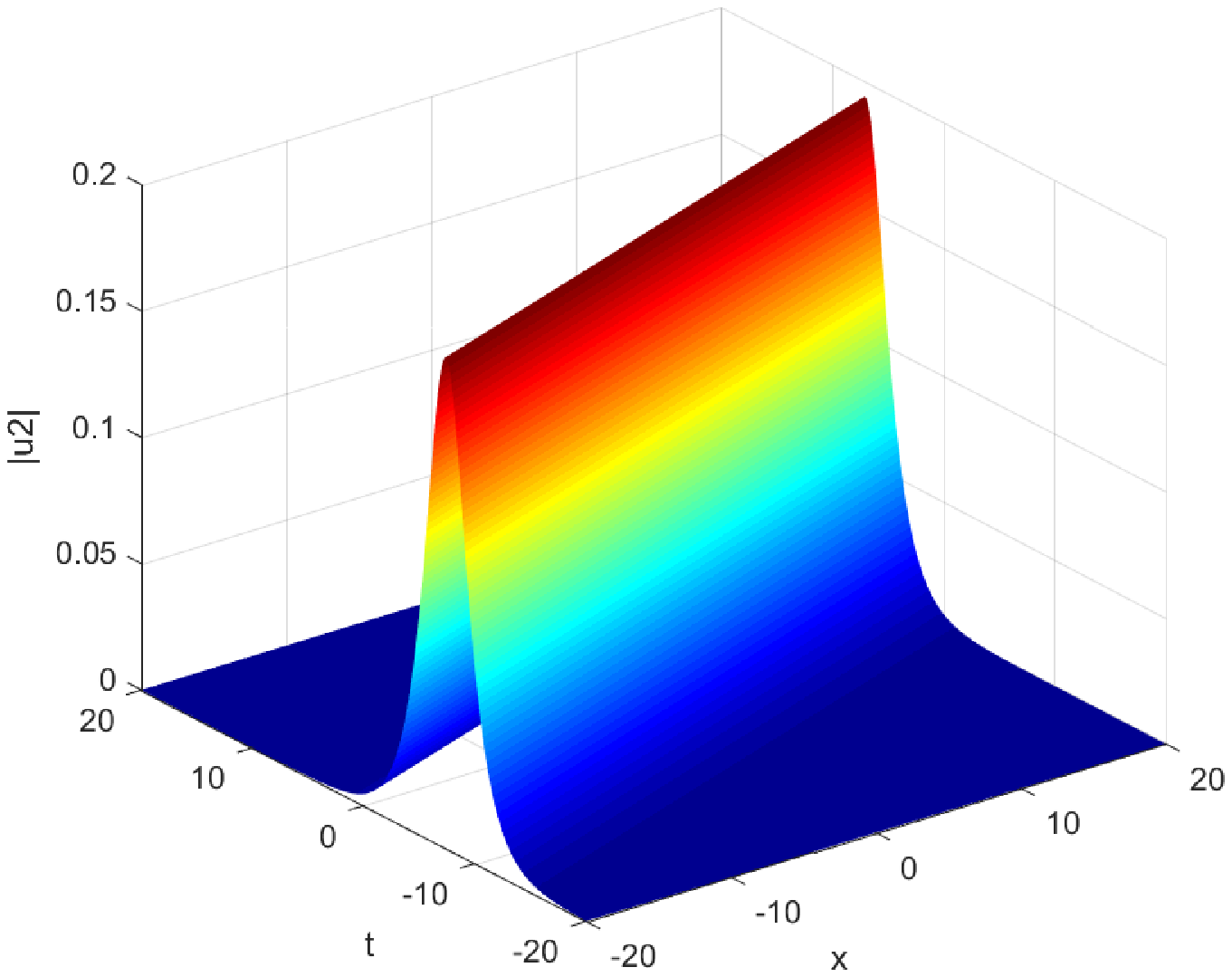}}}
~~~~
{\rotatebox{0}{\includegraphics[width=3.6cm,height=3.0cm,angle=0]{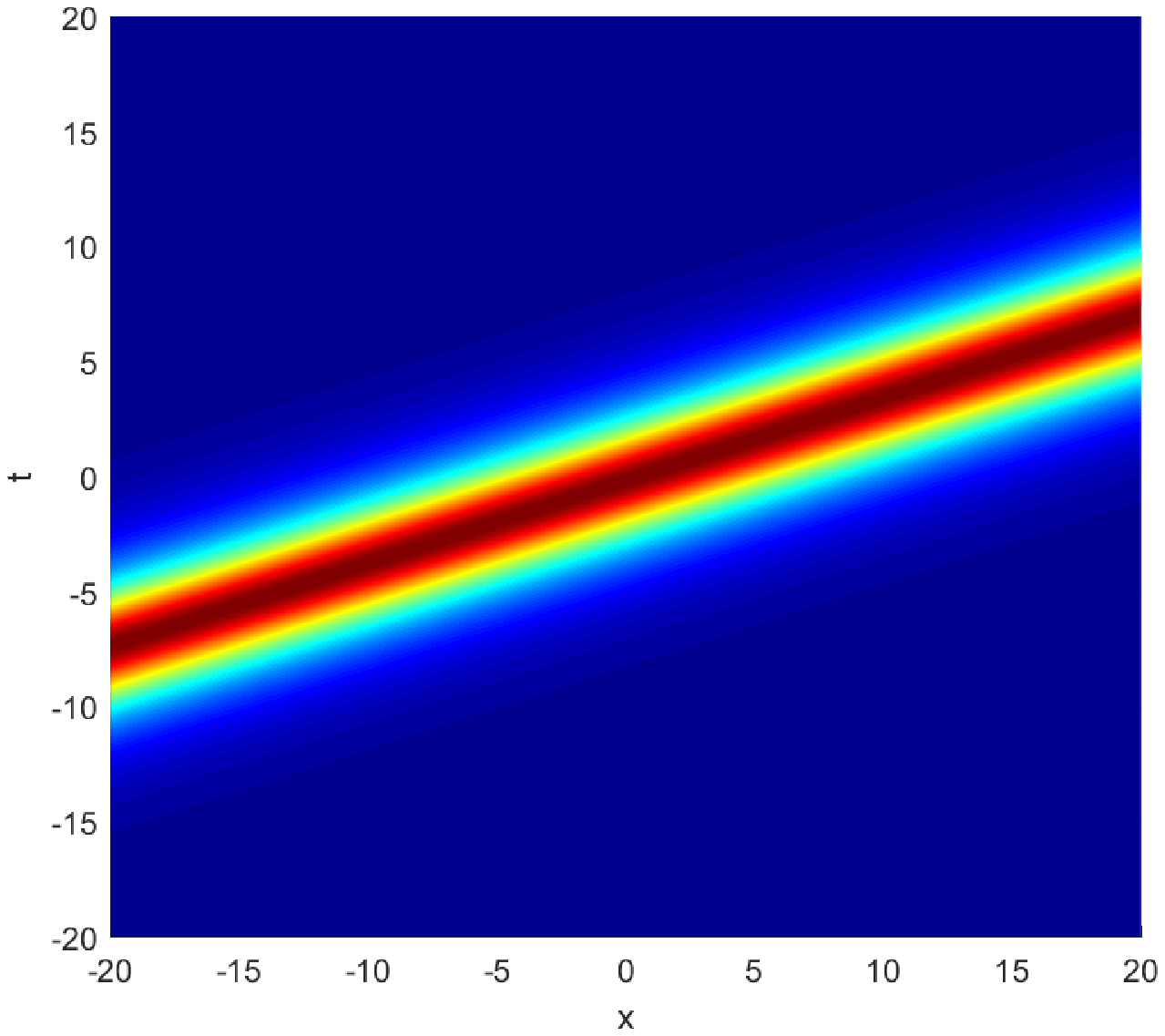}}}
~~~~
{\rotatebox{0}{\includegraphics[width=3.6cm,height=3.0cm,angle=0]{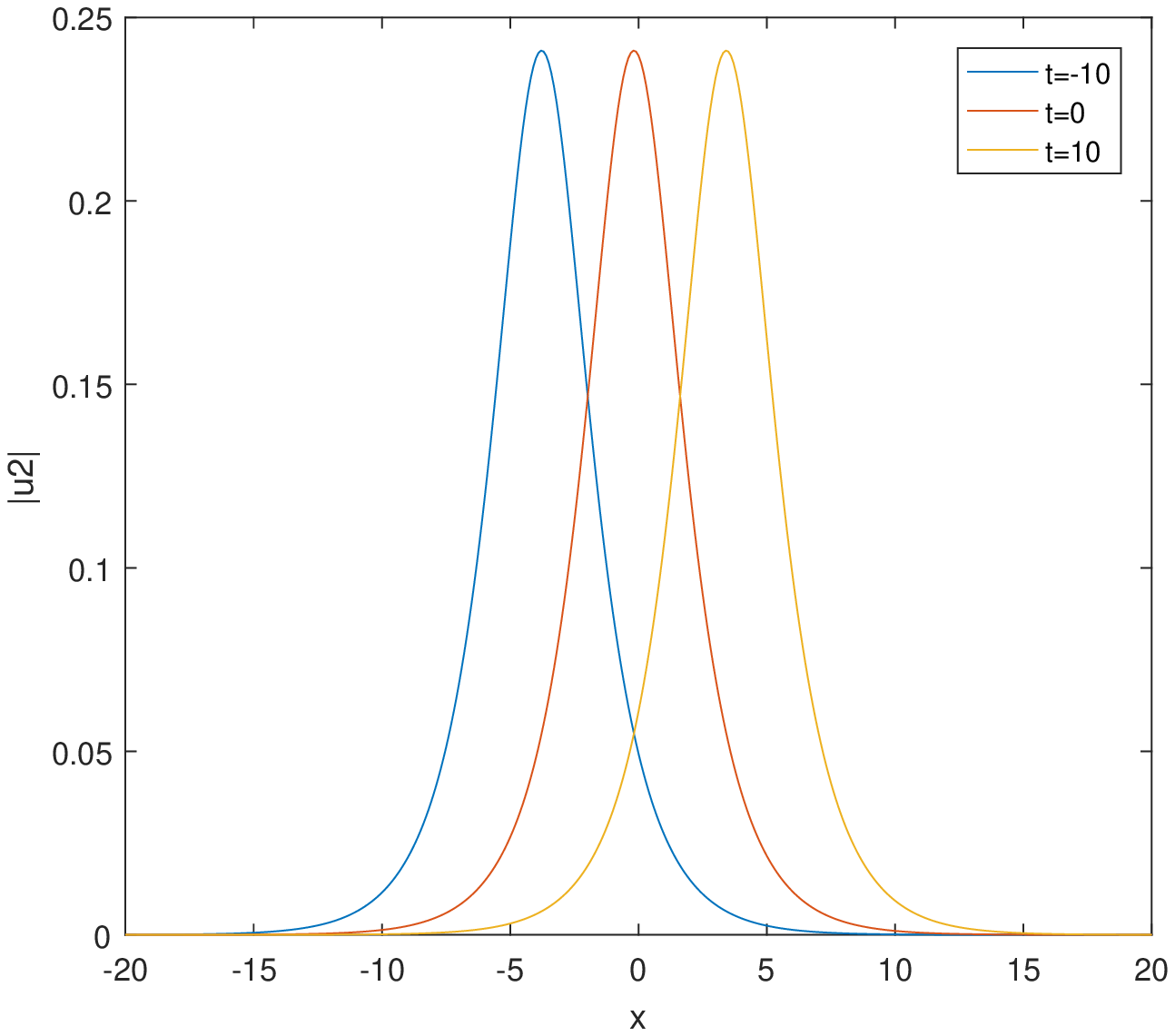}}}

$\ \qquad~~~~~~(\textbf{d})\qquad \ \qquad\qquad\qquad\qquad~(\textbf{e})
\ \qquad\qquad\qquad\qquad\qquad~(\textbf{f})$\\
\noindent
{\rotatebox{0}{\includegraphics[width=3.6cm,height=3.0cm,angle=0]{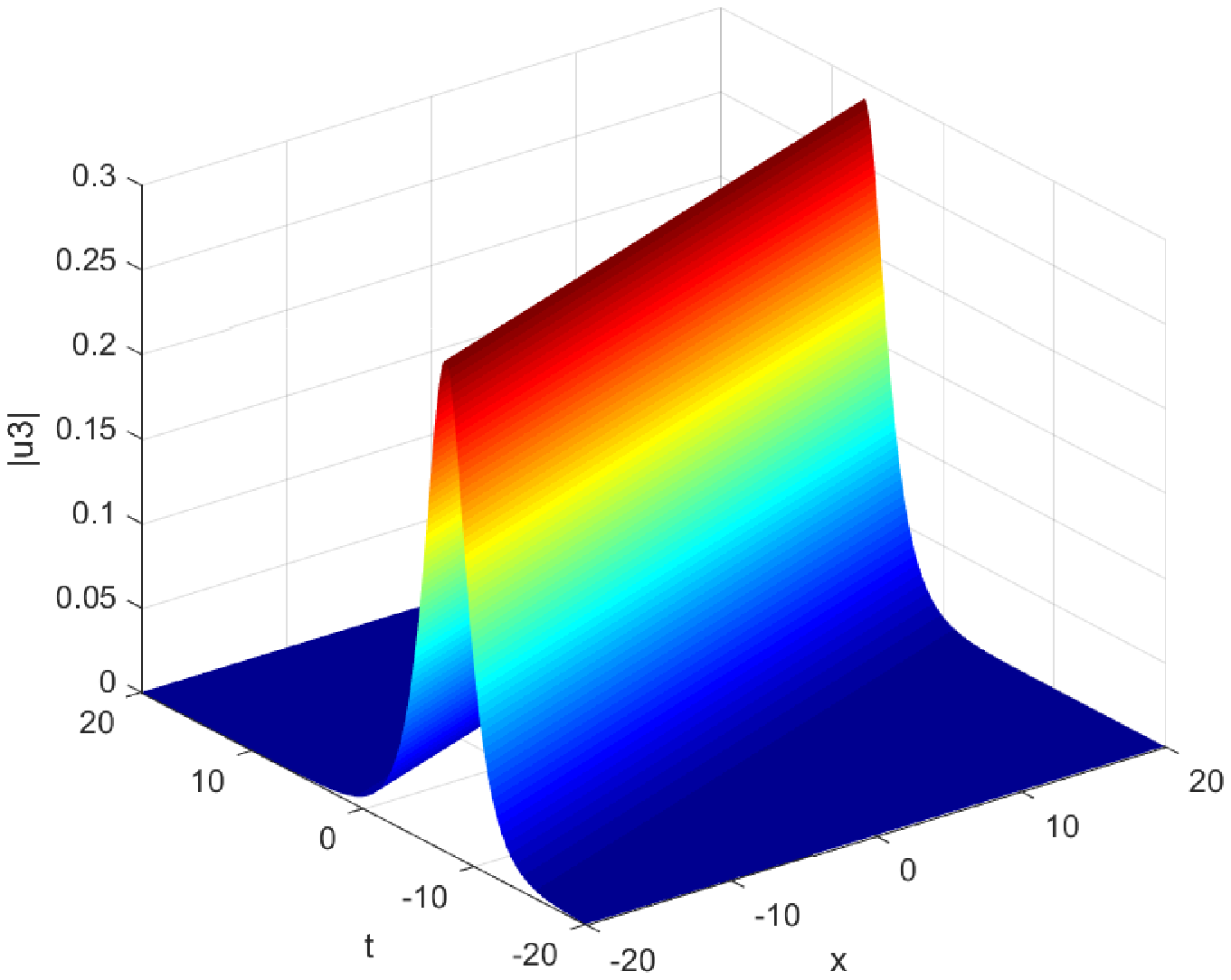}}}
~~~~
{\rotatebox{0}{\includegraphics[width=3.6cm,height=3.0cm,angle=0]{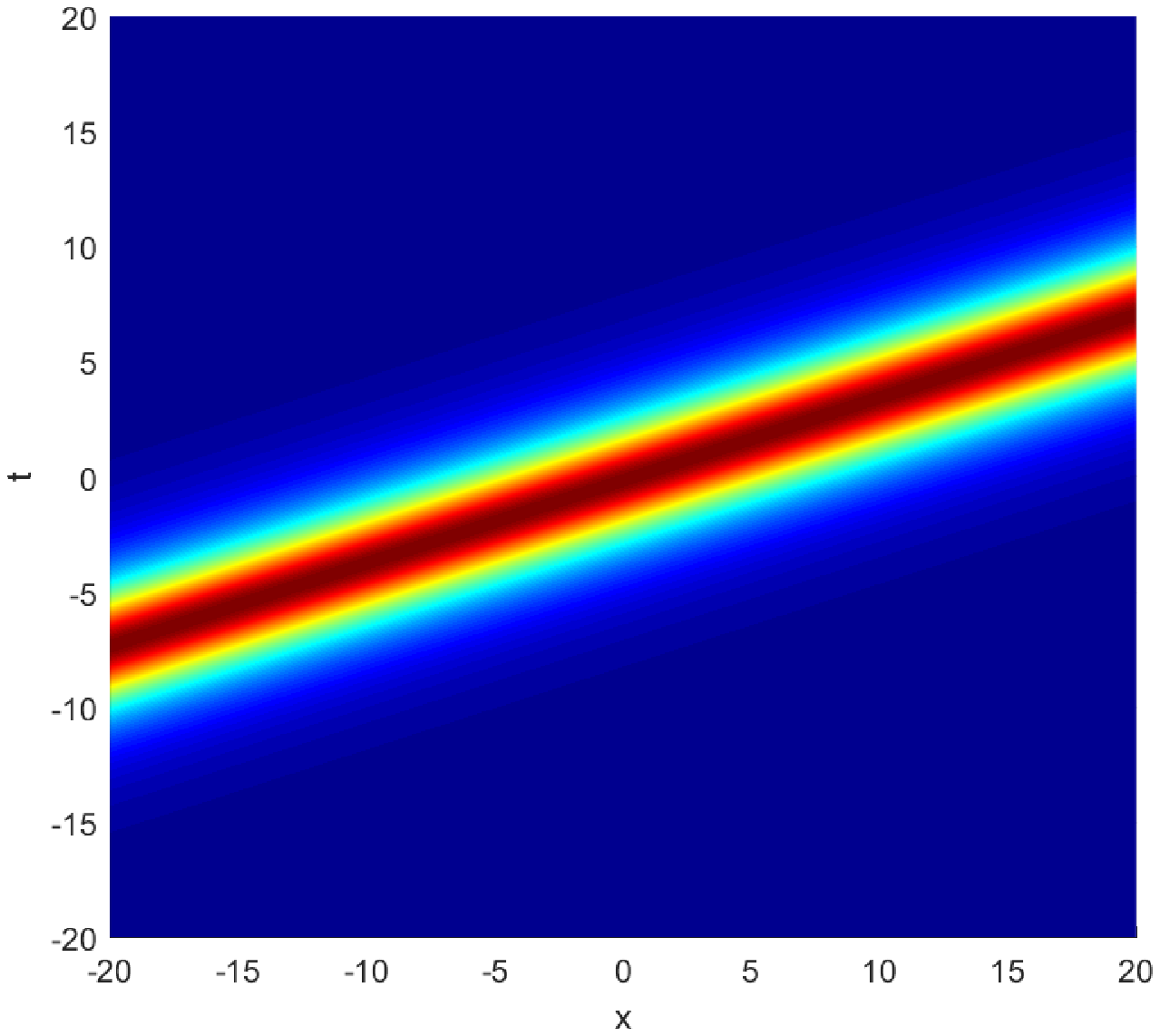}}}
~~~~
{\rotatebox{0}{\includegraphics[width=3.6cm,height=3.0cm,angle=0]{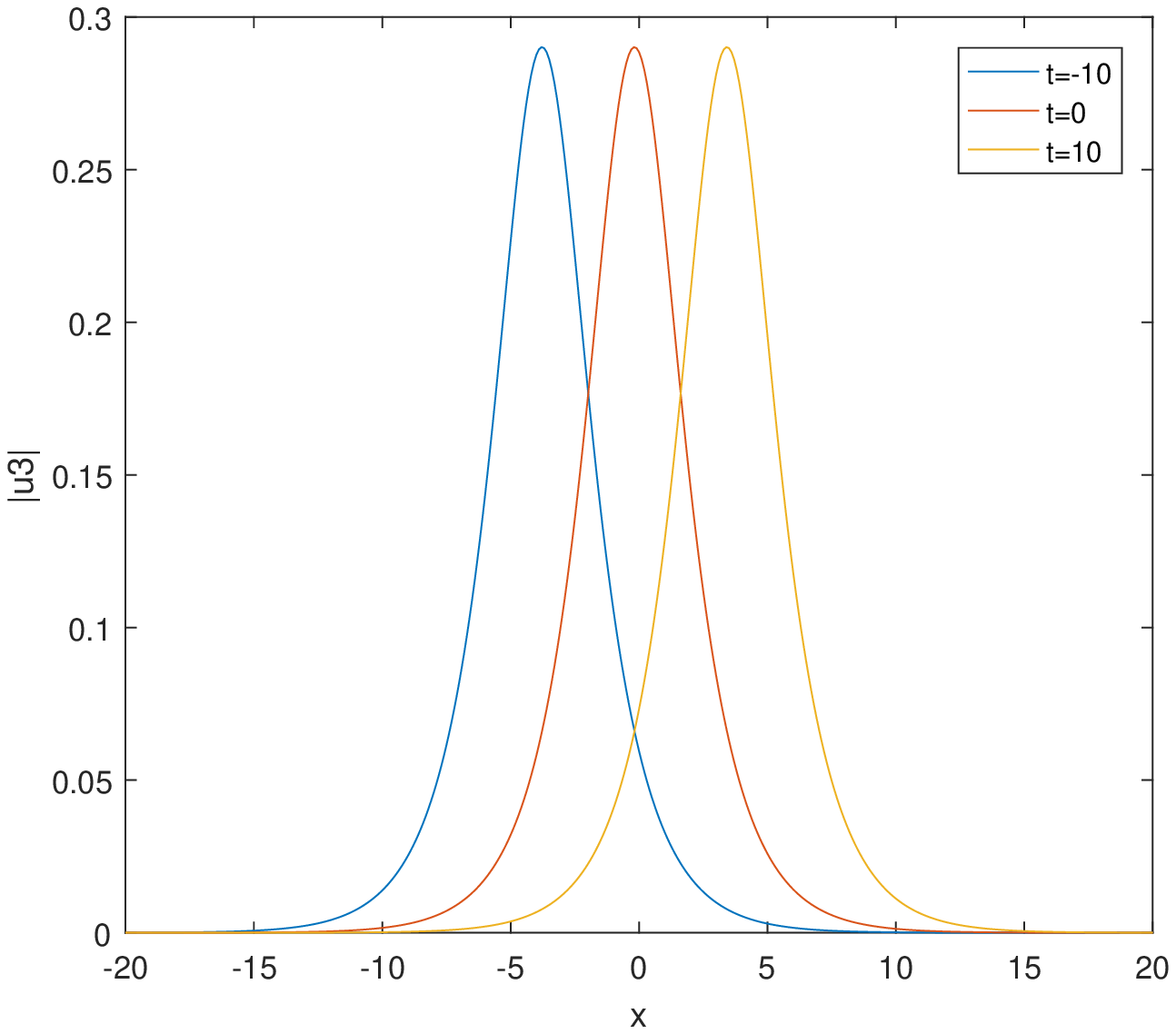}}}

$\ \qquad~~~~~~(\textbf{g})\qquad \ \qquad\qquad\qquad\qquad~(\textbf{h})
\ \qquad\qquad\qquad\qquad\qquad~(\textbf{i})$\\
\noindent { \small \textbf{Figure 8.} (Color online) One-bell solutions  to Eq. \eqref{sf4.2}  with the parameters   $\alpha_{1,1}=0.2+0.3 i$, $\alpha_{3,1}=0.2+0.4 i$, $\alpha_{5,1}=0.2+0.5 i$, $a_1=0$,  $b_1=0.3$.
\label{fig4.21}
$\textbf{(a)(d)(g)}$: the structures of the one-bell solutions,
$\textbf{(b)(e)(h)}$: the density plot,
$\textbf{(c)(f)(i)}$: the wave propagation of the one-bell solutions.} \\

When  taking $N_2=2$,  two-bell soliton solutions are as follows
 \begin{equation} \label{sf4.3}
\left\{
 \begin{aligned}
   u_1(x,t)&= 2i \alpha_{1,1}e^{\theta_{1}-\theta_{1}^{*}}(M^{-1})_{1,1}+2i \alpha_{1,1}e^{\theta_{1}-\theta_{2}^{*}}(M^{-1})_{1,2}\\
        & +2i \alpha_{1,2}e^{\theta_{2}-\theta_{1}^{*}}(M^{-1})_{2,1}+2i \alpha_{1,2}e^{\theta_{2}-\theta_{2}^{*}}(M^{-1})_{2,2},  \\
   u_2(x,t)&= 2i \alpha_{3,1}e^{\theta_{1}-\theta_{1}^{*}}(M^{-1})_{1,1}+2i \alpha_{3,1}e^{\theta_{1}-\theta_{2}^{*}}(M^{-1})_{1,2}\\
       &   +2i \alpha_{3,2}e^{\theta_{2}-\theta_{1}^{*}}(M^{-1})_{2,1}+2i \alpha_{3,2}e^{\theta_{2}-\theta_{2}^{*}}(M^{-1})_{2,2},  \\
   u_3(x,t)&= 2i \alpha_{5,1}e^{\theta_{1}-\theta_{1}^{*}}(M^{-1})_{1,1}+2i \alpha_{5,1}e^{\theta_{1}-\theta_{2}^{*}}(M^{-1})_{1,2} \\
       & +2i \alpha_{5,2}e^{\theta_{2}-\theta_{1}^{*}}(M^{-1})_{2,1}+2i \alpha_{5,2}e^{\theta_{2}-\theta_{2}^{*}}(M^{-1})_{2,2}.  \\
 \end{aligned}
\right.
 \end{equation}

\noindent
{\rotatebox{0}{\includegraphics[width=3.6cm,height=3.0cm,angle=0]{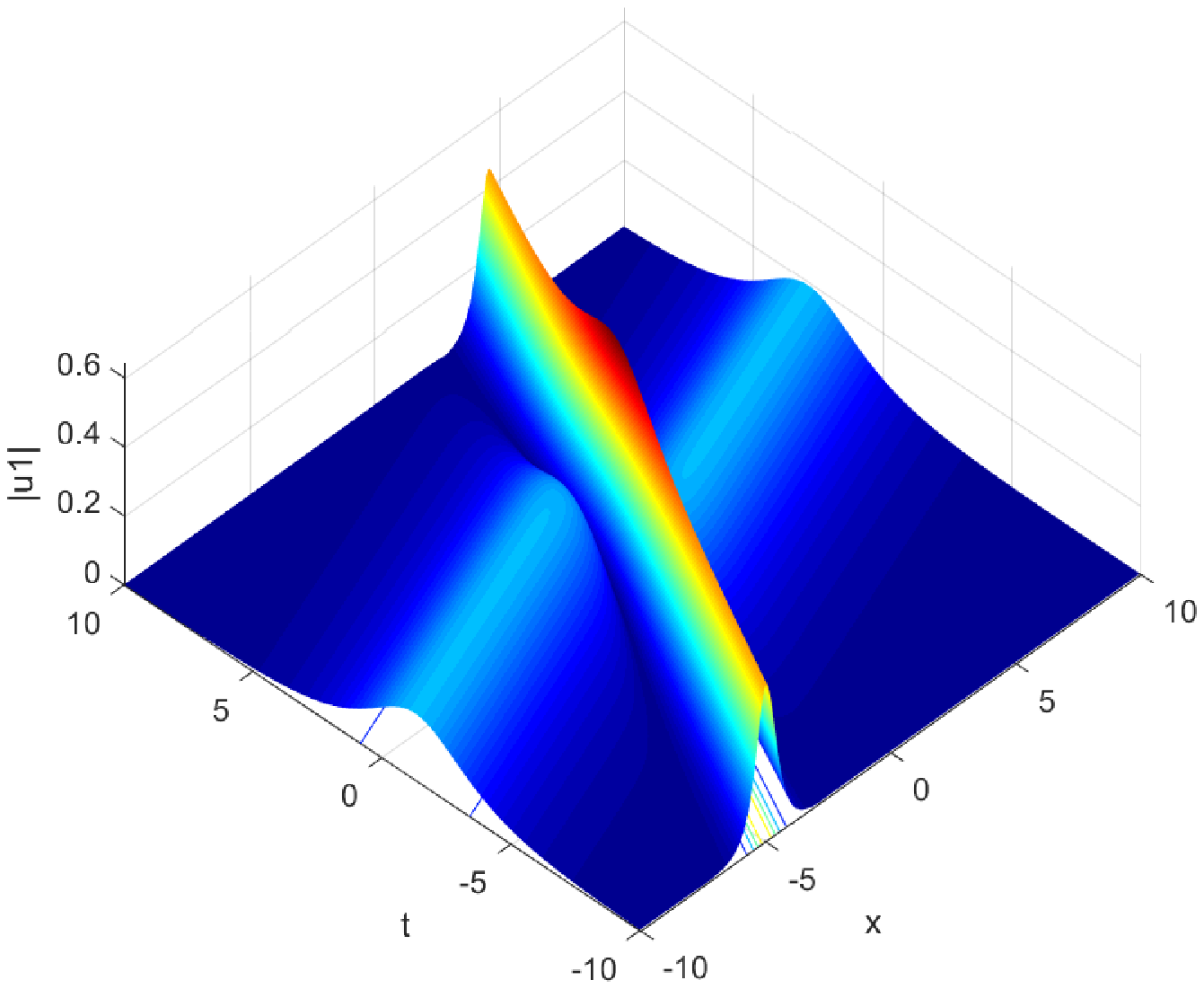}}}
~~~~
{\rotatebox{0}{\includegraphics[width=3.6cm,height=3.0cm,angle=0]{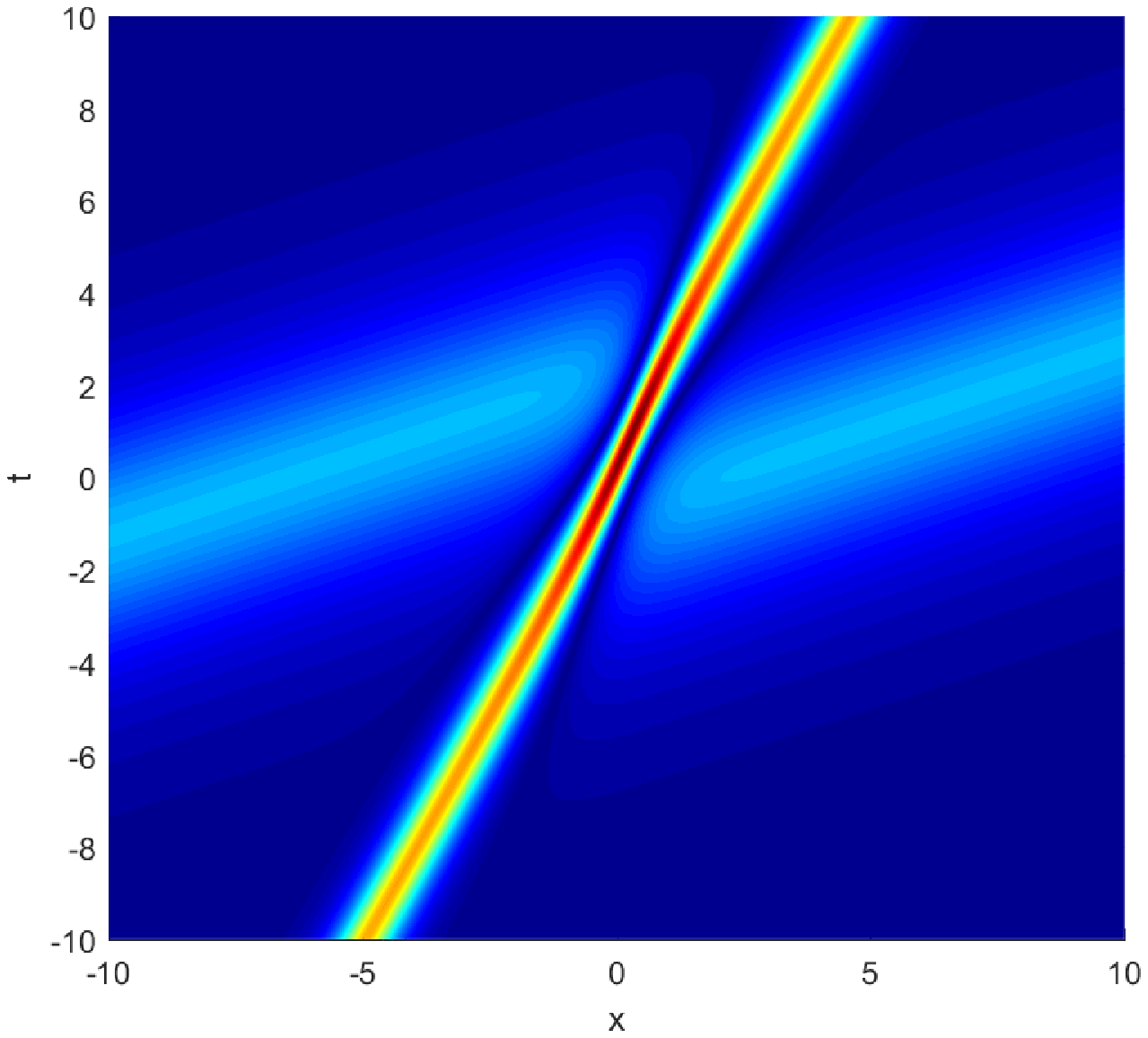}}}
~~~~
{\rotatebox{0}{\includegraphics[width=3.6cm,height=3.0cm,angle=0]{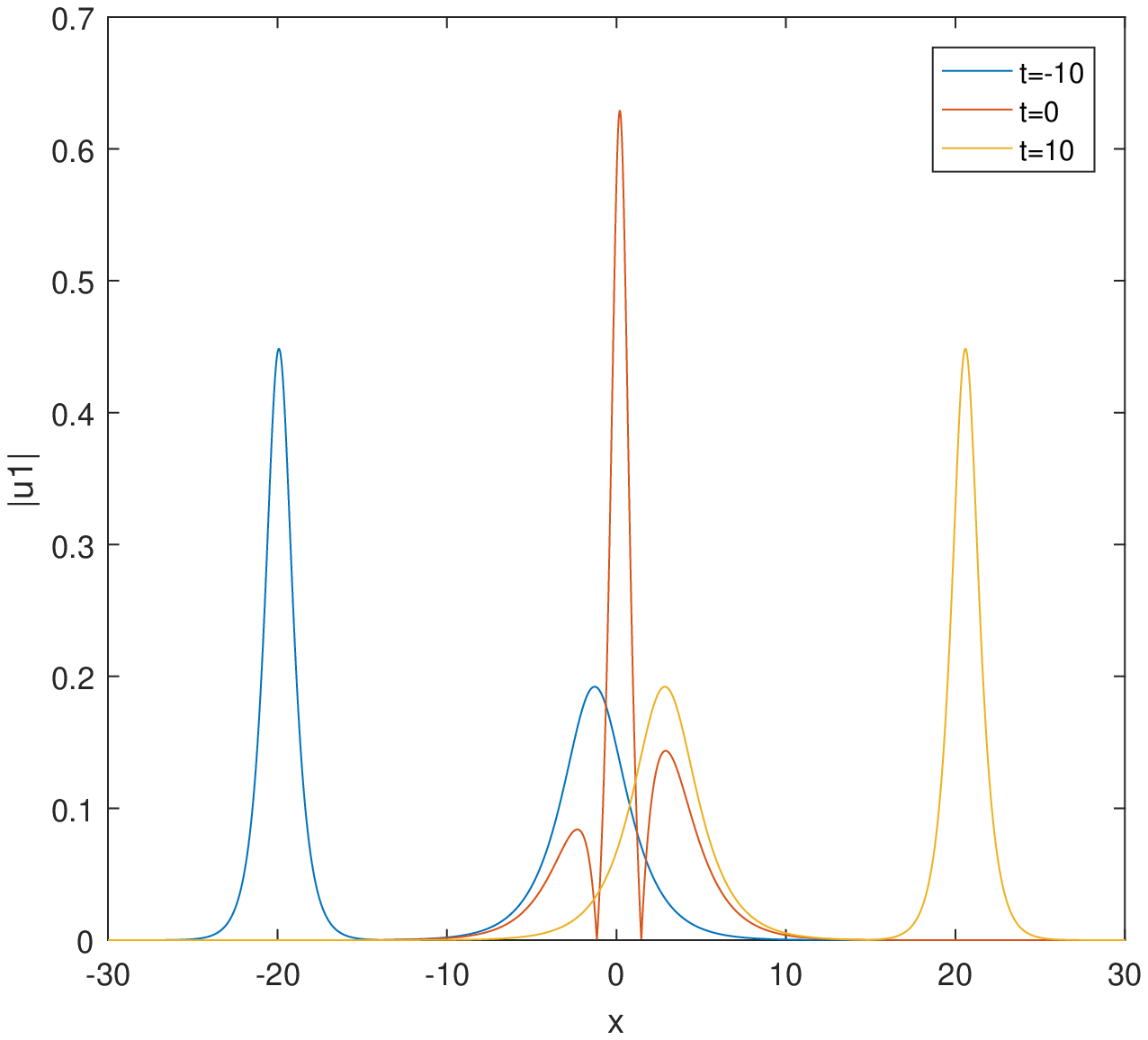}}}

$\ \qquad~~~~~~(\textbf{a})\qquad \ \qquad\qquad\qquad\qquad~(\textbf{b})
\ \qquad\qquad\qquad\qquad\qquad~(\textbf{c})$\\
\noindent
{\rotatebox{0}{\includegraphics[width=3.6cm,height=3.0cm,angle=0]{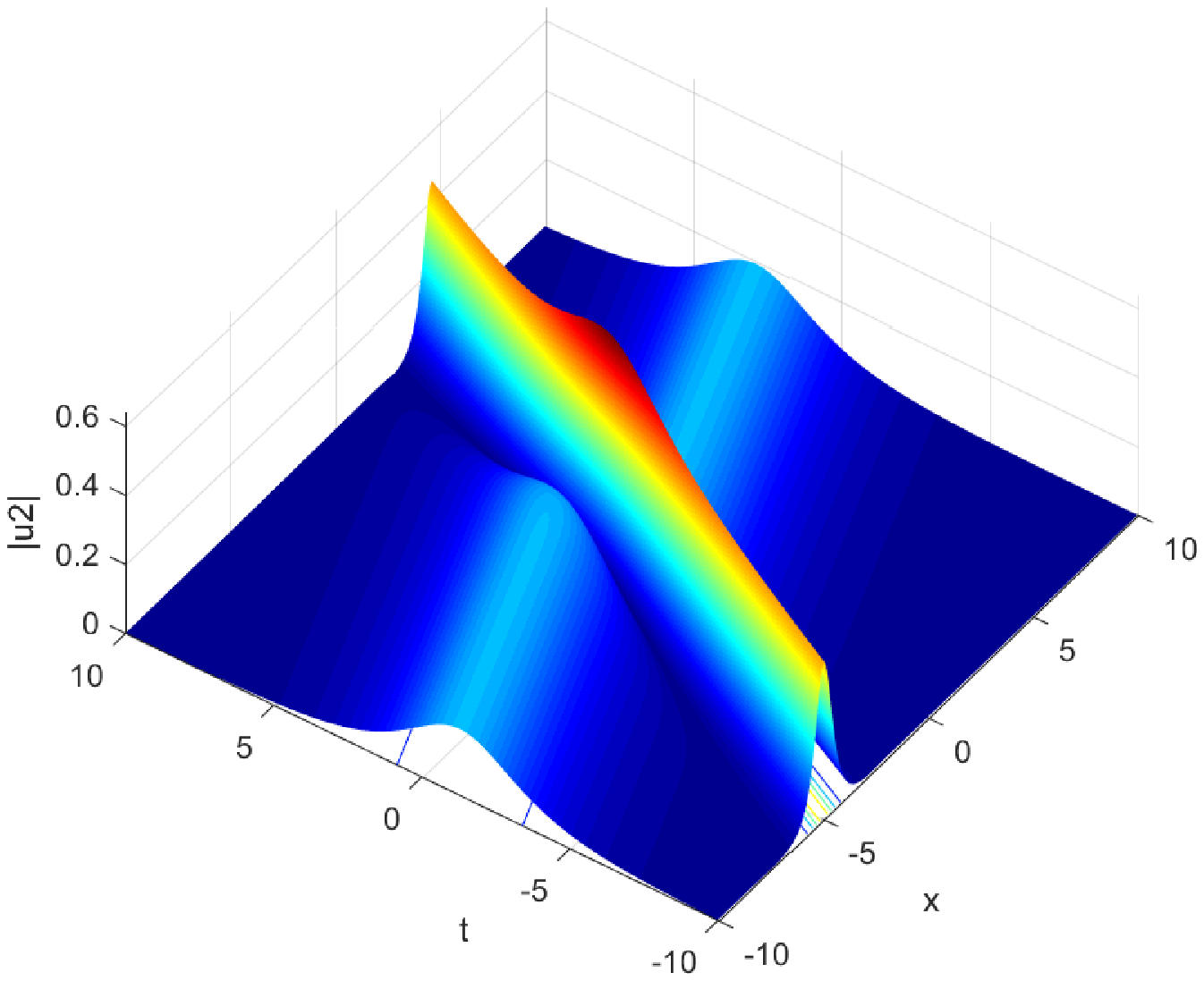}}}
~~~~
{\rotatebox{0}{\includegraphics[width=3.6cm,height=3.0cm,angle=0]{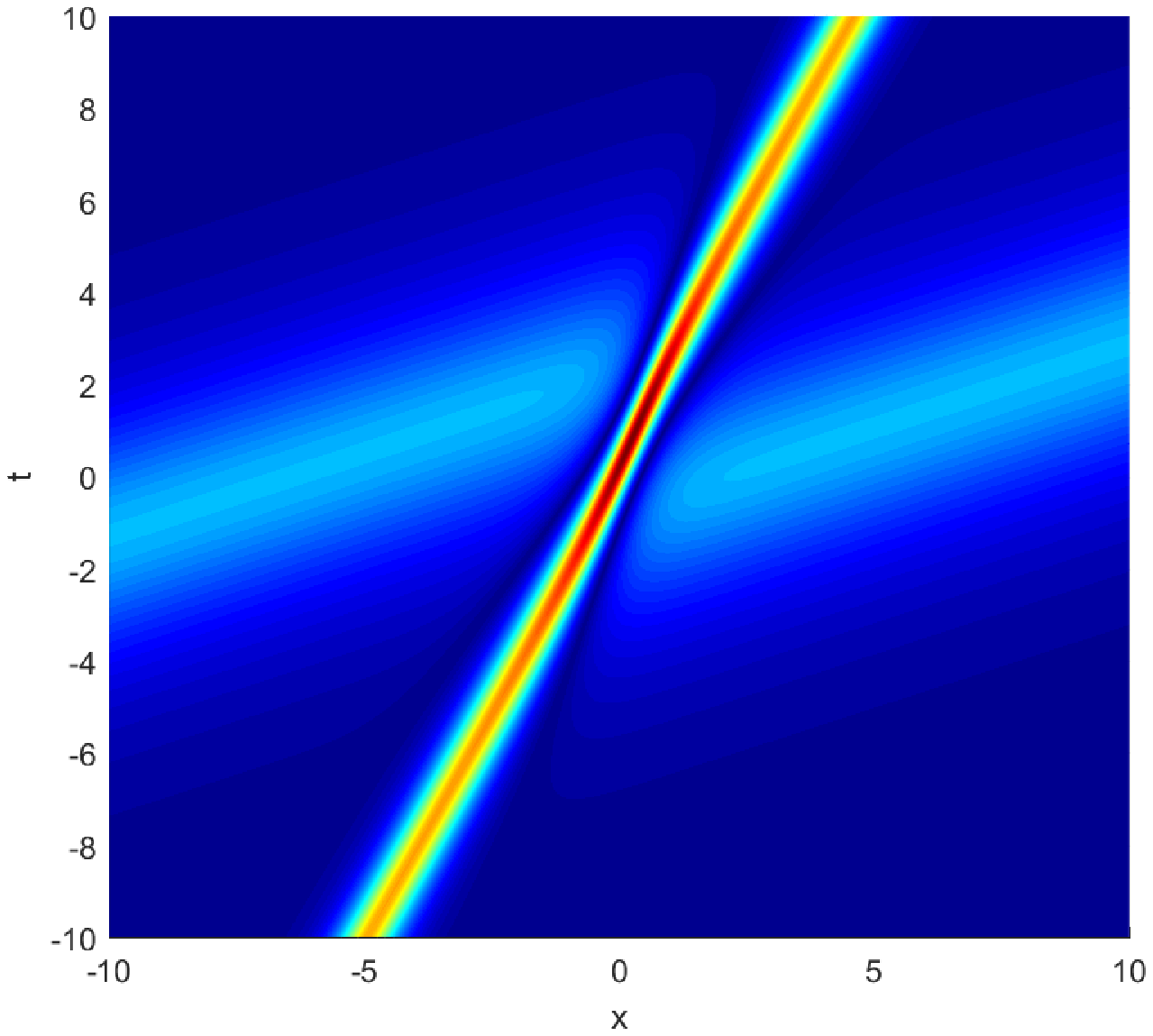}}}
~~~~
{\rotatebox{0}{\includegraphics[width=3.6cm,height=3.0cm,angle=0]{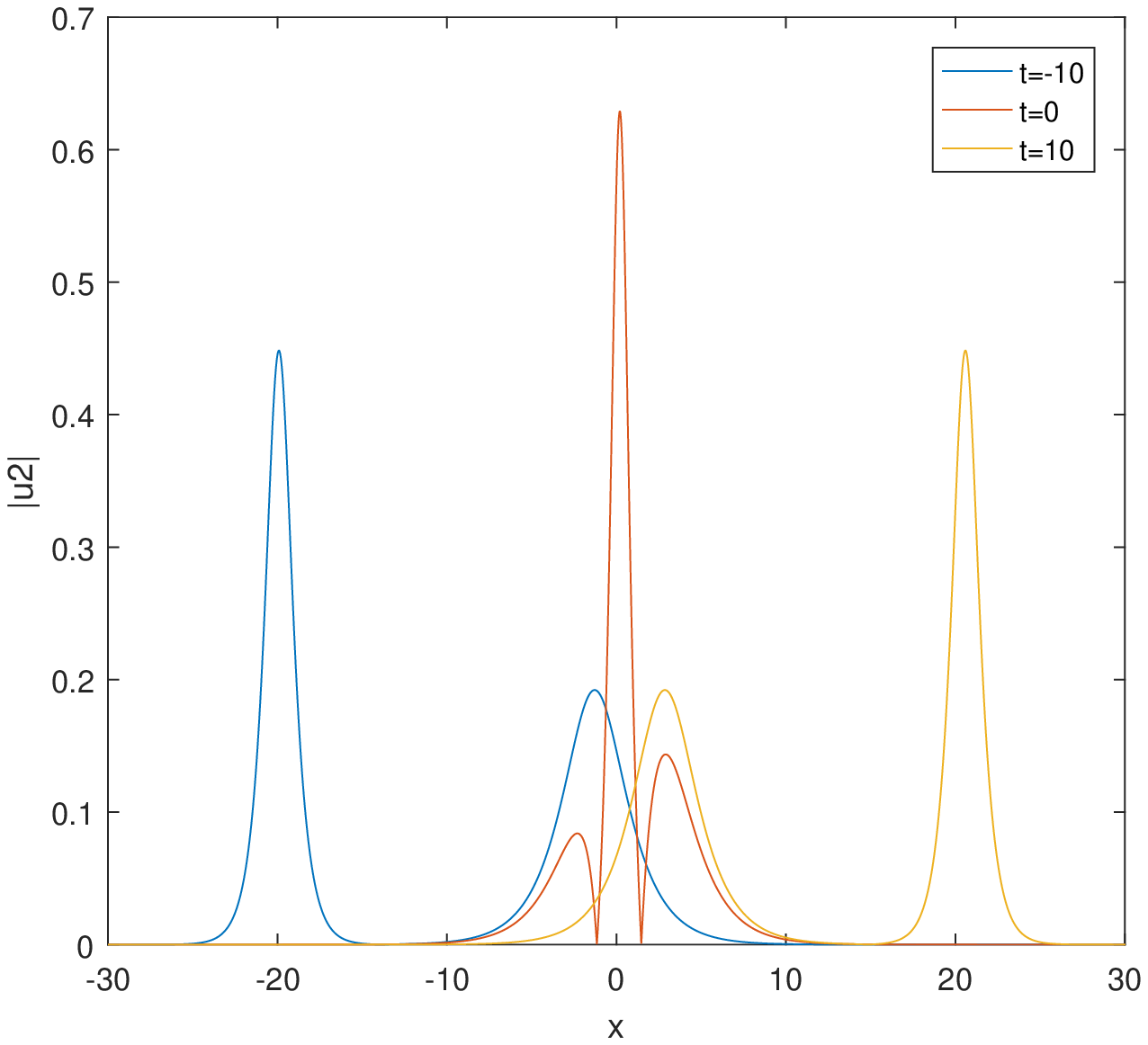}}}

$\ \qquad~~~~~~(\textbf{d})\qquad \ \qquad\qquad\qquad\qquad~(\textbf{e})
\ \qquad\qquad\qquad\qquad\qquad~(\textbf{f})$\\
\noindent
{\rotatebox{0}{\includegraphics[width=3.6cm,height=3.0cm,angle=0]{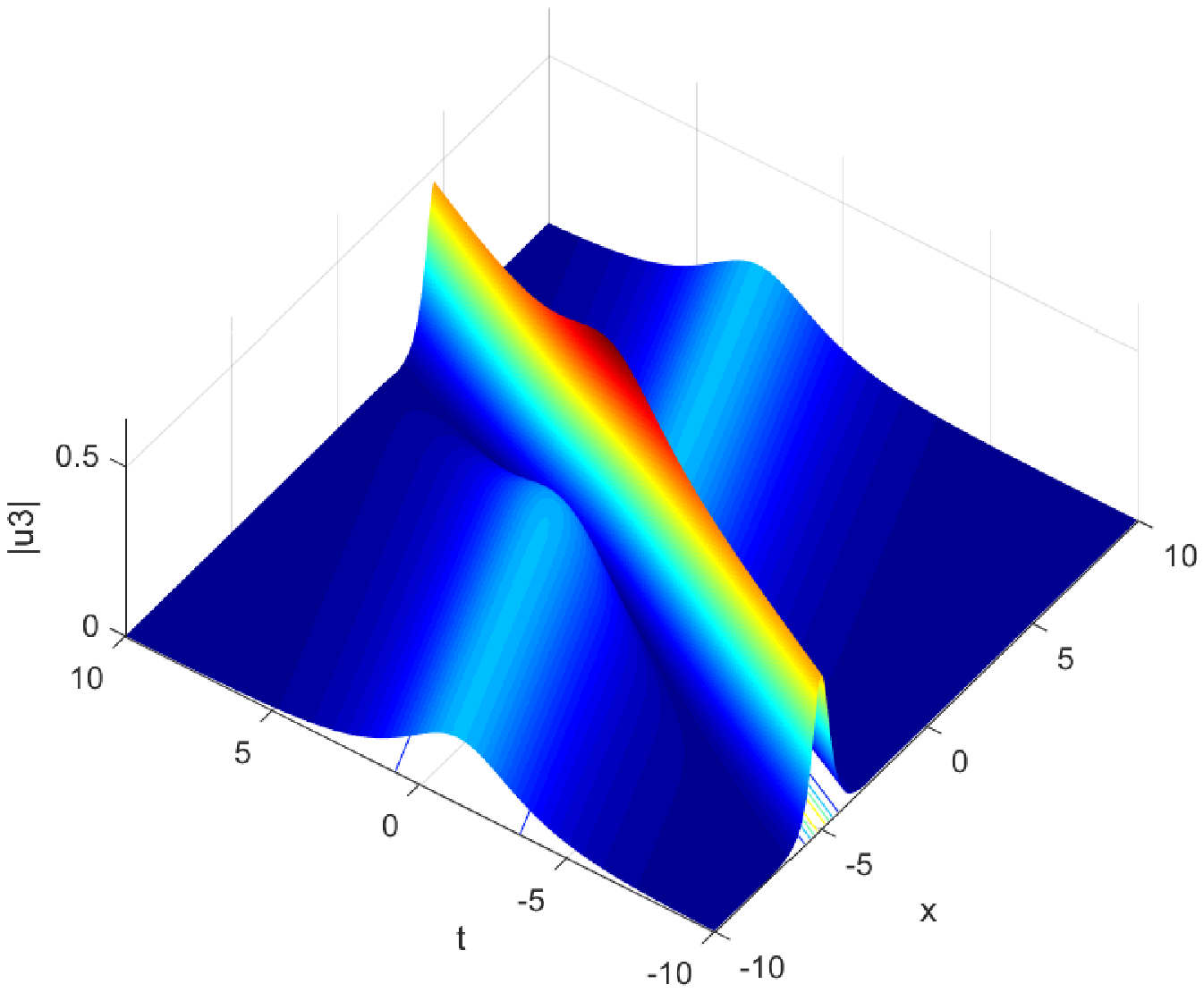}}}
~~~~
{\rotatebox{0}{\includegraphics[width=3.6cm,height=3.0cm,angle=0]{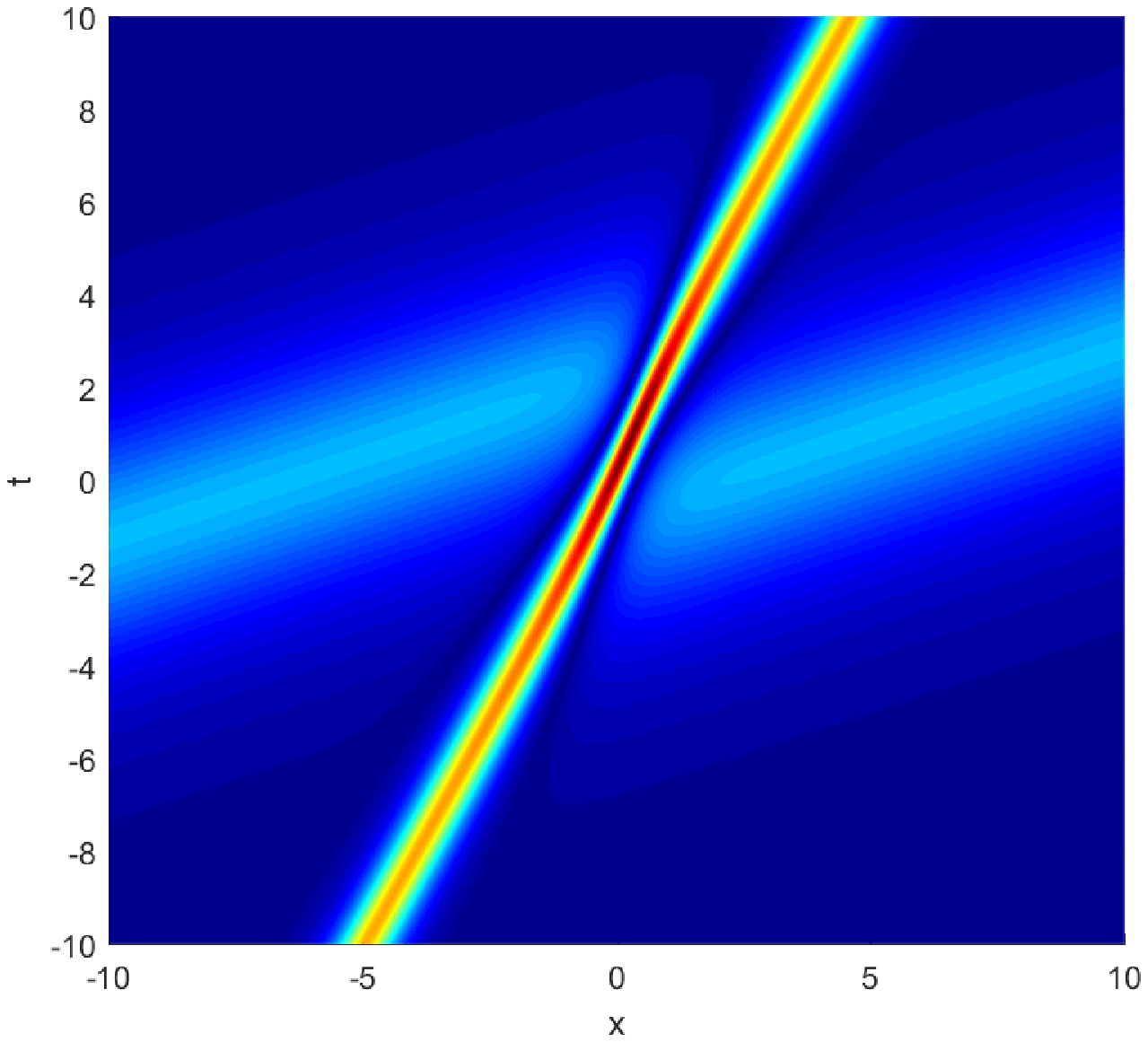}}}
~~~~
{\rotatebox{0}{\includegraphics[width=3.6cm,height=3.0cm,angle=0]{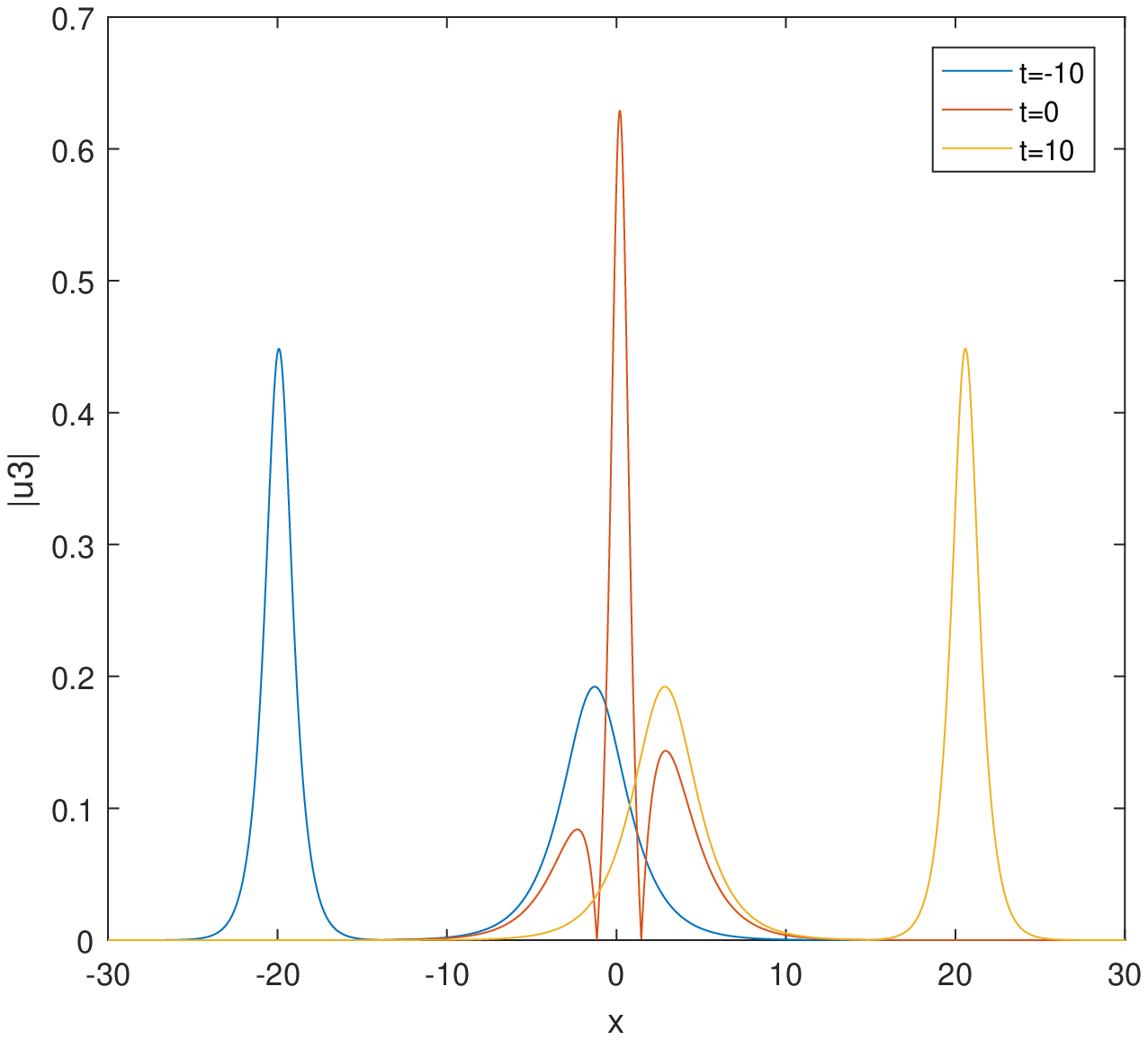}}}

$\ \qquad~~~~~~(\textbf{g})\qquad \ \qquad\qquad\qquad\qquad~(\textbf{h})
\ \qquad\qquad\qquad\qquad\qquad~(\textbf{i})$\\
\noindent { \small \textbf{Figure 9.} (Color online) Two-bell solutions  to  Eq. \eqref{sf4.3} with the parameters  $\alpha_{1,1}=\alpha_{1,2}= \alpha_{3,1}= \alpha_{3,2}= \alpha_{5,1}= \alpha_{5,2}=0.2$, $\alpha_{2,1}=\alpha_{2,2}= \alpha_{4,1}= \alpha_{4,2}= \alpha_{6,1}= \alpha_{6,2}=0.3$, $b_1=0.3$, $b_2=0.7$.
\label{fig4.31}
$\textbf{(a)(d)(g)}$: the structures of the two-bell solutions,
$\textbf{(b)(e)(h)}$: the density plot,
$\textbf{(c)(f)(i)}$: the wave propagation of the two-bell solutions.} \\

In addition, the  localized structures and  dynamic propagation behaviors of two-bell soliton solutions are displayed in Fig. 9, from which  we can see that after two solitons collide with each other, their propagation direction has been misaligned, but the energies of them are  basically  unchanged.

\subsection{Summary}
    From these special cases, we can see that the  exact  solutions of  mmKdV equation we obtain are suitable for matrices of any order  whether $ p = q$, $p> q$ or $ p <q$. More importantly, in the process of  solving specific equations, the special properties of the potential matrices  will help us further classify soliton solutions to find other  meaningful solutions such as  breather-type solutions and bell-type solutions,
which can help us understand the meaning of  solutions  more comprehensively.

\section{Conclusions}
     As we all know, many scholars have done a lot of works on the   mKdV-type  equation or equations, for example,  Zhang et al. have studied two-component cmKdV equations by Darboux transformation \cite{zhang2008lax},   and   Liu et al.  have  discussed initial--boundary problems for the vector modified Korteweg--de Vries(vmKdV) equation via Fokas unified transform method \cite{liu2016initial}. However,  many scholars pay less attention to mmKdV  equation with a $p\times q$ complex-valued  potential matrix function in recent years. Actually, vmKdV equations, cmKdV equations and so on are only the sub cases of the  mmKdV equation. In this work, we study the exact solutions with their propagation behaviors    for  the mmKdV equation which can solve  the solutions of not only previous vmKdV and cmKdV equations, but also other generalized mKdV-type equation  or equations.  The main propose of this work is to study the multi-soliton solutions of mmKdV equation with the aid of RH approach. Firstly, we performed the spectral analysis of Lax pair in order to structure the RH problem. Secondly, we obtain general solution form of  mmKdV equation. Next, we choose certain special cases of potential matrix $Q$ including $2\times 2$, $4\times 4$, $1\times 3$ and $6\times 1$ form.  By analyzing  the special characteristics of potential  matrices,  some interested solutions can be derived on the basis of original exact solutions. At the same time, the localized structures and  dynamic  behaviors of above interesting solutions are displayed vividly.

\section*{Acknowledgements}

This work was supported by    the Postgraduate Research and Practice of Educational Reform for Graduate students in CUMT under Grant No. 2019YJSJG046, the Natural Science Foundation of Jiangsu Province under Grant No. BK20181351, the Six Talent Peaks Project in Jiangsu Province under Grant No. JY-059, the Qinglan Project of Jiangsu Province of China, the National Natural Science Foundation of China under Grant No. 11975306, the Fundamental Research Fund for the Central Universities under the Grant Nos. 2019ZDPY07 and 2019QNA35, and the General Financial Grant from the China Postdoctoral Science Foundation under Grant Nos. 2015M570498 and 2017T100413.



\begin{thebibliography}{99}
\bibitem{hirota1980direct}
R.~Hirota, {Direct methods in soliton theory}, in: Solitons, Springer, 1980,
  pp. 157--176.

\bibitem{matveev1979darboux}
V.~Matveev, {Darboux transformation and explicit solutions of the
  Kadomtcev-Petviaschvily equation, depending on functional parameters}, Lett.
  Math. Phys. 3~(3) (1979) 213--216.

\bibitem{ablowitz1981solitons}
M.~J. Ablowitz, H.~Segur, {Solitons and the inverse scattering transform},
  Vol.~4, Siam, 1981.

\bibitem{beals1984scattering}
R.~Beals, R.~R. Coifman, {Scattering and inverse scattering for first order
  systems}, Commun. Pur. Appl. Math. 37~(1) (1984) 39--90.

\bibitem{beals1988direct}
R.~Beals, P.~Deift, C.~Tomei, {Direct and inverse scattering on the line},
  no.~28, American Mathematical Soc., 1988.
\bibitem{fokas2012unified}
A.~Fokas, J.~Lenells, The unified method: I. nonlinearizable problems on the
  half-line, J. Phys. A: Math. Theor. 45~(19) (2012) 195201.
\bibitem{lenells2012initial}
J.~Lenells, Initial-boundary value problems for integrable evolution equations
  with 3$\times$ 3 lax pairs, Phys. D 241~(8) (2012) 857--875.
\bibitem{guo2012riemann}
B.~Guo, L.~Ling, {Riemann-Hilbert approach and {N}-soliton formula for coupled
  derivative Schr{\"o}dinger equation}, J. Math. Phys. 53~(7) (2012) 073506.
\bibitem{de2013riemann}
A.~B. de~Monvel, D.~Shepelsky, {A Riemann--Hilbert approach for the
  Degasperis--Procesi equation}, Nonlinearity 26~(7) (2013) 2081.
\bibitem{geng2016riemann}
X.~Geng, J.~Wu, Riemann--hilbert approach and n-soliton solutions for a
  generalized sasa--satsuma equation, Wave Motion 60 (2016) 62--72.
\bibitem{yan2017initial}
Z.~Yan, {An initial-boundary value problem for the integrable spin-1
  Gross-Pitaevskii equations with a 4$\times$ 4 Lax pair on the half-line},
  Chaos 27~(5) (2017) 053117.
\bibitem{ma2018riemann}
W.-X. Ma, Riemann--hilbert problems and n-soliton solutions for a coupled mkdv
 system, J. Geom. Phys. 132 (2018) 45--54.
\bibitem{ma2019inverse}
W.-X. Ma, The inverse scattering transform and soliton solutions of a combined
  modified korteweg--de vries equation, J. Math. Anal. Appl. 471~(1-2) (2019)
  796--811.
\bibitem{wang2010integrable}
D.-S. Wang, D.-J. Zhang, J.~Yang, Integrable properties of the general coupled
  nonlinear schr{\"o}dinger equations, J. Math. Phys. 51~(2) (2010) 023510.
\bibitem{zhang2017riemann}
Y.~Zhang, Y.~Cheng, J.~He, Riemann-hilbert method and n-soliton for
  two-component gerdjikov-ivanov equation, J. Nonlinear Math. Phys. 24~(2)
  (2017) 210--223.
\bibitem{tian2017initial}
S.-F. Tian, {Initial-boundary value problems of the coupled modified
  Korteweg--de Vries equation on the half-line via the Fokas method}, J. Phys.
  A: Math. Theor. 50~(39) (2017) 395204.
\bibitem{tian2018initial}
S.-F. Tian, Initial-boundary value problems for the coupled modified
  korteweg-de vries equation on the interval., Commun. Pure \& Appl. Anal.
  17~(3).
\bibitem{tian2016mixed}
S.-F. Tian, The mixed coupled nonlinear schr{\"o}dinger equation on the
  half-line via the fokas method, Proc. R. Soc. A 472~(2195) (2016) 20160588.
\bibitem{xia2018initial}
B.~Xia, A.~Fokas, Initial--boundary value problems associated with the
  ablowitz--ladik system, Phys. D 364 (2018) 27--61.
\bibitem{peng2019riemann}
W.-Q. Peng, S.-F. Tian, X.-B. Wang, T.-T. Zhang, Y.~Fang, {Riemann--Hilbert
  method and multi-soliton solutions for three-component coupled nonlinear
  Schr{\"o}dinger equations}, J. Geom. Phys. 146 (2019) 103508.
\bibitem{yang2019n}
J.-J. Yang, S.-F. Tian, W.-Q. Peng, T.-T. Zhang, The n-coupled higher-order
  nonlinear schr{\"o}dinger equation: Riemann-hilbert problem and multi-soliton
  solutions, Math. Meth. Appl. Sci. (2019) https://doi.org/10.1002/mma.6055.
\bibitem{ablowitz2018inverse}
M.~J. Ablowitz, X.-D. Luo, Z.~H. Musslimani, Inverse scattering transform for
  the nonlocal nonlinear schr{\"o}dinger equation with nonzero boundary
  conditions, J. Math. Phys. 59~(1) (2018) 011501.
\bibitem{vekslerchik1992discrete}
V.~E. Vekslerchik, V.~V. Konotop, {Discrete nonlinear Schrodinger equation
  under nonvanishing boundary conditions}, Inverse Probl. 8~(6) (1992) 889.
\bibitem{biondini2014inverse}
G.~Biondini, G.~Kova{\v{c}}i{\v{c}}, {Inverse scattering transform for the
  focusing nonlinear Schr{\"o}dinger equation with nonzero boundary
  conditions}, J. Math. Phys. 55~(3) (2014) 031506.
\bibitem{prinari2015inverse}
B.~Prinari, F.~Vitale, {Inverse scattering transform for the focusing nonlinear
  Schr{\"o}dinger equation with one-sided nonzero boundary condition}, Cont.
  Math 651 (2015) 157--194.
\bibitem{yang2019riemann}
J.-J. Yang, S.-F. Tian, {Riemann-Hilbert problem for the modified
  Landau-Lifshitz equation with nonzero boundary conditions}, arXiv preprint
  arXiv:1909.11263.
\bibitem{deift1992steepest}
P.~Deift, X.~Zhou, A steepest descent method for oscillatory riemann-hilbert
  problems, Bulletin Amer. Math. Soc. 26~(1) (1992) 119--123.
\bibitem{xu2015long}
J.~Xu, E.~Fan, {Long-time asymptotics for the Fokas--Lenells equation with
  decaying initial value problem: without solitons}, J. Differ. Equations
  259~(3) (2015) 1098--1148.
\bibitem{tian2018long}
S.-F. Tian, T.-T. Zhang, {Long-time asymptotic behavior for the
  Gerdjikov-Ivanov type of derivative nonlinear Schr{\"o}dinger equation with
  time-periodic boundary condition}, Proc. Am. Math. Soc. 146~(4) (2018)
  1713--1729.
\bibitem{wang2019long}
D.-S. Wang, B.~Guo, X.~Wang, Long-time asymptotics of the focusing
  kundu--eckhaus equation with nonzero boundary conditions, J. Differ.
  Equations 266~(9) (2019) 5209--5253.
\bibitem{liu2019long}
N.~Liu, B.~Guo, Long-time asymptotics for the sasa--satsuma equation via
  nonlinear steepest descent method, J. Math. Phys. 60~(1) (2019) 011504.
\bibitem{deift1999orthogonal}
P.~Deift, {Orthogonal polynomials and random matrices: a Riemann-Hilbert
  approach}, Vol.~3, American Mathematical Soc., 1999.
\bibitem{tsuchida1998coupled}
T.~Tsuchida, M.~Wadati, {The coupled modified Korteweg-de Vries equations}, J.
  Phys. Soc. Jpn. 67~(4) (1998) 1175--1187.
\bibitem{athorne1987generalised}
C.~Athorne, A.~Fordy, {Generalised KdV and MKdV equations associated with
  symmetric spaces}, J. Phys. A: Math. Gen. 20~(6) (1987) 1377.
\bibitem{zhang2008lax}
H.-Q. Zhang, B.~Tian, T.~Xu, H.~Li, C.~Zhang, H.~Zhang, {Lax pair and Darboux
  transformation for multi-component modified Korteweg--de Vries equations}, J.
  Phys. A: Math. Theor. 41~(35) (2008) 355210.
\bibitem{liu2016initial}
H.~Liu, X.~Geng, {Initial--boundary problems for the vector modified
  Korteweg--de Vries equation via Fokas unified transform method}, J. Math.
  Anal. Appl. 440~(2) (2016) 578--596.
\end{thebibliography}
\end{document}